\documentclass[a4paper,11pt]{article}

\usepackage{booktabs} 
\usepackage[ruled]{algorithm2e} 
\usepackage{anysize,amsmath,amsfonts,amssymb,amsthm,dsfont,graphics,graphicx,enumerate,multirow}
\usepackage{soul}
\usepackage{xcolor}
\usepackage{comment}
\usepackage{hyperref}
\usepackage{ifthen}

\graphicspath{{figures/}} 

\newtheorem{principle}{Principle}
\newtheorem{principle*}{Principle}
\newtheorem{alg}{Algorithm}
\newtheorem{thm}{Theorem}

\newtheorem{lem}{Lemma}

\newtheorem{prop}{Proposition}

\newtheorem{defn}{Definition}
\newtheorem{conj}{Conjecture}
\newtheorem{rem}{Remark}

\DeclareMathOperator*{\defeq}{\triangleq}

\newcommand{\rank}{{\textup{Rank}}}

\newcommand{\ind}{{\mathbb{I}}}
\newcommand{\Ex}{{\mathbb{E}}}
\newcommand{\prob}{{\mathbb{P}}}
\newcommand{\eps}{{\epsilon}}
\newcommand{\Geo}{{\rm Geometric}}

\newcommand{\pop}{{\nu}}
\newcommand{\Rmen}{ R_{\textup{\tiny MEN}}}
\newcommand{\Rwomen}{ R_{\textup{\tiny WOMEN}}}
\newcommand{\Rshort}{ R_{\textup{\tiny SHORT}}}
\newcommand{\Rlong}{ R_{\textup{\tiny LONG}}}
\newcommand{\pmen}{ p_{\textup{\tiny MEN}}}
\newcommand{\pwomen}{ p_{\textup{\tiny WOMEN}}}

\newcommand{\cP}{{\mathcal{P}}}

\newcommand{\cC}{{\mathcal{C}}}

\newcommand{\cM}{{\mathcal{M}}}
\newcommand{\cW}{{\mathcal{W}}}
\newcommand{\hcM}{{\hat{\mathcal{M}}}}
\newcommand{\cMf}{{\mathcal{M}_{{\tiny \textup{fake}}}}}

\newcommand{\Mrank}{ R_{\textup{\tiny MEN}}(\textup{MOSM}) }
\newcommand{\Wrank}{ R_{\textup{\tiny WOMEN}}(\textup{MOSM}) }

\newcommand*{\dotleq}{\mathrel{\dot{\leq}}}
\newcommand*{\dotgeq}{\mathrel{\dot{\geq}}}

\newcommand{\Ev}{\mathcal{E}}

\newboolean{NYC}
\setboolean{NYC}{false}	
\newcommand{\NYC}[2]{\ifthenelse{\boolean{NYC}}{#1}{#2}}		

\usepackage[shortlabels,inline]{enumitem}

\newlist{compactitem}{itemize}{3}
\setlist[compactitem]{topsep=0pt,partopsep=0pt,itemsep=0pt,parsep=0pt}
\setlist[compactitem,1]{label=\textbullet}
\setlist[compactitem,2]{label=---}
\setlist[compactitem,3]{label=*}

\newlist{compactdesc}{description}{3}
\setlist[compactdesc]{topsep=0pt,partopsep=0pt,itemsep=0pt,parsep=0pt}

\newlist{compactenum}{enumerate}{3}
\setlist[compactenum]{topsep=0pt,partopsep=0pt,itemsep=0pt,parsep=0pt}
\setlist[compactenum,1]{label=\arabic*}
\setlist[compactenum,2]{label=\alph*}
\setlist[compactenum,3]{label=\roman*}

\usepackage{natbib}
 \bibpunct[, ]{(}{)}{,}{a}{}{,}%
 \def\bibfont{\small}%
\usepackage{setspace}
\newcommand{\dashrule}[1][black]{%
  \color{#1}\rule[\dimexpr.5ex-.2pt]{4pt}{.4pt}\xleaders\hbox{\rule{4pt}{0pt}\rule[\dimexpr.5ex-.2pt]{4pt}{.4pt}}\hfill\kern0pt%
}

\makeatletter
\renewcommand\paragraph{\@startsection{paragraph}{4}{\z@}%
{1.15ex \@plus.3ex \@minus.2ex}%
	{-0.8em}%
	{\normalfont\normalsize\bfseries}}
\makeatother

\allowdisplaybreaks

\begin{document}

\setlength{\abovedisplayskip}{6pt}
\setlength{\belowdisplayskip}{6pt}
\setlength{\abovedisplayshortskip}{0pt}
\setlength{\belowdisplayshortskip}{6pt}

\title{
The Competition for Partners in Matching Markets
\thanks{Y. Kanoria and P. Qian gratefully acknowledge the support of the National Science Foundation’s Division of Civil, Mechanical, and Manufacturing Innovation (Grant CMMI-1201045).}
}
%
\author{Yash Kanoria\thanks{Graduate School of Business, Columbia University, Email: \texttt{ykanoria@columbia.edu}} \and Seungki Min\thanks{Department of Industrial and Systems Engineering, KAIST, Email: \texttt{skmin@kaist.ac.kr}} \and Pengyu Qian\thanks{Krannert School of Management, Purdue University, Email: \texttt{qianp@purdue.edu}}.}
%
\date{}

\maketitle
%


\begin{abstract}
We study the competition for partners in two-sided matching markets with heterogeneous agent preferences, with a focus on how the equilibrium outcomes depend on the connectivity in the market.
We model random partially connected markets, with each agent having an average degree $d$ in a random (undirected) graph, and a uniformly random preference ranking over their neighbors in the graph. We formally characterize stable matchings in large markets random with small imbalance and find a threshold in the connectivity $d$ at $\log^2 n$ (where $n$ is the number of agents on one side of the market) which separates a  ``weak competition'' regime, where agents on both sides of the market do equally well, from a ``strong competition'' regime, where agents on the short (long) side of the market enjoy a significant advantage (disadvantage). 
Numerical simulations confirm and sharpen our theoretical predictions, and demonstrate robustness to our assumptions.
We leverage our characterizations in two ways: First, we derive prescriptive insights into how to \emph{design} the connectivity of the market to trade off optimally between the average agent welfare achieved and the number of agents who remain unmatched in the market. For most market primitives, we find that the optimal connectivity should lie in the weak competition regime or at the threshold between the regimes. Second, our analysis uncovers a new conceptual principle governing whether the short side enjoys a significant advantage in a given matching market, which can moreover be applied as a diagnostic tool given only basic summary statistics for the market. Counterfactual analyses using data on centralized high school admissions in a major USA city show the practical value of both our design insights and our diagnostic principle.
\end{abstract}

\noindent{\bf Keywords:} matching markets, stable matching, competition, market design, diagnosis from summary statistics.


\section{Introduction} \label{sec:intro}
In recent years, the use of matching platforms for various purposes such as dating, the labor market, and school and college admissions, has experienced rapid growth. Many of platforms feature agents with heterogeneous preferences for potential partners that include a significant \emph{idiosyncratic} ``beauty lies in the eye of the beholder'' component, which varies significantly from one agent to the next. While the equilibrium notion of stable matchings is known to capture well the outcomes which arise in two-sided matching platforms \citep[see, e.g.,][]{hitsch2010matching}, our understanding of the nature of these market equilibria which arise as a function of market characteristics remains partial at best. In turn, improving our understanding of market equilibria is essential to throw light on how a platform can be designed so as to maximize market performance. 

A stable matching in a two-sided matching market with ordinal preferences on both sides is one in which there is no blocking pair, namely, a pair of agents who would prefer to be matched to each other over their outcome in the present matching. 
The structure of stable matchings which arises in a given two-sided matching market is determined by the two-sided competition for partners. To improve our understanding of stable matchings which arise in two-sided markets with idiosyncratic preferences, and to derive operational insights, the field has found it useful to investigate variants of the \emph{random matching markets} model introduced by \cite{knuth1976mariages}, where agents have independent, complete and uniformly random preference lists over the other side of the market.
One key finding from theoretical studies of random matching markets and real-world evidence is that in most two-sided matching markets, the stable matching is nearly unique \citep{AKL17,immorlica2005marriage,RothPeranson99,kojima2009incentives}.\footnote{In matching contexts with specific structure, such as college admissions with financial aid \citep{romm-large-core-ec2022}, and markets with a ``small-world network'' preference structure \citep{rheingans2020large}, a large set of stable matchings has been found to be typical.}

What is the structure of this (nearly unique) stable matching? 
Previous research has characterized stable matchings only for random matching markets at each of two extremes. (i) One line of papers studies markets which are fully connected, meaning that each agent's preference list includes all agents on the other side of the market \citep[e.g.,][]{knuth1976mariages,pittel1992likely,AKL17}. \cite{AKL17} finds that ``\emph{matching markets are extremely competitive, with even the slightest imbalance greatly benefiting the short side}''. 
Specifically, in a matching market with $n$ men and $n+1$ women, and uniformly randomly complete preference lists, independent across agents, there is a nearly unique stable matching, where the average rank of men for their wives is just $\log n (1+o(1))$, 
whereas the average rank of women for their husbands is $\frac{n}{\log n}(1+o(1))$. For example, with $n=1,000$, men get matched to their seventh most desired woman, whereas women are matched to only their 145th most preferred man.\footnote{Of course the situation is completely reversed if, instead, there are $999$ women, while the number of men is still $1,000$.}(ii) Another line of papers studies random markets which are very sparsely connected, meaning that agents have preference lists of constant length regardless of market size \cite[e.g.,][]{immorlica2005marriage, arnosti2022lottery,ashlagi2020matters}, and finds that such markets exhibit only weak competition, namely, small changes in market composition only result in small changes in the resulting equilibrium outcome. Notably, the nature of competition and the structure of market equilibria remains largely unknown for markets with intermediate connectivity larger than a constant, despite the prevalence of real-world  markets which fall in this regime.\footnote{
One example is centralized college admissions in China, where most provinces allow candidates to list up to 40 institutions and use a hybrid mechanism with resemblance to DA \citep{ChenJiangKesten2020}, and, anecdotally, most candidates rank between 10 and 40 of $\sim 1000$ institutions (with $n \sim 300,000$ candidates per province, $\log n \approx 13$ so list lengths appear to be larger than a ``constant'' but smaller than $\log^2 n$).}

The present paper aims to develop a complete understanding of competition in partially connected random matching markets as a function of market connectivity, and to derive consequent guidance on how to \emph{design} the connectivity of a matching market so as to maximize market performance. Our contribution is three-fold: we introduce a model of partially-connected random matching markets and explicitly characterize the equilibrium arising in these markets (by introducing new technical tools), we use our characterizations to derive guidance on the design of simple platform interventions which control the market connectivity, and we discover a new conceptual principle governing whether being on the short side confers a significant advantage in a given matching market, which can moreover be used as a diagnostic tool given only summary statistics for the market.
We next discuss each of these contributions in turn.

\paragraph{Model.} Our model generalizes the random matching market model to allow ``partially connected'' markets with each agent having an average degree $d$ in a random (undirected) connectivity graph. Each agent has a preference ranking over only their neighbors in the connectivity graph. We assume there are $n+k$ men and $n$ women, where the ``imbalance'' $k$ may be positive or negative but we restrict to ``small'' imbalances $|k| = o(n)$ for our theoretical analysis. For technical convenience, the random graph model we work with is one where each man is connected to a uniformly random subset of exactly $d$ women, independent of other men.\footnote{As a result, each woman has $\textup{Binomial}(n+k, d/n) \xrightarrow[n \to \infty]{\textup{d}} \textup{Poisson}(d)$ neighbors where $\xrightarrow{\textup{d}}$ denotes convergence in distribution. Throughout the paper we will restrict attention to $d = \omega(1)$, as a result of which $\textup{Poisson}(d)/d \xrightarrow{\textup{d}} 1$, 
i.e., the degree of each woman is also very close to $d$, and so the asymmetry between the two sides in the model is mainly technical.}

\paragraph{Main theoretical findings.}
We characterize stable matchings as a function of $d$ and the number of women $n$ for $|k| = O(n^{1-\epsilon})$ (small) market imbalance, and find that the short side enjoys a significant advantage only for $d$ exceeding $\log^2 n$: For moderately connected markets, specifically any $d$ such that $d=o(\log^2 n)$ and $d = \omega(1)$ and large $n$, we find that \emph{there is only a weak effect of competition}, namely, the short and long sides of the market are almost equally well off, with agents on both sides getting a $\sqrt{d}(1+o(1))$-ranked partner on average.  
On the other hand, for densely connected markets, specifically for any $d = \omega(\log^2 n)$ and large $n$, we find that there is a strong effect of competition: assuming a small imbalance $|k| = o(n)$, the short side agents get a partner of rank $\log n$ on average, while the long side agents get a partner of (much larger) rank $d/\log n$ on average. A substantial technical challenge we overcome is the complexity in the way that Deferred Acceptance (DA) terminates when there is a positive (but vanishing) fraction of unmatched agents on both sides of the market;
see Section~\ref{sec:proof-sketch} for an overview of our formal analysis. 

Numerical simulations of our model confirm the theoretical predictions, and in fact further refine our findings, including by capturing the dependence on the imbalance $k$: they suggest a \emph{sharp threshold between the two regimes close to  $d \approx 1.0 \times \log^2 (n/|k|)$}, where $k$ is the market imbalance, and that \emph{this holds even for small $n$ down to $n \approx 10$}. We furthermore provide a heuristic detailed calculation which predicts the threshold being located at $1.0 \times \log^2 (n/|k|)$; the calculation also predicts that the ranks of partners follow a truncated Geometric distribution. Figure~\ref{fig:no-stark-effect} provides a schematic depicting our main findings. 

Note that the weak competition regime includes a wide range of well connected markets with connectivity $d \in (\Theta(\log n), o(\log^2 (n/|k|)))$; see Figure~\ref{fig:no-stark-effect}, $d \approx \log n$ is the threshold beyond which all agents becomes connected to each other with high probability. This is in sharp contrast to buyer-seller markets, where, roughly, connectedness of the consideration graph in the market implies strong competition where the short side of the market captures all the surplus (see Appendix~\ref{app:buyer-seller-app} for a detailed description of this phenomenon). We return to the underlying driver governing strong versus weak competition in matching markets after summarizing our findings for market design. 

\paragraph{Prescriptive insights: guidance for market design.}
\begin{table}[!ht]
    \centering
    \begin{tabular}{cccc}
    \hline
     \multirow{3}{*}{\textbf{Objectives}} & \multicolumn{3}{c}{\textbf{Platform Interventions}}
    \\ 
    \cline{2-4}
        & \emph{Optimal Consideration} 
        & \multicolumn{2}{c}{\emph{Optimal Preference List Length $d^*(n)$}}  
        \\ \cline{3-4}
        & \emph{Set Size $d^*(n)$} 
        & Short side proposes 
        & Long side proposes \\ \hline
        Short Side Welfare  
        & $\Theta(1)$ 
        & $\Theta(1)$ 
        & $\Theta(1)$ \\ 
        Long Side Welfare
        & $\Theta(1)$ 
        & $\Omega(1)$ and $O(\log^2 n)$ & $\Theta(1)$  \\ 
        Unmatched Agents  & \textup{Any} $\Omega(\log^2 n)$ & \textup{Any} $\Omega(\log^2 n)$ & \textup{Any} $\Omega(\log^2 n)$ \\ \hline
        {Efficient Frontier}  
        & \multicolumn{3}{c}{ Between $\Theta(1)$ and $\Theta(\log^2 n)$}  \\ \hline
    \end{tabular}
    \caption{\label{table:market-design-insight} Summary of our market design insights under the assumption $|k|\leq n^{1-\epsilon}$. 
    }
\end{table}

We analyze the effects of two platform interventions which allow to design the market connectivity: \emph{restricting the size of agents' consideration sets} and \emph{limiting the length of preference lists submitted by agents}. 
The platform is assumed to have two objectives: maximizing the welfare of all agents (which equals to match value minus preference discovery cost), 
and minimizing the number of unmatched agents. 
Our analysis is not tied to these specific objectives and can be adapted to optimize alternate performance metrics. We emphasize that we are able to obtain our prescriptive insights only as a consequence of our novel equilibrium characterizations for moderate connectivity levels exceeding a constant. 

The optimal level of connectivity we find in different situations and for different objectives (assuming small imbalance $|k|\leq n^{1-\epsilon}$) is summarized in Table \ref{table:market-design-insight}; see Section \ref{sec:design-implications} for a detailed discussion.
In most situations, the welfare of agents is found to be maximized when the size of their consideration set/preference list (denoted by $d$) is small, specifically $d=O(1)$. However, the number of unmatched agents is minimized when $d$ is moderate or larger, specifically $d=\Omega(\log^2 n)$. This means that the range of $d$ values that achieve the efficient frontier is relatively small, namely $[\Theta(1), \Theta(\log^2 n)]$, and lies in the weak competition regime.
Intuitively, increasing $d$ beyond $\Theta(\log^2 n)$ intensifies competition and reduces welfare, but does not help reduce the number of unmatched agents and is therefore dominated. Within the range  $d \in [\Theta(1), \Theta(\log^2 n)]$, as the platform increases $d$, it achieves fewer unmatched agents while also suffering lower agent welfare.
Notably, our prescription to operate the market in the weak competition regime is the opposite of that resulting from \cite{che2019efficiency}'s analysis of only densely connected markets (and idiosyncratic utilities with finite support, e.g., Uniform$(0,1)$), which suggests to operate the market in the strong competition regime. In Section~\ref{subsec:related-work}, we provide a detailed comparison and argue that one should indeed operate real-world markets in the weak competition regime.

Our analysis also shows that, in typical cases, the platform should adopt the following policy when it has the ability to choose which side initiates contact: When the two sides have different costs associated with preference discovery, it is optimal to have the side with lower cost reach out. On the other hand, if the platform is more concerned with the well-being of one side of the market, it would be optimal to have the other side initiate contact. 
%


\paragraph{In which matching markets does being on the short side confer an advantage?} Our analysis uncovers a new principle governing which random matching markets exhibit weak versus strong competition, which appears to generalize well beyond random markets.
\begin{principle}
A market exhibits weak competition, i.e., being on the short (long) side does not confer a significant advantage (disadvantage) in terms of match quality, if and only if the number of unmatched agents on the short side $\gtrsim$ the imbalance in the market.
\label{prin:condition-for-weak-competition}
\end{principle} 
Before discussing the practical usefulness of this simple principle, we provide some informal intuition for why it holds for random matching markets: 
Clearly, due to the matching constraint the number of unmatched men must be exactly $k$ plus the number of unmatched women. Hence, if and only if (i.f.f.) more than $|k|$ short side agents remain unmatched, the number of unmatched agents on the two sides must be within a factor two of each other. But, in a random market, the number of unmatched men should grow with $\Rmen$ (the more proposals men need to make in men-proposing DA, the larger the number of men that will reach the end of their preference list), whereas the number of unmatched women should similarly grow with $\Rwomen$ (one can consider women-proposing DA, and assume that, as is typical, the WOSM is close to the MOSM). We then deduce that the average ranks are similar on the two sides of the market if and only if more than $|k|$ short side agents remain unmatched. In particular, using the geometric decay of partner ranks suggested by our analysis (see Section~\ref{subsec:detailed-heuristic-picture}), we deduce, e.g., that under imbalance $k = o(n)$, we have more than $|k|$ agents remain unmatched on the short side i.f.f. \# unmatched short side agents $>$ (\# unmatched long side agents)$/2$ i.f.f. $\Rshort \in [(1-o(1))\Rlong,\Rlong]$. 

Notably, the aforementioned intuition for Principle~\ref{prin:condition-for-weak-competition} would 
seem to extend beyond uniformly random preferences to more realistic preference structures, and we indeed find that simulation results for markets with correlated preferences confirm the predictions of the principle (see Section~\ref{sec:numeric} and Appendix~\ref{app:numerical-additional}). 
The attractiveness of Principle~\ref{prin:condition-for-weak-competition} as a diagnostic tool in a real-world context is that (while the number of unmatched agents on the short side is admittedly an endogeneous quantity) the market imbalance and the number of unmatched agents on each side of the market are basic summary statistics which are publicly known for many markets, so the principle allows an ``outsider'' who lacks access to preference data to nevertheless estimate whether being on the short side confers an advantage in that market. 

In Section~\ref{sec:NYC-HS-counterfactuals}, we test the validity of Principle~\ref{prin:condition-for-weak-competition} in a real-world matching market. Using
data from centralized high school admissions in a major US city, we conduct a counterfactual analysis studying the impact of varying the imbalance in the market on the resulting match quality for applicants. Encouragingly, we find that the prediction from the principle aligns well with our findings, both for the real market, as well as for counterfactual markets with other levels of imbalance. 
Specifically, we find that  Principle~\ref{prin:condition-for-weak-competition} predicts strong competition for a given level of imbalance based on summary statistics alone, if and only if the detailed analysis uncovers a substantial impact of market imbalance on match quality, e.g., the percentage of applicants who get their top choice school changes by {$\gtrsim 4.5\%$}. Notably, whether applicants are on the short side or on the long side, Principle~\ref{prin:condition-for-weak-competition} is found to predict strong competition for markets exceeding nearly the same ``cutoff'' on match quality impact of imbalance (of course the sign of the change depends on the sign of the imbalance).  
Encouraged, we then leverage Principle~\ref{prin:condition-for-weak-competition} to make a conjecture about a different real-world market for which researchers have typically been denied access to detailed preference data. Namely, Principle~\ref{prin:condition-for-weak-competition} suggests that in the medical residency matching market in USA, being on the long (short) side only moderately degrades (improves) match quality for applicants (programs). 

\begin{figure}[htbp!]
\begin{center}
        \includegraphics[width=0.6\textwidth]{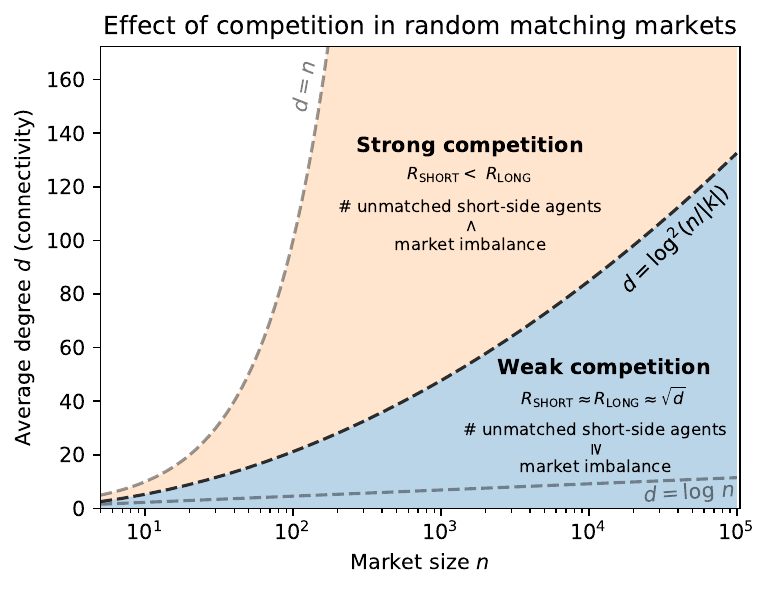}
        \caption{
        Schematic showing the two ``competitiveness'' regimes for partially connected random matching markets with $n$ agents on the long side,  connectivity (average preference list length) $d$, and general imbalance $k$ (numerical values are based on $k=-1$). $\Rshort$ ($\Rlong$) denotes the average rank of the short (long) side agents for their partners.
        }
        \label{fig:no-stark-effect}
\end{center}
\end{figure}

\medskip


\subsection{Related work}
\label{subsec:related-work}
Our work belongs to a vast theoretical literature on matching markets, which began with the work of \cite{gale1962} introducing stable matching and the deferred acceptance algorithm, and has developed over the last six decades with major contributions by Roth, Sotomayor, and a large number of other prominent researchers \citep[see, e.g.,][]{roth1992,david2013algorithmics}.
Closely related to our work are previous papers studying random matching markets with complete preference lists \citep{knuth1976mariages,Pittel89,knuth1990stable,pittel1992likely,AKL17,pittel2019likely}. 
Whereas the early papers focused on balanced random markets and found that the proposing side (in DA) has a substantial advantage, \cite{AKL17} and follow up papers found that in unbalanced markets, the short side has a substantial advantage, and the core (i.e., the set of stable matchings) is small. The main technical difficulty we face relative to these papers is that a positive number of agents remain unmatched on \emph{both} sides of the market in moderately connected markets $d = o(\log^2 n)$, preventing us from directly leveraging the analogy with the coupon collector problem as in previous works.

Notable papers by Immorlica and Mahdian and others \citep{immorlica2005marriage,kojima2009incentives} show a small core  while working with short (constant-sized) preference lists, leading to a linear fraction of unmatched agents. \cite{arnosti2022lottery} and \cite{menzel2015large} characterize the (nearly unique) stable outcome in settings with constant-sized preference lists, and in particular, we expect their characterizations can be used to show that the outcome changes ``smoothly'' as a function of the market imbalance under short lists. In contrast to the aforementioned papers, our work restricts attention to the case $d = \omega(1)$ and indeed identifies the existence of a threshold at $d \sim \log^2 n$, as a result of which the fraction of unmatched agents in our setting is vanishing. Technically, the consequence of this phenomenon is that ``rejection chains'' in the progress of DA are $\omega(1)$ in length in our work, making them harder to analyze, and the (approximate) system ``state'' no longer has bounded dimension as in \cite{arnosti2022lottery}. 

There is a robust and growing body of practical work on designing real world matching markets, especially in the contexts of school and college admissions \citep[e.g.,][]{rios2019improving,APR05,dur2018reserve}, and various labor markets \citep[e.g.,][]{RothPeranson99,hassidim2017redesigning}. Stability, namely, that no two agents should prefer to match with each other rather than their current partners, has been found to be crucial in the design of centralized clearinghouses \citep{roth1991natural} and predictive of outcomes in decentralized matching markets \citep{kagel2000dynamics,hitsch2010matching}. In most real world matching datasets, the short side of the market does not benefit from being on the short side if the market imbalance is small, consistent with the weak competition regime we find. 

Variants of the market design interventions we study (limiting preference list lengths and restricting the size of consideration sets) have previously been proposed and studied in a different context, namely, that of reducing congestion in \emph{decentralized} matching markets. We find it convenient to employ labor market terminology in summarizing this line of work. \cite{roth1994jumping} introduce the phenomenon of congestion, i.e., that it is impossible or costly for employers to make
offers to many applicants and that employers compete for the same applicants. In such cases, a reasonable design choice is to enable applicants to signal their interest to employers; see, e.g., \cite{lee2015propose} and \cite{coles2010job} for applications
in online data markets and the academic job market, respectively, and \cite{coles2013preference},
\cite{halaburda2018competing}, and \cite{jagadeesan2018varying}, for examples of how signaling can improve equilibrium outcomes. We emphasize that our market design insights are distinct from the aforementioned works, since we assume that a stable matching on the reported preferences will be implemented (e.g., using a centralized clearinghouse), and focus on gains in match quality from deploying these interventions. 

We now discuss market design interventions previously proposed for improving match quality in centralized matching markets. \cite{ashlagi2019assigning} and \cite{arnosti2022lottery} show that in school choice, single tie-breaking (the same lottery number being using by all programs to break ties between students) can produce higher quality match outcomes for students than multiple tie-breaking (different programs breaking ties independently). \cite{ashlagi2019assigning} also show that under {multiple} tie-breaking, when applicants are on the long side of a fully connected market,  applicants do poorly under unrestricted preference lists, and suggest to impose a restriction on the length of applicant preference lists can improve average rank without significantly reducing the number of assigned students. \cite{che2019efficiency} obtain similar results, and find that 
restricting the length of the preference lists of applicants to a level which is $\omega(\log ^2 n)$ but $o(n)$ can produce asymptotic efficiency while preserving asymptotic stability. The latter paper proposes DA with a circuit breaker (a generalization of DA with restricted preference list length) as a mechanism to achieve a good tradeoff between efficiency and stability in object allocation problems. 
Notably, these papers provide sharp characterizations only in densely connected markets (i.e., in the ``strong competition'' regime we find) \citep{ashlagi2019assigning,che2019efficiency}, or in markets with constant-length preference lists \citep{arnosti2022lottery}. We now make a more detailed comparison with the aforementioned finding of \cite{che2019efficiency}, 
which notably suggests to operate the market in what we identify as the strong competition regime. This is both quantitatively and qualitatively different from our finding that the Pareto optimal list lengths (when the long side proposes, optimizing for the long-side match utility and the number of unmatched agents) are in $[\Theta(1), \Theta(\log^2 n)]$ and cause the market to lie in the weak competition regime. Why this striking difference in findings, and which guidance is more practically relevant? 
The dual platform objectives  considered in the papers are similar (the stability objective in \cite{che2019efficiency} is closely related to our ``number of unmatched agents'' metric). 
The key difference lies in the way match utilities are modelled, which determines whether a non-trivial tradeoff arises between the two objectives. \cite{che2019efficiency} model the idiosyncratic part of agent match utilities as being uniform $(0,1)$ (or finite-support generalizations).
As a result, any $d = o(n)$ leads to asymptotically optimal long side welfare because almost all matched agents get match utility close to the maximum in the support of the idiosyncratic utility distribution. There is no real tradeoff between efficiency and stability (which may be viewed as unrealistic behavior of that model), and any $d \in (\omega(\log ^2 n), o(n))$ asymptotically optimizes both objectives. In contrast, we model heavy-tailed idiosyncratic cardinal preferences (which are arguably more realistic than finite-support preferences).  As a result, the long-side welfare is maximized for a preference list length restriction  $d=\Theta(1)$, and decays polynomially in $d$ in the strong competition regime. The number of unmatched agents is minimized for all $d$ larger than the threshold for strong competition $d \geq \Theta( \log^2 n)$, causing such $d$ to be Pareto dominated, and leading to a non-trivial tradeoff between the two objectives for $d \in (\Theta(1), \Theta(\log^2 n))$ in the weak competition regime. Our new, sharp characterization of stable outcomes in the weak competition regime allows us to quantify this tradeoff and prescribe an optimal list length restriction. Notably the data driven ``field study'' of \cite{che2019efficiency} indeed reveals a non-trivial tradeoff between stability and efficiency and ultimately arrives at a recommendation of relatively short lists (of length between 2 and 6), consistent with our theoretical findings.


\paragraph{Organization of the paper.}
In Section~\ref{sec:model}, we introduce our model of partially connected random matching markets.
In Section~\ref{sec:result}, we state our main theorems (Theorem \ref{thm:main-result} and \ref{thm:main-result-dense}) and discuss them.
An overview of our proof of our characterization of moderately connected markets (Theorem~\ref{thm:main-result}) is provided in Section~\ref{sec:proof-sketch}.
In Section~\ref{sec:design-implications}, we obtain market design insights based on our main results.
In Section~\ref{sec:numeric}, we provide the simulation results that confirm and sharpen our theoretical predictions.
In Section~\ref{sec:NYC-HS-counterfactuals} we use real-world data to test our prescriptive and design insights.
Formal proofs are relegated to the appendix. 

\paragraph{Asymptotic notations.}
For two sequences of positive real numbers $\{a_n\}_{n=1}^{\infty}$ and $\{b_n\}_{n=1}^{\infty}$: We write $a_n = O(b_n)$ as $n\to\infty$ if $\limsup_{n\to\infty}a_n/b_n < \infty$; We write $a_n = o(b_n)$ as $n\to\infty$ if $\limsup_{n\to\infty}a_n/b_n =0$; We write $a_n = \Omega(b_n)$ as $n\to\infty$ if $b_n = O(a_n)$; We write $a_n = \omega(b_n)$ as $n\to\infty$ if $b_n=o(a_n)$; We write $a_n = \Theta(b_n)$ as $n\to\infty$ if $a_n = O(b_n)$ and $b_n = O(a_n)$.

\section{Model} \label{sec:model}

We consider a two-sided market that consists of a set of men $\mathcal{M} = \{1,\ldots,n+k\}$ and a set of women $\mathcal{W}=\{1,\ldots,n\}$. Here $k$ is a positive or negative integer, which we call the \emph{imbalance}.

An undirected bipartite random graph $G$ connecting men $\cM$ to women $\cW$. 
Given $G$, each agent has a strict preference ranking (denoted by $\succ_i$ for agent $i$) over all his/her neighbors in $G$ and does not rank any other agents. 
Woman $j$ (man $i$)'s neighbors in $G$ are denoted by $\cM_j$ (resp., $\cW_i$).
A \textit{matching}
is a mapping $\mu$ from $\mathcal{M}\cup\mathcal{W}$ to itself such
that for every $i\in\mathcal{M}$, $\mu(i)\in\mathcal{W}_i\cup\{i\}$,
and for every $j\in\mathcal{W}$, $\mu(j)\in\mathcal{M}_i\cup\{j\}$,
and for every $i,j\in\mathcal{M}\cup\mathcal{W}$, $\mu(i)=j$ implies
$\mu(j)=i$. We use $\mu(j)=j$ to denote that agent $j$ is unmatched under $\mu$.
A matching $\mu$ is \textit{unstable} if there is a man $i$ and
a woman $j$ such that $j\succ_i\mu(i)$ and $i\succ_j\mu(j)$ (called a \emph{blocking pair}).
A matching is \textit{stable} if there is no blocking pair. 

A \emph{random matching market} is generated by drawing:
\begin{itemize}
    \item An undirected bipartite random consideration graph $G$ connecting men $\cM$ with women $\cW$, where each man is connected to $d$ neighboring women (denoted by $\cW_i \subset \cW$ for man $i$) selected uniformly at random and independently across men (from among the $\binom{n}{d}$ possibilities).
    \item For each man $i$, a uniformly random complete preference list over $\cW_i$, and for each woman $j$, a uniformly random complete preference list over $\cM_j$, independently across agents.
\end{itemize} 

A stable matching always exists. It can be found using the Deferred Acceptance (DA) algorithm by Gale and Shapley \citep{gale1962}. They show that the men-proposing DA finds the \textit{men-optimal stable matching} (MOSM), in which every man is matched with his most preferred stable woman. The MOSM matches every woman with her least preferred stable man. Likewise, the women-proposing DA produces the women-optimal stable matching (WOSM) with symmetric properties. All of our results will characterize the MOSM. Given the strong evidence from \cite{immorlica2005marriage,kojima2009incentives,AKL17} and other works that the MOSM and WOSM are nearly the same in typical matching markets (with the exception of balanced and densely connected random markets, which we avoid by assuming $k< 0$ in Theorem~\ref{thm:main-result-dense}), we omit to formally show this fact for our setting though we believe it can be done, e.g., using the method developed in \cite{cai2019short} (the property $\textup{MOSM} \approx \textup{WOSM}$ is found to hold consistently in our numerical simulations of our model).

We are interested in 
how matched agents rank their assigned partners under stable matching, and in the number of agents who are left unmatched.
Denote the rank of woman $j$ in the preference list $\succ_i$ of man $i$ by $\rank_i(j)\equiv\nobreak|\{j':j'\succeq_i j\}|$. Smaller ranks are preferred,  and $i$'s most preferred woman has a rank of $1$.
Symmetrically, denote the rank of $i$ in the preference list of $j$ by $\rank_j(i)$.

\begin{defn}
\label{def:avg-ranks-and-unmatched-counts}
Given a matching $\mu$, the \emph{men's average rank of wives} is given by
\[
\Rmen(\mu) = \frac{1}{n+k} \left ( |\bar{\cM}(\mu)| (d+1) \  +  \sum_{i\in\cM\backslash \bar{\cM}(\mu)}\rank_i(\mu(i))  \right )\, ,
\]
where $\bar{\cM}(\mu)$ is the set of men who are unmatched under $\mu$, and the \emph{number of unmatched men} is denoted by $\delta^m(\mu)$, i.e., $\delta^m(\mu) = |\bar{\cM}(\mu)|$.

Similarly,  the \emph{women's average rank of husbands} is given by
\[
\Rwomen(\mu) = \frac{{1}}{n}  \left ( \sum_{j \in \bar{\cW}(\mu)} (|\cM_j| +1) \  +  \sum_{j\in\cW\backslash \bar{\cW}(\mu)}\rank_j(\mu(j))  \right )
\]
where $\bar{\cW}(\mu)$ is the set of women who are unmatched under $\mu$, and the \emph{number of unmatched women} is denoted by $\delta^w(\mu)$, i.e., $\delta^w(\mu) = |\bar{\cW}(\mu)|$.
\end{defn}

(Note here that if an agent is unmatched, we take the rank for the agent to be one more than the length of the agent's preference list.) By the rural hospital theorem \citep{Roth86}, the set of unmatched agents ($\bar{\cM}(\mu)$ and $\bar{\cW}(\mu)$) is the same in every stable matching $\mu$, and therefore we simply represent the number of unmatched men and women under stable matching by $\delta^m$ and $\delta^w$ respectively throughout the remainder of paper.

We remark that the only asymmetry in our model is that the lengths of men's preference lists are deterministically $d$,
whereas each woman has $\textup{Binomial}(n+k, d/n) \xrightarrow[n \to \infty]{\textup{d}} \textup{Poisson}(d)$ neighbors\footnote{The approximation is correct since we assume in our main result (Theorem \ref{thm:main-result}) that $|k|=o(n)$.} where $\xrightarrow{\textup{d}}$ denotes convergence in distribution. Since our theoretical analysis will assume $d = \omega(1)$, we have $\textup{Poisson}(d) \xrightarrow{\textup{p}} d$, i.e., the degree of each woman is also very close to $d$, and so the asymmetry between the two sides in the model is a technical one.\footnote{{We numerically tested the behavior of random markets with a bipartite Erdos-Renyi connectivity graph with edge probability $d/n$ (and hence average degree $d$), and found the behavior to be very similar to that under our formal model (if anything, the theoretical predictions were found to be even more accurate for the Erdos-Renyi model); see Appendix~\ref{app:numerical-additional}. Note that the Erdos-Renyi connectivity model is symmetric in the two sides of the market.}} 

\section{Theoretical Results}\label{sec:result}
In this section we state and discuss our main theoretical results, {which characterize the average rank of partners and the number of unmatched agents in markets, as a function of market connectivity. 
In Section \ref{subsec:detailed-heuristic-picture} we provide intuition leading informally to a detailed picture of the market equilibrium.}

\paragraph{\bf Moderately and sparsely connected markets.}
In our first main result, we show that there is an insignificant advantage from being on the short side in partially connected markets with small imbalance whose connectivity parameter $d$ is $o(\log^2 n)$. 

\begin{thm}[Moderately Connected Markets]
\label{thm:main-result}
	Fix any $\epsilon > 0$. Consider a sequence of random matching markets indexed by $n$, with $n+k$ men and $n$ women ($k=k(n)$ can be positive or negative or zero), and connectivity (average degree) $d=d(n)$, with $d=\omega(1)$ and $d=o(\log^2 n)$, and 
	{$|k|\leq n^{1-\epsilon}$}.
	Then with high probability,\footnote{Specifically, our characterization holds with probability at least $1- O(\exp(-d^{1/4})) = 1- o(1)$.} we have
	\begin{align*}
		\left|\Mrank - \sqrt{d}\right|
		\leq \
		d^{0.3}
		\, , \\
		\left|\Wrank - \sqrt{d}\right|
		\leq \
		d^{0.3}
		\, , \\
		\left|\log \delta^m
		-
		\log\left(ne^{-\sqrt{d}}\right)
		\right|
		\leq \
		d^{0.3}
		\, , \\
		\left|\log \delta^w
		-
		\log\left(ne^{-\sqrt{d}}\right)
		\right|
		\leq \
		d^{0.3}
		\, .
	\end{align*}
\end{thm}
Informally, in large random matching markets with average degree $d = o (\log^2 n)$ and a small imbalance $k \leq n^{1-\epsilon}$, under stable matching we have $\Rmen \approx \Rwomen \approx \sqrt{d}$ irrespective of which side is the short side, and there are approximately $n e^{-\sqrt{d}} = \omega(1)$ unmatched agents on both sides of the market.
Thus there is no short-side advantage 
and agents on both sides are matched to their $\sqrt{d}$-th ranked partner on average.
%
%
A significant number of agents are left unmatched even on the short side, in contrast to a fully connected unbalanced matching market where all agents on the short side are matched. 
Though we only characterize the MOSM in the present version of the paper, we believe the same characterization extends to the WOSM as well. We give an overview of the proof of Theorem~\ref{thm:main-result} in Section~\ref{sec:proof-sketch} and the formal proof in Appendix~\ref{append:small-medium-d-proof}. Note that $\log (n/|k|) = \Theta(\log n)$ under $|k| \leq n^{1-\epsilon}$ and hence the theorem can equivalently be stated for $d = o (\log^2 (n/|k|))$.

The main intuition for Theorem~\ref{thm:main-result} is that for $d  = o(\log^2 (n/|k|))$, more than $|k|$ men remain unmatched with high probability, because they reach the end of their preference lists in men-proposing DA (cf. \cite{pittel2019likely}, who showed that some men need to go $\log^2 n$ deep in their preference lists in the fully connected market). Clearly, the number of unmatched men must be exactly $k$ plus the number of unmatched women.
Then, assuming a small imbalance $k$, the number of unmatched agents on the two sides must be nearly the same (up to a factor less than 2).
But the number of unmatched men should grow with $\Rmen$ (the more men need to propose, the larger  the number that will reach the end of their preference lists), whereas the number of unmatched women should similarly grow with $\Rwomen$ (e.g., one can consider women proposing
DA, and assume that, as usual, the WOSM is close to the MOSM). We deduce that we should have $\Rmen \approx \Rwomen$ in the $d \ll \log^2 (n/|k|)$ regime. (Informal quantitative intuition leading to the precise estimates of $\Rmen$ and $\delta^m$ will be provided later in Section \ref{subsec:detailed-heuristic-picture}.)



We highlight that Theorem~\ref{thm:main-result} encompasses a wide range of connectivity parameters $d = o(\log^2 n)$ (for $k \leq n^{1-\epsilon}$), which extends far beyond the connectedness threshold of the consideration graph $d_{{\rm {\textup{conn}}}}^* \approx \log n $ (this is also the connectedness threshold for Erd\H{o}s-R\'{e}nyi random graphs).
Thus our ``weak competition regime'' result does not require a disconnected or fragmented market.
Rather, the result applies even to very well connected markets.\footnote{For example, with $n = 1,000$, $\log^2 n \approx 48$. Taking $d = 10$ (much less than 48), numerics tell us that 9.6\% of pairs of men are within 1 hop of each other (i.e., there is woman who is ranked by both men), and 99.98\% of pairs of men are within 2 hops of each other.} This is in sharp contrast to buyer-seller markets, where, roughly, connectedness of the consideration graph implies a strong effect of competition. See Appendix~\ref{app:buyer-seller-app} for a precise description of the latter behavior of buyer-seller markets.
%

Numerical simulations in the Section~\ref{sec:numeric} show that the finding in Theorem~\ref{thm:main-result} holds up extremely well for all $d \lesssim 1.0 \log^2 (n/|k|)$ for realistic values of $n$ (not just asymptotically in $n$ for $d = o(\log ^2 n)$), and indeed extends much beyond the connectedness threshold for small $k$.  

\paragraph{\bf Densely connected markets.}
Our next result shows that for $d \gtrsim \log^2 n$, the finding of \cite{AKL17} holds true, i.e., the short side is markedly better off even in (large) markets with a  small imbalance (we note that the recent paper \cite[][Theorem 4]{che2019efficiency} contains a related result). Moreover, this benefit of being on the short side arises in conjunction with the key property that {\emph{most agents on the short side of the market are matched}. (The theorem establishes a stronger property leveraging its strong assumption $d = \omega(\log^2 n)$, namely, with high probability, \emph{all} short side agents are matched.)} 

\begin{thm}[Densely Connected Markets]
	\label{thm:main-result-dense}
	Consider a sequence of random matching markets indexed by $n$, with $n+k$ men and $n$ women, and connectivity (average degree) $d=d(n)$, with $k=k(n)<0$ and $|k|=o(n)$, $d=\omega(\log^2 n)$ and $d=o(n)$. Then, with high probability, all men are matched under stable matching, and we have
	\begin{align*}
	\Mrank
	\leq& \
	(1+o(1))\log n
	\, , \\
\Wrank
	\geq& \
(1+o(1))
	\frac{d}{\log n}
	\, .
	\end{align*}
\end{thm}

This result shows that the short-side advantage emerges in densely connected markets even when the imbalance is small (including for an imbalance of one, i.e., $k = -1$).
More specifically, when $d = \omega(\log^2 n)$, it predicts that the agents on the short side are matched to their $\log n$-th ranked partner on average whereas the agents on the long side are matched to their $\big( \frac{d}{\log n} \big)$-th ranked partner on average.
Theorem \ref{thm:main-result-dense} smoothly interpolates between the result in AKL \citep{AKL17} and our Theorem \ref{thm:main-result} (though the extremes $d = \Omega(n)$ and $d = \Theta(\log^2 n)$ are not covered by the formal statement in present form): as connectivity $d$ increases, a phase transition happens at $d=\Theta(\log^2 n)$, and the short side advantage starts to emerge for $d=\omega(\log^2 n)$. The magnitude of the advantage increases as the market becomes denser.
Combining Theorems \ref{thm:main-result} and \ref{thm:main-result-dense}, we conclude that, assuming a small imbalance, a short-side advantage exists if and only if a matching market is connected densely enough, and the threshold level of connectivity $d \sim \log^2 n$. 

The analysis leading to Theorem~\ref{thm:main-result-dense} is similar to that leading to \cite[][Theorem~2]{AKL17}. The number of proposals in men-proposing DA remains unaffected; the only change is that women now have rank lists of approximate length $d$ (instead of length $n+k$), and so, receiving about $\log n$ proposals leads to an average rank of husband of about $d/\log n$. The proof is in Appendix~\ref{append:dense-d}.

\subsection{Detailed heuristic picture of market equilibrium}\label{subsec:detailed-heuristic-picture}

In this section, we provide more detailed quantitative intuition which helps explain our estimates of $\Rmen$ and $\Rwomen$ in Theorem~\ref{thm:main-result}, and also yields Principle~\ref{prin:condition-for-weak-competition} capturing which random markets exhibit strong competition. 
This intuition is based on a detailed heuristic picture of the stable outcome in a random matching market.
We do not formally prove this detailed picture in this paper (Theorem~\ref{thm:main-result} is instead proved via a ``shortcut'').
%

\paragraph{Intuition for Theorem~\ref{thm:main-result}.} Consider the man-proposing deferred acceptance algorithm in a random market. Intuitively, both $\Rmen$ and the number of unmatched men $\delta^m$ should be governed by the (endogenous) probability $\pmen$ that a neighboring woman $j$ (independently of other women) is ``interested'' in given man $i$ (the woman $j$ is said to be interested if she receives no proposal which she prefers to $i$): in particular, the rank of man $i$ for his wife (his most preferred woman who accepts his proposal) should be distributed as $\Geo(\pmen)$ truncated at $d$, leading to $\Rmen \approx 1/\pmen$ (assuming $1/\pmen \ll d$) and $\delta^m \approx n\prob(\Geo(\pmen) > d) = n(1-\pmen)^{-d} \approx n\exp(-d \pmen)$. 
Analogously for women, letting $\pwomen$ denote the (endogeneous) probability that woman $j$ receives a proposal from each neighboring man $i$,
we expect $\Rwomen \approx 1/\pwomen$ and $\delta^w \approx 
n\exp(-d \pwomen)$, with partner ranks distributed as Geometric$(\pwomen)$ truncated at the woman's degree (which is random but concentrated around $d$).
For $k$ small, in moderately connected markets where many short side agents are unmatched $\delta^m \gg k$, we have that both sides must have nearly the same number of unmatched agents $\delta^w \approx \delta^m$ and hence $\pwomen \approx \pmen$ and $\Rmen \approx \Rwomen$. But we can further get quantitative estimates: the average number of proposals received by women is nearly the same as the average number of proposals made by men $(n+k)\Rmen/n \approx \Rmen \approx 1/\pmen$, and since $\pwomen \approx \textup{average number of proposals received}/\textup{(woman's degree)} \approx 1/(d \pmen)$. We deduce that $\pmen \approx \pwomen \approx  \frac{1}{\sqrt{d}}$ and so $\Rmen \approx \Rwomen \approx \sqrt{d}$ and $\delta^m \approx \delta^w \approx n e^{-\sqrt{d}}$. 

\paragraph{Principle~\ref{prin:condition-for-weak-competition} and the threshold level of connectivity.}
{In Appendix~\ref{app:heuristic-numerical}, we refine the detailed heuristic picture above, and deduce a conjecture which amounts to a detailed version of Principle~\ref{prin:condition-for-weak-competition} for random markets. The conjecture (stated in Appendix~\ref{app:heuristic-numerical}) may be informally summarized as follows:
\emph{Fix any $\nu \in (0, \infty)$ and assume $|k| \leq n^{1-\epsilon}$ with men on the short side. Then there is a threshold connectivity $d^* = \log^2 (n/|k|) (1+o(1))$, and a rank ratio threshold $t = 1- \Theta(1/\log n)$ such that:
\begin{compactitem}[leftmargin=*]
    \item If $d \leq d^*$, with high probability, $\Rmen /\Rwomen \gtrapprox t$ and $\delta^m \gtrapprox \nu |k|$. 
    \item If $d > d^*$, with high probability, $\Rmen /\Rwomen \lessapprox t$ and $\delta^m \lessapprox \nu |k|$. 
\end{compactitem}
} 
In particular, few agents are unmatched on the short side (relative to the imbalance) if and only if the short side is significantly better off. We find that $\nu=0.5$ produces high quality numerical estimates in finite random markets, and hence use that choice of $\nu$ in deploying our insights in Sections~\ref{sec:numeric} and \ref{sec:NYC-HS-counterfactuals}.
 Note that the sharp estimate $d^* \approx 1.0\times \log^2(n/|k|)$ of the boundary between the strong and weak competition regimes is consistent with our (weaker) formal results, Theorems \ref{thm:main-result} and \ref{thm:main-result-dense}. Theorem~\ref{thm:main-result-dense}  establishes strong competition for $d = \omega( \log ^2 n) = \omega(\log^2 (n/|k|))$. Theorem~\ref{thm:main-result}, under the assumption of small imbalance $|k|=O(n^{1-\epsilon})$, shows weak competition under $d = o(\log^2 n) = o(\log^2 (n/|k|))$. }
 



\paragraph{Distribution of agents' ranks of partners.} In order to understand how market connectivity impacts the welfare of agents, one needs to consider the distribution (rather than the average) of agents' ranks of partners. Reassuringly, our market design insights, obtained in Section \ref{sec:design-implications}, are not sensitive to the specifics of these distributions as long as they have a subexponential tail, which we conjecture is indeed the case. For example, the insights hold when men's ranks of wives follow a (truncated) Geometric distribution as conjectured above, and the number of proposals a woman receives follows a Poisson distribution (as one may analogously conjecture).

\section{Implications for the Design of Matching Platforms}
\label{sec:design-implications}
In this section, we demonstrate how to use the theoretical findings made in Section~\ref{sec:result} to tackle market design questions. We discuss the following platform interventions which allow to control the market connectivity, and how to optimally deploy them in random markets:
%
\begin{itemize}
    \item \emph{Limiting consideration sets} (Section \ref{subsec:limit-consideration-set}). The platform constructs a random consideration graph and only presents to each agent their neighbors in the consideration graph (consideration sets) as potential partners. We denote the average cardinality of each consideration set by $d(n)$, which is chosen by the platform. Agents provide a preference ranking over their consideration sets and the platform then implements a stable matching.
    \item \emph{Limiting preference list lengths} (Section \ref{subsec:limiting-list-length}). We assume that contacts are always initiated by one side of the market, which we call the proposing side. The platform sets a limit $d(n)$ on the length of preference lists of the agents on the proposing side. The other (``receiving'') side of the market is asked to construct a preference ranking over proposing-side agents who reached out to them. Agents provide their preference lists and the platform then implements a stable matching.
\end{itemize}
We quantify market performance in terms of two metrics: the utilitarian welfare of agents, and the number of unmatched agents. For both interventions, we reach similar high-level conclusions regarding the impact of $d(n)$ on market performance: 
\begin{compactenum}[(1),wide,labelwidth=!,labelindent=0pt]
    \item Typically, utilitarian welfare of the agents is maximized for a consideration set size/allowed preference list length (both denoted by $d$) which is small (i.e., $d=\Theta(1)$). 
    \item The number of unmatched agents is minimized for moderate-sized or larger $d$; specifically, $d=\Omega(\log^2 n)$. While the number of unmatched agents is (weakly) decreasing in $d$, there is no further reduction in the number of unmatched agents from increasing $d$ above $d=\Theta(\log^2 n)$. 
    \item As a result, the set of $d$ that are Pareto optimal is small to moderate, in particular, between $\Theta(1)$ and $\Theta(\log^2 n)$.
    This provides theoretical evidence that supports the case of using small to moderate $d$ in real-world applications. Large $d$ leads to wasteful competition among agents which decreases utilitarian welfare without reducing the number of unmatched agents. 
\end{compactenum}

When the platform has the flexibility to choose which side initiates contact, we find that in typical cases the following policy is optimal: When the two sides have different costs associated with preference discovery, it is optimal to have the side with lower cost reach out. On the other hand, if the platform is more concerned with the well-being of one side of the market (and preference discovery costs are similar across the two sides), it is optimal to have the other side initiate contact. We discuss this finding in more detail at the end of Section \ref{subsec:limiting-list-length} and compare it with results in previous works, e.g., \cite{kanoria2021facilitating}.

{We conclude our summary by comparing the efficacy of the two interventions studied: For the same $d$, we find that limiting list lengths provides higher match quality for the proposing side, while achieving the same number of unmatched agents and the same utility for the receiving side. The caveat is that the limiting list length option inherently requires more ``preference discovery" effort from proposing agents in constructing their preferences over the receiving side of the market, relative to the intervention of limiting consideration set size. This tradeoff governs the appropriate choice between the interventions for a given market, if either intervention is feasible a priori.\footnote{We do not formally study this tradeoff in this paper.}} 

\subsection{Metrics for evaluating matchings}\label{subsec:metrics-market-design}
We study how different platform interventions perform according to the following two objectives:
\begin{compactitem}[leftmargin=*]
    \item \emph{Utilitarian welfare.} Intuitively, the utility obtained by each agent should have two components: the value derived from matching with another agent, and the cost of discovering their own preferences. We will define them formally below.
    \item \emph{Number of unmatched agents.} In many applications, the platform also cares about how many agents are left unmatched. We use the number of unmatched short-side agents $\delta$ as the second performance metric (note that the number of unmatched long-side agents is simply $\delta$ plus the market imbalance $|k|$).
\end{compactitem}

\paragraph{Match value through random utility.} In this section, we think of the {ordinal} preferences in the random matching market model defined in Section \ref{sec:model} can be generated by the following {cardinal} random utility model: the value agent $i$ obtains when matched with agent $j$ is drawn i.i.d. from distribution\footnote{The analysis extends immediately to more general value distributions, e.g., $F$ depends on which side the agent belongs, though at a significant notational burden. We reason that the cost of carrying the reader through this generalization exceeds the benefit of doing so, and hence assume the same value distribution for both sides throughout the paper.} $V_{ij}\sim F$. Note that the ordinal preference ranking resulting from such random utility draws is uniformly distributed among all permutations, and independent across agents.
We assume that unmatched agents receive a match value of zero.

A particular family of value distributions of interest is the power-law family, i.e., the Pareto distributions. These distributions capture the phenomenon that the value differential near the top of an agents' preference list (e.g., the difference in value between the $1$-st and $3$-rd ranked choices) is typically much larger than the value differential lower in the agents' preference list (e.g., the difference in value between the $13$-th and $15$-th ranked choices). Such heavy-tailed valuations have been fruitfully modelled in other contexts such as bundling of products \citep{ibragimov2010optimal}, and more generally, heavy-tailed distributions have been observed and studied in a wide variety of contexts in finance, economics, marketing, and operations, among other fields \citep[see, e.g.,][]{ibragimov2015heavy, resnick2007heavy, shapiro1999information, anderson2006long, taleb2007black}.
In this section, we assume that $F$ follows a \emph{Pareto distribution} with parameters $(1,\alpha)$, i.e., the probability density function of $F$ is $\alpha x^{-(\alpha+1)}$. Here we assume $\alpha > 1$ since otherwise $F$ has unbounded mean. 

\paragraph{Cost of preference discovery.} We model the cost incurred by an agent in determining their preference ranking over a given set of $d$ potential partners as $d^{\gamma}$, where $\gamma \in (0,1]$. This cost structure is based on the idea that the marginal cost of considering an additional potential partner tends to decrease as the number of partners increases. It is worth noting that if an agent needs to report their top $d$ preferences among a larger pool of $n$ options, the cost is modelled as $n^{\gamma}$ (as the agent presumably evaluates all of the available choices).



\paragraph{Preliminaries: random order statistics.}
If an agent gets matched to their $k$-th most preferred partner among $d$ choices, the match value they obtain is distributed as the $k$-th largest \emph{order statistic} out of $d$ independent samples from $F$, denoted by $F^{(k),d}$. As a result, order statistics play a key role in our analysis, and we provide below a technical result that will come in handy later.

To make the presentation easier to understand, it is helpful to introduce additional asymptotic notation. For two positive real-valued sequences $\{a_n\}$ and $\{b_n\}$: We use the notation $a_n \doteq b_n$ if $a_n = \Theta(b_n)$ as $n\to\infty$, and the notation $a_n \dotleq b_n$ if $a_n = O(b_n)$ as $n\to\infty$. 

The following lemma (proved in Appendix \ref{append:additional-proof})  specifies the scaling behavior of the relevant \emph{random} order statistics.

\begin{lem}\label{lem:random-order-statistics}
    We have the following results:
   \begin{compactenum}[(1),wide,labelwidth=!,labelindent=0pt]
        \item Let $M(d(n),\alpha,r(n))$ be a sub-exponential random variable which is the $\zeta$-th largest order statistic out of $d(n)$ samples from Pareto distribution with scale parameter $1$ and shape parameter $\alpha>1$. Here $\zeta$ is a sub-exponential random variable with mean $r(n)$. Assume that $d(n),r(n)\to\infty$ as $n\to\infty$. We have
     $
            \mathbb{E}[M(d(n),\alpha,r(n))] \doteq \left(\frac{d(n)}{r(n)}\right)^{1/\alpha}\, .
      $
        %
        \item Let $W(d(n),\alpha,r(n))$ be a random variable which is the largest order statistic out of $\eta$ samples from a Pareto distribution with scale parameter $1$ and shape parameter $\alpha>1$. Here $\eta$ is a sub-exponential random variable with mean $r(n)$. We have
       $
            \mathbb{E}[W(d(n),\alpha,r(n))]\doteq
            (r(n))^{1/\alpha}\, 
    $.
\end{compactenum}
\end{lem}

Roughly speaking, Lemma \ref{lem:random-order-statistics} states that if an agent's rank distribution has a light tail, the stochasticity in the distribution can be disregarded when determining the order of magnitude of the average match value.
We  make use of  Theorem~\ref{thm:main-result}, in conjunction with Lemma~\ref{lem:random-order-statistics} and our assumption that rank distributions are truncated Geometric (and hence subexponential), to obtain the estimates of welfare and the number of unmatched agents presented in the rest of this section. 


\subsection{Limiting consideration sets}\label{subsec:limit-consideration-set}
We first study the effect of limiting the size of agents' considerations sets. Throughout this subsection, we denote the consideration set size by $d(n)$. Let $U_S^\cC(d(n))$ ($\cC$ is short for $\mathcal{C}$onsideration set) be the utility of a representative agent on the short side of the market, and $U_L^{\cC}(d(n))$ is defined similarly for an agent on the long side. Let $\delta_S(d(n))$ be the number of unmatched short-side agents. 
Note that the match value of short-side (long-side) agents has the same distribution as $M(d(n),\alpha,r(n))$ ($W(d(n),\alpha,r(n))$) defined in Lemma \ref{lem:random-order-statistics}, where $r(n)$ and other estimates come from Theorem \ref{thm:main-result}.

Assuming the rank distributions of agents are sub-exponential and small market imbalance satisfying $|k|\leq n^{1-\epsilon}$ for some fixed $\epsilon>0$, we obtain  the following approximations of $U_S^\cC(d(n))$ and $U_L^\cC(d(n))$ (using the asymptotic notation we introduced in Section \ref{subsec:metrics-market-design}):

\begin{itemize}
    \item Short-side agents' average utility is
$
        \mathbb{E}[U_S^\cC(d(n))] \doteq \left( \max\left\{\sqrt{d(n)},\, \frac{d(n)}{\log n}\right\} \right)^{\frac{1}{\alpha}} - d(n)^{\gamma} 
 $.
    \item Long-side agents' average utility is
$
        \mathbb{E}[U_L^\cC(d(n))] \doteq \left( \min\left\{\sqrt{d(n)},\, \log n\right\} \right)^{\frac{1}{\alpha}} - d(n)^{\gamma} 
 $.
    \item The average number of unmatched (short-side) agents is
$
       \mathbb{E}[\delta_S(d(n))]\doteq n\cdot e^{-\sqrt{d(n)}}\, 
 $.
\end{itemize}

Simulation results (see Figure \ref{fig:theory-vs-simulation}) confirm that the above predictions approximate well the stable matching in random matching markets.
The estimates above establish that two regimes arise in the planner's problem.
\begin{compactitem}[leftmargin=*]
    \item \emph{If the preference learning cost grows slowly in $d$, in particular if $\gamma < \frac{1}{\alpha}$, it is optimal to set $d(n)=n$}, as there is no trade-off between welfare and number of matches formed, and it is best to have a fully connected market.
    \item \emph{If the preference learning cost grows quickly in $d$, in particular if $\gamma \geq \frac{1}{\alpha}$, the Pareto optimal $d(n)$ lie in the range $1 \dotleq d(n) \dotleq \log^2 n$}. 
    In this case, there is a trade-off between the two objectives:
    Utilitarian welfare is maximized at $d(n)\doteq 1$ since the preference learning cost dominates the match value so smaller $d(n)$ is preferred.
    The number of unmatched agents is minimized for any $d(n) \dotgeq \log^2 n$.
    Therefore, the Pareto optimal $d(n) \in [\Theta(1), \Theta(\log^2 n)]$. 
\end{compactitem}

We expect that the preference learning cost grows quickly in $d$ in most real world markets, i.e., the second regime above is typical. The finding above suggests that a platform should deploy small to moderate-sized consideration sets ($1 \dotleq d(n) \dotleq \log^2 n$) in such markets. This insight is corroborated by our simulation results: Figure~\ref{fig:theory-vs-simulation} illustrates that increasing $d(n)$ within the range $1 \dotleq d(n) \dotleq \log^2 n$ can significantly reduce the number of unmatched agents (this finding applies to both interventions we study), whereas increasing $d(n)$ beyond $\log^2 n$ does not reduce this number by much. Meanwhile, the average welfare of agents is maximized for a smaller ($\Theta(1)$) value of $d$. The Pareto frontier between the two objectives is generated by $d$ between these values.

\begin{figure}[h]
    \centering
    \begin{minipage}{0.49\textwidth}
        \centering
        \includegraphics[width=0.95\textwidth]{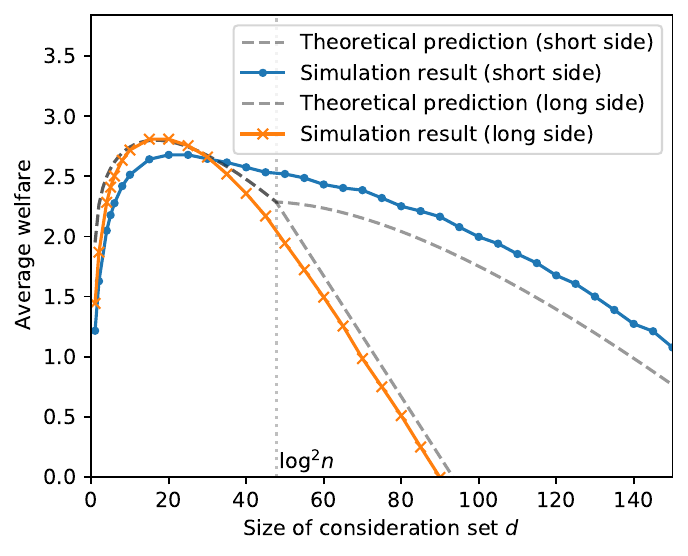}
    \end{minipage}\hfill
    \begin{minipage}{0.49\textwidth}
    \centering
    \includegraphics[width=0.95\textwidth]{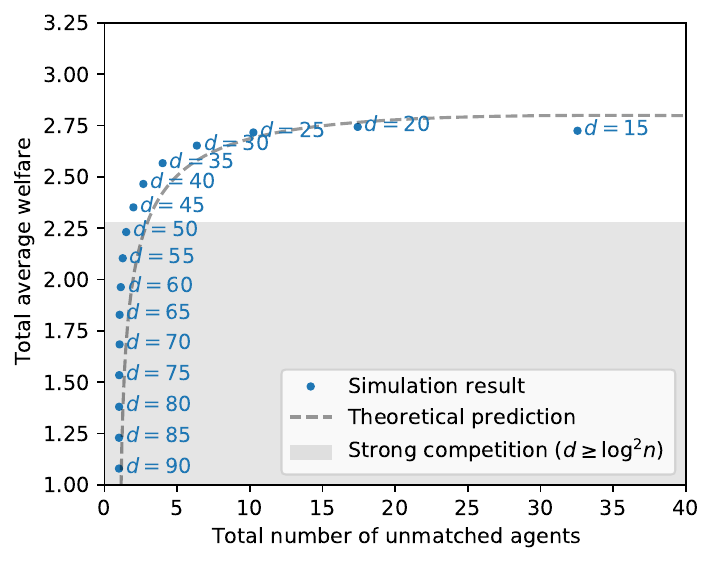}
    \end{minipage}
 \caption{Illustration of effect of limiting consideration set size on agents average welfare and the number of unmatched agents, in a random matching market with 1,000 men and 1,001 women, where Pareto utility model ($\alpha=2$) and linear preference discovery cost model ($\gamma=1$ with $0.05$ unit cost per discovery) are assumed.
 In the left figure, solid lines indicate the average welfare of the agents on each side, obtained through 1,000 runs of simulations, and dashed lines indicate their corresponding values that our analysis predicts.
 The vertical line at $d = \log^2(1000/1) \approx 47.7$ indicates the threshold consideration set size beyond which the market enters the strong competition regime.
 The right figure shows the Pareto frontier: each data point reports the total number of unmatched agents vs. the total average welfare across all agents in the market given the consideration set size limit $d$.} 
 \label{fig:theory-vs-simulation}
\end{figure}

\subsection{Limiting preference list length}
\label{subsec:limiting-list-length}

We now study the optimal preference list length restriction. 
Let $U_S^{\cP}(d(n))$ ($\cP$ is short for $\cP$reference list) be the utility of a representative agent on the short side of the market, and $U_L^{\cP}(d(n))$ is defined similarly for an agent on the long side. Let $\delta_S(d(n))$ be the number of unmatched short-side agents. 
Note that it matters which side proposes: The preference discovery cost is $n^\gamma$ on the proposing side of the market, and only $(d(n))^\gamma$ on the receiving side of the market. On the other hand, proposing allows the proposing side to obtain higher match utility. Therefore we present the findings in two cases.
We assume subexponential rank distributions and small imbalance $|k| \leq n^{1-\epsilon}$ throughout. 

\paragraph{\bf Short side proposes.}
We have the following approximations of $U_S^\cP(d(n))$ and $U_L^\cP(d(n))$:
\begin{itemize}
    \item Short-side agents' average utility is
$
        \mathbb{E}[U_S^{\cP}(d(n))] \doteq \left( \max\left\{\left(\frac{n}{\sqrt{d(n)}}\right),\, \frac{n}{\log n}\right\} \right)^{\frac{1}{\alpha}} - n^{\gamma}\, 
 $.
    \item Long-side agents' average utility is
$
        \mathbb{E}[U_L^{\cP}(d(n))] \doteq \left( \min\left\{\sqrt{d(n)},\, \log n\right\} \right)^{\frac{1}{\alpha}} - d(n)^{\gamma}\, $.
    \item The average number of unmatched (short-side) agents is:
$
       \mathbb{E}[\delta^m(d(n))] \doteq n\cdot e^{-\sqrt{d(n)}}\, $.
\end{itemize}

\emph{We find that it is optimal to set $1 \dotleq d(n) \dotleq \log^2 n$.} The short-side agents' utility is non-increasing in $d(n)$ hence it is maximized at $d(n)\doteq 1$. The long-side agents' utility is non-increasing in $d(n)$ for $d(n) \dotgeq \log^2 n$, hence it is maximized between $\Theta(1)$ and $\Theta(\log^2 n)$. 
The number of unmatched agents is minimized for any $d(n)\dotgeq \log^2 n$. 
Therefore, the Pareto optimal $d(n)$ lie between $\Theta(1)$ and $\Theta(\log^2 n)$.
This finding suggests that a platform should deploy a short to moderate-length preference list length restriction ($1 \dotleq d(n) \dotleq \log^2 n$).


    %

\paragraph{Long side proposes.} Given consideration set size $d(n)$, we have the following approximations:
\begin{itemize}
    \item Short-side agents' average utility is
     $   \mathbb{E}[U_S^{\cP}(d(n))] \doteq \left( \max\left\{\sqrt{d(n)},\, \frac{d(n)}{\log n}\right\} \right)^{\frac{1}{\alpha}} - d(n)^{\gamma}\, $.
    %
    \item Long-side agents' average utility is $
        \mathbb{E}[U_L^{\cP}(d(n))] \doteq \left( \min\left\{\left(\frac{n}{\sqrt{d(n)}}\right),\, \left(\frac{n\log n}{d(n)}\right)\right\} \right)^{\frac{1}{\alpha}} - n^{\gamma}\,$.
    \item The number of unmatched (short-side) agents is
$
        \mathbb{E}[\delta^m(d(n))] \doteq n\cdot e^{-\sqrt{d(n)}}\, 
 $.
\end{itemize}

Two regimes arise from the planner's problem, but \emph{we find that it is optimal to set $1 \dotleq d(n) \dotleq \log^2 n$ in both cases.} 
The utility of long-side agents is maximized at $d(n)\doteq 1$ in both cases since it is decreasing in $d(n)$. Regarding the utility of short-side agents:
\begin{compactitem}[leftmargin=*]
    \item \emph{If the preference learning cost grows slowly in $d$, in particular if $\gamma < \frac{1}{\alpha}$}: The utility of short-side agents is maximized at $d(n)\doteq n$. However, the total welfare of all agents is maximized by $d(n)\doteq 1$.
    \item \emph{If the preference learning cost grows quickly in $d$, in particular if $\gamma \geq \frac{1}{\alpha}$}: The utility of short-side agents is maximized at $d(n)\doteq 1$ since the preference learning cost dominates the match value so smaller $d(n)$ is preferred.
\end{compactitem}
Therefore, the platform achieves maximum welfare by choosing $d(n)\doteq 1$. The number of unmatched agents is minimized for any $d(n)\dotgeq \log^2 n$. As a result, the Pareto optimal $d(n)$ lie between $\Theta(1)$ and $\Theta(\log^2 n)$, and the platform should set a preference length restriction within this range. The Pareto frontier is shown in Figure~\ref{fig:limiting-list-length}.

\begin{figure}[htb!]
    \centering
    \begin{minipage}{0.49\textwidth}
        \centering
        \includegraphics[width=0.95\textwidth]{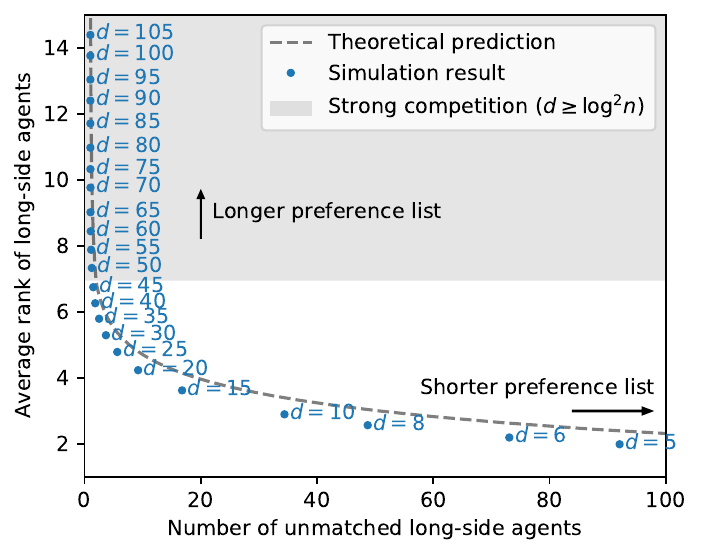}
    \end{minipage}\hfill
    \begin{minipage}{0.49\textwidth}
    \centering
    \includegraphics[width=0.95\textwidth]{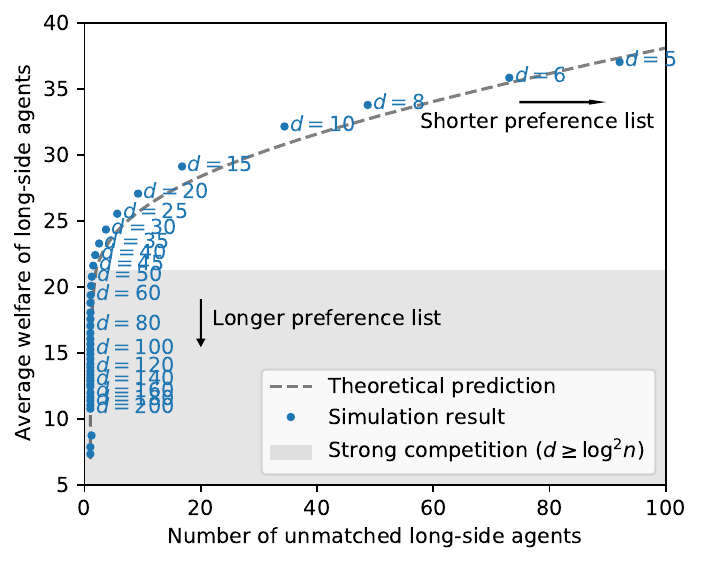}
    \end{minipage}
    \vspace{-0.1cm}
        \caption{
        	Illustration of effect of limiting preference list length on the matching outcomes when the long side proposes in a matching market with 1,000 men and 1,001 women ($|\mathcal{M}|=1000$, $|\mathcal{W}|=1001$).
	    In the left plot, a data point reports the number of unmatched long-side agents vs. the long-side agents' average rank, obtained through 1,000 runs of simulation at each preference list length value.
	    In the right plot, a data point report the number of unmatched long-side agents vs. the long-side agents' average welfare, where Pareto utility model ($\alpha=2$) is assumed (no preference discovery cost).
	    In both plots, dashed lines indicate our theoretical predictions, and shaded areas represent the strong competition regime (corresponding to $d \geq \log^2(1000/1) \approx 47.7$), in which increasing $d$ no longer improves the matching outcome.
	}
	\label{fig:limiting-list-length}
\end{figure}

\paragraph{Which side should propose.} In some applications, the platform can choose which side initiates contact. An example of this is the dating app Bumble, where only women are allowed to send the first message. Recent works have studied how to make this design decision \citep[see, e.g.,][]{kanoria2021facilitating}.
Our analysis above can also shed light on which side should propose (in our model, which is distinct from those in the existing literature). Assuming the preference learning cost grows quickly in $d$ (i.e., $\gamma> \frac{1}{\alpha}$, which may be typical in real world applications), we find that:

\begin{compactitem}[leftmargin=*]
\item \emph{If the cost of discovering preferences differs between the two sides, it is optimal for the side with the lower discovery cost to initiate contact.} This is because the side that initiates contact needs to screen (a lot) more potential partners than the other side. Specifically, if one side's preference learning cost has an exponent of $\lambda$ while the other side's costs have an exponent of $\gamma>\lambda$, it is optimal for the side with the exponent $\lambda$ to initiate contact. 
\item \emph{If a platform's primary concern is the welfare of a particular side of the market, it is optimal for that side to be the receiving side rather than the initiating side.} Similar to the case above, this is because the initiating side typically bears a higher preference discovery cost, as they need to screen more potential partners. By allowing the side that the platform is concerned with to be the receiving side, the platform can reduce the cost burden on that side and potentially improve their experience on the platform.
\end{compactitem}

We want to comment that the above messages align qualitatively with the findings obtained in \cite{kanoria2021facilitating} (in a different, decentralized search and matching setting). \cite{kanoria2021facilitating} suggests that the platform should let the short side initiate contact, because the long side agents need to screen more potential partners than the short side agents do if being on the initiating side, which lead to high screening costs and system inefficiency.

\color{black}

\section{Numerical Simulations for Random Matching Markets} \label{sec:numeric}
This section provides simulation results for random markets that confirm and sharpen the theoretical predictions made in Section~\ref{sec:result}.
Our simulations reveal that (i) our theoretical findings are valid for a wide range of market size $n$, average degree $d$, and imbalance $k$, (ii) the predicted sharp threshold between the strong and weak competition regimes at connectivity $d \approx 1.0 \times \log^2 (n/|k|)$ is confirmed to hold, and (iii) Principle~\ref{prin:condition-for-weak-competition} appears to be robust even to strong correlations in preferences (our quantitative are inaccurate under strong correlations, as expected).

\paragraph{Numerical verification of Theorem~\ref{thm:main-result} and \ref{thm:main-result-dense}.}
{
We first verify our main theoretical findings, Theorem~\ref{thm:main-result} and \ref{thm:main-result-dense}, which express four market statistics $R_\text{MEN}$, $R_\text{WOMEN}$, $\delta^m$ and $\delta^w$ in terms of market size $n$ and average degree $d$ when imbalance $k$ is small.}
Specifically, we consider a market with 1,000 men and 1,001 women ($n=1001$, $k=-1$) where the length of each man's preference list $d$ varies from 5 to 150.
For each degree $d$ we generate 1,000  realizations of random matching markets according to the generative model described in Section~\ref{sec:model}, and for each realization we compute the man optimal stable matching (MOSM) by running the men-proposing DA algorithm.

Figure \ref{fig:fixed-size} reports the men's average rank of wives and the women's average rank of husbands (left) and the number of unmatched men and women (right) at each $d$.
While not reported here to avoid cluttering the figures, we observe almost identical results for the WOSM.
Observe that when $d < \log^2n$ both men's average rank and women's average rank are highly concentrated at $\sqrt{d}$ and both the number of unmatched men and the number of unmatched women are close to $n e^{-\sqrt{d}}$, which confirms the estimates in Theorem \ref{thm:main-result}.
As $d$ grows past $\log^2 n$, the average rank of men and women start to deviate from each other, and specifically, the average rank of short side (men) stops increasing whereas the average rank of long side (women) increases linearly: i.e., $\Rmen \approx \log n$ and $\Rwomen \approx \frac{d}{\log n}$ when $d > \log^2n$, confirming Theorem \ref{thm:main-result-dense}.
We also remark that the expected number of unmatched men quickly vanishes as $d$ increases beyond $\log^2 n$ (note that the $y$-axis of the plot has a log-scale).

\begin{figure}[htb!]
    \centering
    \begin{minipage}{0.48\textwidth}
        \centering
        \includegraphics[width=0.95\textwidth]{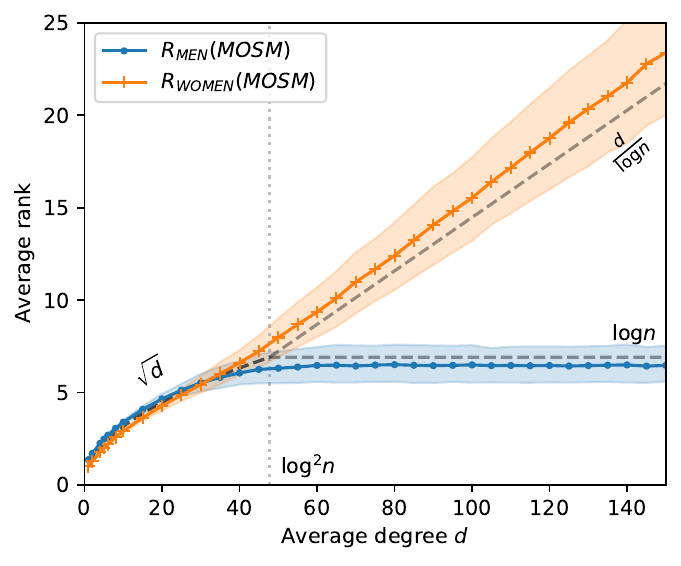}
    \end{minipage}\hfill
    \begin{minipage}{0.48\textwidth}
        \centering
        \includegraphics[width=0.95\textwidth]{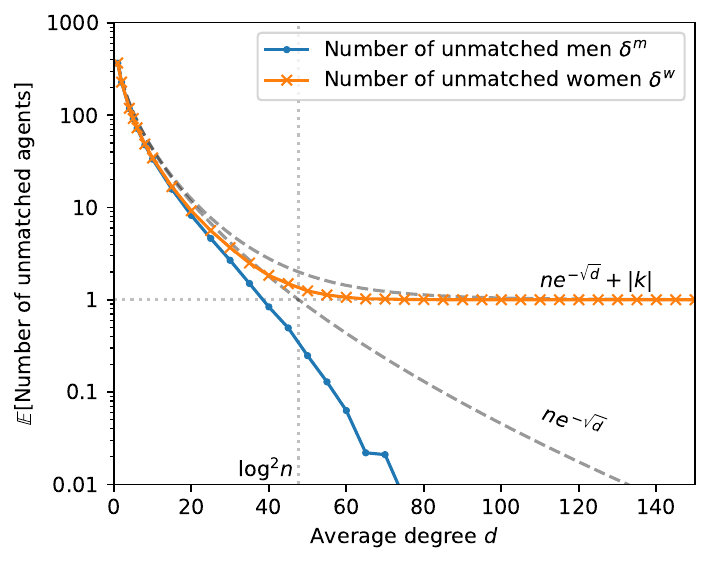}
    \end{minipage}
        \caption{
        	Men's average rank of wives $\Rmen$ and women's average rank of husbands $\Rwomen$ (left) and the number of unmatched men $\delta^m$ and the number of unmatched women $\delta^w$ (right) under the MOSM in random matching markets with 1,000 men and 1,001 women ($n=1001$, $k=-1$), and varying the length of men's preference lists $d$.
	In both figures, solid lines indicate the average value across 1,000 random realizations, and gray dashed lines indicate our theoretical predictions (Theorem \ref{thm:main-result} and \ref{thm:main-result-dense}) annotated with their expressions.
	In the left figure, the shaded areas surrounding solid lines represent the range between the top and bottom 10th percentiles of 1,000 realizations of men's and women's average rank. 
	}
	\label{fig:fixed-size}
\end{figure}

{In Appendix~\ref{app:numerical-additional}, we provide additional simulations results that confirm our quantitative predictions for unbalanced matching markets (e.g., $n=1050, k=-50$), heuristically developed in Section~\ref{subsec:detailed-heuristic-picture} and computed according to a procedure described in Appendix~\ref{app:heuristic-numerical}.
}

\paragraph{Numerical verification of connectivity threshold.}
The above observation extends to a wide range of market size $n$ (even for small $n \leq 50$).
To better illustrate, we investigate three kinds of threshold degree levels $d^*_\text{rank}(n,k)$, $d^*_\delta(n,k)$, and $d^*_\text{conn}(n,k)$ that sharply characterize the phase transitions that occur when degree $d$ varies in random matching markets of size $n$.
We define these thresholds as follows: given $n \geq 1$ and $k < 0$,
\begin{align}
	d^*_\text{rank}(n,k) &= \min_d\left\{ \mathbb{E}_{n,k,d}[\Wrank] / \mathbb{E}_{n,k,d}[ \Mrank ] \geq 1.15  \right\},
	\label{eq:threshold-rank-gap}
	\\
	d^*_\delta(n,k) &= \min_d\left\{ \mathbb{E}_{n,k,d}[ \delta^m ] \leq 0.5 \cdot |k| \right\},
	\label{eq:threshold-unmatched-men}
	\\
	d^*_\text{conn}(n,k) &= \min_d\left\{ \mathbb{E}_{n,k,d}[ \text{the number of connected components} ] \leq 2 \right\},
	\label{eq:threshold-connectivity}
\end{align}
where $\mathbb{E}_{n,k,d}[\cdot]$ represents the expected value of some random variable in a random matching market with $n+k$ men each of whose degree is $d$ and $n$ women.
The rank-gap threshold $d^*_\text{rank}(n,k)$ indicates the degree value beyond which men's average rank and women's average rank start to deviate from each other (in particular, we require a 15\% or larger difference in the average ranks on the two sides of the market); the unmatched-man threshold $d^*_\delta(n,k)$ is the degree value beyond which all men are (typically) matched; and the connected-component threshold $d^*_\text{conn}(n,k)$ is the degree value beyond which the entire market becomes a single giant connected component. Regarding the factor $0.5$ in the definition \eqref{eq:threshold-unmatched-men} of $d^*_\delta(n,k)$, recall that our statement of Principle~\ref{prin:condition-for-weak-competition} for random markets (see Section~\ref{subsec:detailed-heuristic-picture}) is an asymptotic one which holds for any constant factor in place of the $0.5$ here. We employ the factor $0.5$ because it provides the best numerical fit in finite random markets.

We consider two regimes: one with $k(n)=-1$ (fixed, small imbalance) and one with $k(n)=\sqrt{n}$ (varying, large imbalance).
For each regime, we quantify these threshold values based on numerical simulations: given $n$ and $k$, we use bisection method with a varying $d$ to find the threshold degrees, where the expected values are approximated with sample averages across 500 random realizations.

\begin{figure}[htb!]
    \centering
    \begin{minipage}{0.48\textwidth}
        \centering
        \includegraphics[width=0.98\textwidth]{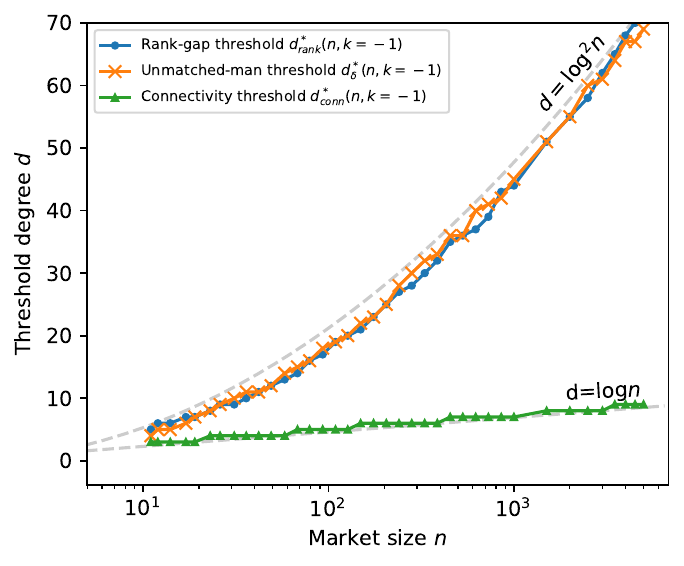}
    \end{minipage}\hfill
    \begin{minipage}{0.48\textwidth}
        \centering
        \includegraphics[width=0.98\textwidth]{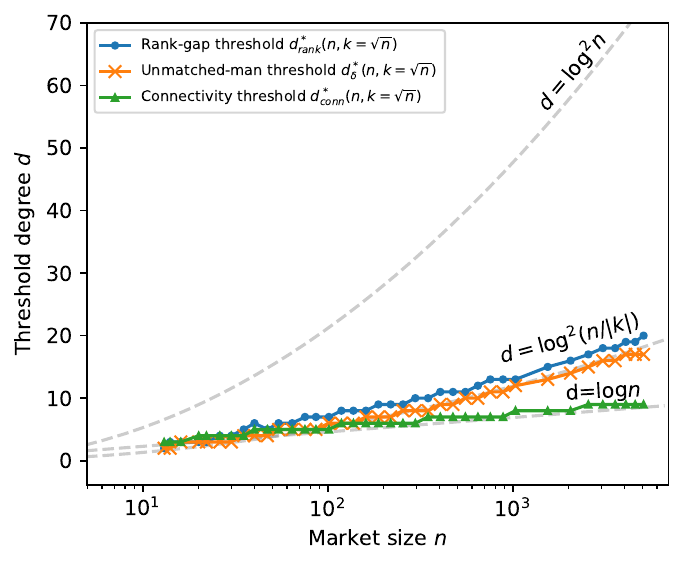}
    \end{minipage}
        \caption{
        Threshold degrees $d^*_\text{rank}(n,k(n))$, $d^*_\delta(n,k(n))$, and $d^*_\text{conn}(n,k(n))$, defined in \eqref{eq:threshold-rank-gap}--\eqref{eq:threshold-connectivity}, in random matching markets with $n+k$ men and $n$ women where $n$ ranges from 10 to 5,000, and $k(n)=-1$ (left) or $k(n)=\sqrt{n}$ (right). 
        For each $n$ and $k(n)$, the threshold values are found using bisection method in which we simulate 500 realizations at each attempted $d$.
        The gray dashed lines indicate the theoretical predictions annotated with their expressions.
        }
        \label{fig:threshold}
\end{figure}

Figure \ref{fig:threshold} plots the measured threshold degrees.
The thresholds $d^*_\text{rank}(n,k(n))$ and $d^*_\delta(n,k(n))$ are found to be close to $\log^2 (n/|k|)$ for all tested values of $n$. 
This suggests that our predicted threshold is fairly sharp and consistent with Principle~\ref{prin:condition-for-weak-competition}: the strong competition takes place i.f.f. the number of unmatched agents on the short side $\gtrsim$ the imbalance in the market  i.f.f. $d \gtrsim 1.0 \times \log^2 (n/|k|)$.
Also note that this threshold is much larger than the connected-component threshold $d^*_\text{conn}(n,k) \approx \log n$.

\paragraph{Behavior under correlated preferences.} 
Our theoretical analysis heavily relies on the assumption that the preferences are independent across agents, which may not be true in real-world matching markets.
We test robustness of our findings to correlated preferences by introducing a modified generative model, partly adopted from \cite{hitsch2010matching}.
After generating a connectivity graph as done in the previous experiments, we construct the preference lists according to the following random utility model: for each edge $(i,j)$ in the graph, the value agent $i$ obtains when matched with agent $j$ is given by (the value of $j$ for $i$ is symmetrically defined with independent $\epsilon_{ji}$)
\begin{equation}
    U_i(j) = \beta x_j + \epsilon_{ij},
    \label{eq:correlated-preferences}
\end{equation} 
where $\epsilon_{ij}$ is an idiosyncratic term independently drawn from a standard logistic distribution, $x_j \in [0,1]$ is a vertical quality of agent $j$ independently drawn from $U(0,1)$, and $\beta \geq 0$ is the sensitivity to quality component, which determines the level of correlation.
When $\beta = 0$, the model reduces to our main model in Section~\ref{sec:model}. When $\beta = \infty$, all agents share the same preference ranking over the other side of the market.

Figure~\ref{fig:robust-check-correlated-unbalanced} summarizes the simulation results for random matching markets with $\beta = 5$, $n=1050$ and $k=-50$. Note that $\beta = 5$ represents strong correlations where the quality term is 5 times larger than the idiosyncratic term.
The red vertical line represents the threshold degree $d_\delta^*(n,k)$, defined in \eqref{eq:threshold-unmatched-men}, which turns out to be $\approx 14$.
 Remarkably, we observe that $R_\text{MEN}$ and $R_\text{WOMEN}$ deviate from each other starting from this threshold $d_\delta^*(n,k)$, confirming Principle~\ref{prin:condition-for-weak-competition} (Principle~\ref{prin:condition-for-weak-competition} predicts whether a particular side of the market does better/worse than they would have in the corresponding balanced market. Since the setting here is symmetric across sides except for the market imbalance, we can instead compare average ranks across the two sides of the market.) As expected, our detailed quantitative predictions (shown via gray dashed lines) are inaccurate under such strong correlations.
See Appendix~\ref{app:numerical-additional} for additional numerical studies of {moderately correlated preferences ($\beta=1$) and} very strongly correlated preferences ($\beta=10$), and different market imbalances, which similarly confirm Principle~\ref{prin:condition-for-weak-competition} and moreover show that our quantitative predictions are surprisingly accurate under moderate correlations $\beta  = 1$.  Appendix~\ref{app:numerical-additional} also demonstrates robustness of our findings to a bipartite Erdos-Renyi connectivity graph (which exhibits heterogeneity in men's degree).

\begin{figure}[htb!]
    \centering
    \begin{minipage}{0.49\textwidth}
        \centering
        \includegraphics[width=0.95\textwidth]{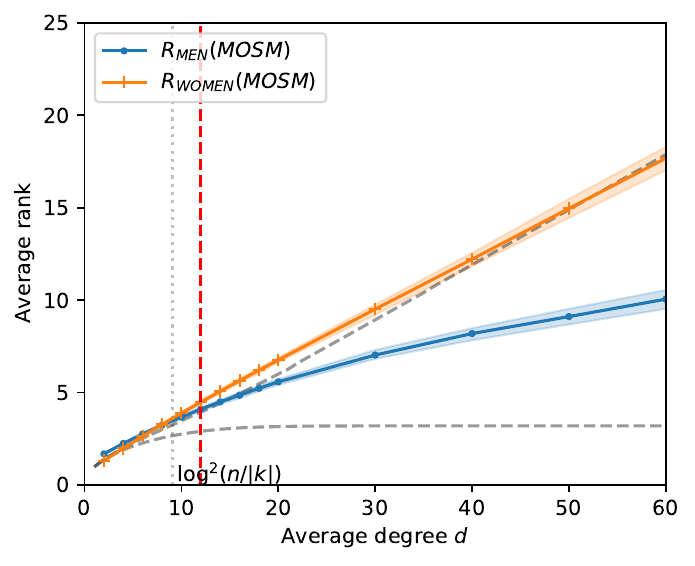}
    \end{minipage}\hfill
    \begin{minipage}{0.49\textwidth}
    \centering
    \includegraphics[width=0.95\textwidth]{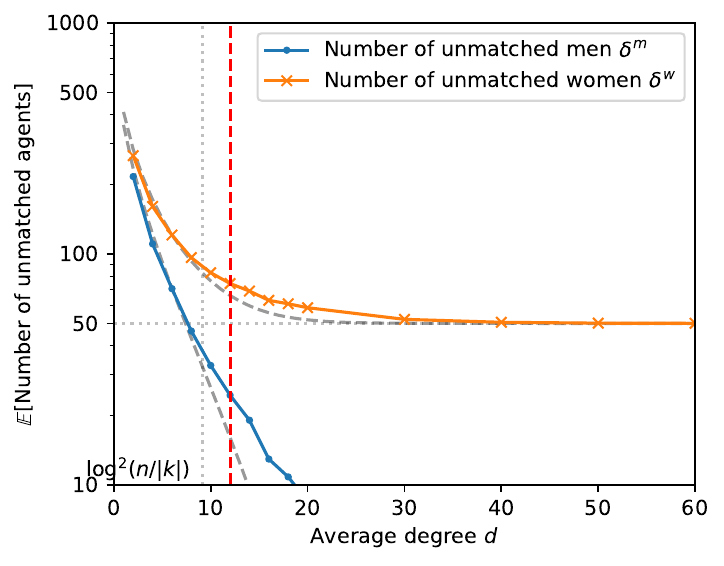}
    \end{minipage}
    \caption{
    The MOSM in random matching markets ($n=1050, k=-50$) with correlated preferences, 
    generated by a random utility model \eqref{eq:correlated-preferences} with $\beta=5$.
    Solid lines represent simulation results averaged across 100 runs of simulation, and gray dashed lines represent uncorrelated-market-based estimates obtained through the procedure described in Appendix~\ref{app:heuristic-numerical}. 
        }
    \label{fig:robust-check-correlated-unbalanced}
\end{figure}

\section{Counterfactual Analyses on High School Admissions Data}
\label{sec:NYC-HS-counterfactuals}

In this section, we perform counterfactual analyses on public high school admissions data from a major city in the U.S., and test our descriptive/prescriptive insights obtained from our theoretical investigation of random matching markets.
Our data-driven investigation shows that our theoretical findings are useful both to estimate the level of competition in a real-world matching market, and to provide design guidelines for improving market performance.

\paragraph{Data description.}
The admissions ``market'' in this city is cleared via a centralized deferred acceptance (DA)-based clearinghouse which collects preference data from applicants and priority rankings from schools. The data contains the preference lists over programs provided by nearly 75,000 applicants, 700 programs with a total capacity of 73,000, and the priorities of schools over applicants. The applicant preference lists have average length 6.92, median length 7, and maximum allowed length 12. Since DA-based mechanism used here is incentive compatible,\footnote{Incentive compatibility of DA for the proposing side was established by \cite{dubins1981machiavelli}. The only caveat for this particular market is that preference lists are not permitted to exceed a length of 12, however, only 16.6\% of students are found to reach this maximum list length, and 94.5\% of students are allocated to one of their top 5 programs.} the preference rankings collected may be assumed to reflect the true underlying preferences. We assume DA employs the widely used single-tie breaking approach to break ties between applicants within the same priority class, namely, the same uniformly random permutation over applicants is used by all programs. The algorithm is applicant-proposing deferred acceptance (DA).\footnote{We observe that program-proposing DA algorithm yields almost identical results: in our series of simulations, typically less than five students (and at most 20 students) are assigned to different programs under program-proposing DA.}
We find that about 6,000 seats remain unfilled in the resulting allocation, which is substantially larger than the market imbalance of about 2,000.

\paragraph{Descriptive insights: Testing the competition regime prediction coming from  Principle~\ref{prin:condition-for-weak-competition}.}
Applying Principle~\ref{prin:condition-for-weak-competition} to these summary statistics yields the prediction that the market is in the weak competition regime, i.e., that being on the long side should not be significantly hurting the rank of the allocation obtained by applicants.

To check this prediction and to study the effect of competition in the real market, we vary the market imbalance across a wide range by dropping up to 20,000 applicants from the data (uniformly at random) at one extreme, and duplicating up to 20,000 applicants (uniformly at random) at the other extreme, while holding the set of programs and their capacities fixed, and study the resulting change in the outcomes for applicants. 
As per the usual practice, we summarize the allocation in terms of the fraction of students who are allotted to one of their top-$k$ most preferred programs (for $k=1, 3$) and the fraction who are unassigned; see the solid lines in Figure~\ref{fig:NYC-HS-counterfactuals}. Observe that the quality of the allocation from the perspective of applicants is little changed (the aforementioned fractions each increase by less than 2\%) from the original under the counterfactual of a balanced market (which corresponds to randomly dropping 2000 applicants from the dataset). This confirms the correctness of the prediction from Principle~\ref{prin:condition-for-weak-competition} for this market. 

We now subject our principle to a broader test in this environment. For which counterfactual markets would Principle~\ref{prin:condition-for-weak-competition} have predicted that they exhibit strong competition, based on just the summary statistics? These turn out to be the markets where we add 3500 or more students, or subtract 8700 or more students (see the dashed gray vertical lines of Figure~\ref{fig:NYC-HS-counterfactuals}; here, the factor $0.5$ is chosen to be consistent with the definition of $d_\delta^*(n,k)$ \eqref{eq:threshold-unmatched-men} in Section~\ref{sec:numeric}). We see from the figure that for the markets predicted to exhibit strong competition, the fraction assigned to top choice or one of top-3 choices is $\gtrsim 4.5\%$ different from that under the balanced market, whereas for the markets predicted to exhibit weak competition, this difference is $\lesssim 4.5\%$. It is notable and very encouraging that nearly the same $4.5\%$ ``threshold'' for predicting strong competition emerges both when applicants are on the short side, and when applicants are on the long side, despite multiple asymmetries across the two sides of this market including its many-to-one nature. (For the uncorrelated version of the admissions market, described and tracked in the same figure, a similar phenomenon is observed, although for the markets predicted to exhibit strong competition the allocated fractions are now $\gtrsim 8\%$ different from those under the balanced market.) Overall, we find the quality and consistency of predictions coming from Principle~\ref{prin:condition-for-weak-competition} for this real-world market to be very encouraging, with the obvious caveats that the principle does not {quantify} the threshold beyond which it deems a short-side advantage to be significant, and that it remains to systematically incorporate the impact of correlations in preferences.

We now briefly discuss the National Residency Matching Program (NRMP) which matches medical residents to residency programs once a year, again via a centralized DA-based clearinghouse. Like many real-world clearinghouses, the NRMP has typically not allowed market-design researchers to access preference ranking data. However, the NRMP does publish an annual report each year, which includes the summary statistics needed to apply Principle~\ref{prin:condition-for-weak-competition}. For the 2022 match year, there were 42,549 active applicants vying for 39,205 residency positions, and 93.9\% positions were filled in the main match (i.e., about 2,390 positions went unfilled). 
Since the imbalance of 3,344 exceeded the number $2,390$ of unfilled positions but by a factor less than 2 (recall that we have been using a factor 2 in our application of Principle~\ref{prin:condition-for-weak-competition}), Principle~\ref{prin:condition-for-weak-competition} predicts that the market imbalance had only a moderate impact on match quality, somewhat benefiting programs and hurting applicants. 

\begin{figure}[htb!]
\begin{center}
	\NYC{
	}{
		\includegraphics[width=0.62\textwidth]{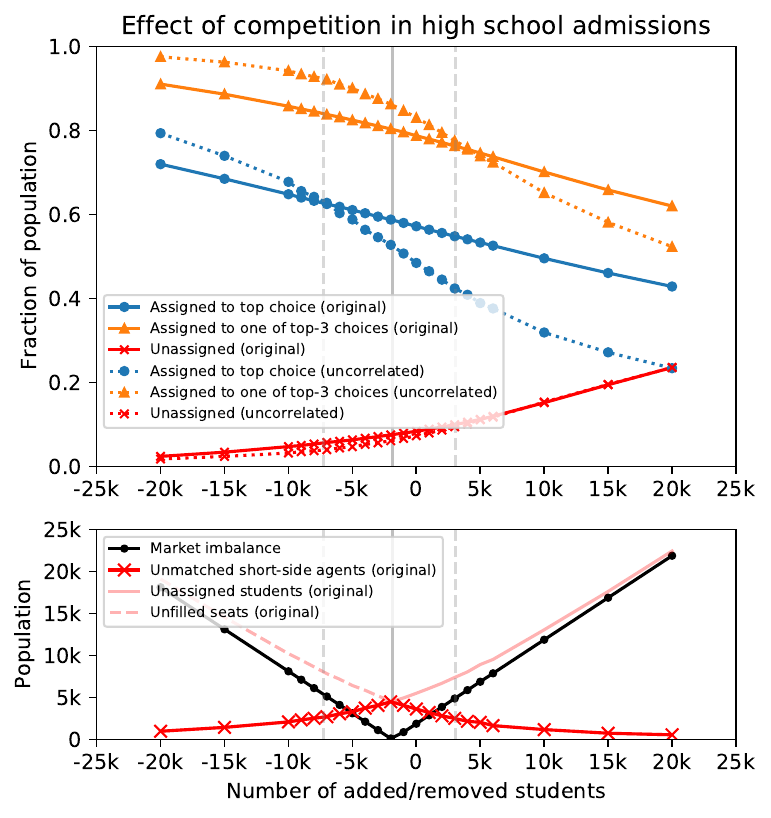}
	}
        \caption[Caption for HS figure]{
        Counterfactual analysis studying the effect of imbalance, based on \NYC{}{high school admissions data containing 75k applicants and 73k seats across 700 programs}.
        The top figure plots the quality of the assignment from the perspective of students, 
        as a function of the number of students removed or duplicated uniformly at random (averaged across 100 realizations).
        The solid colored curves use the student preference rankings and program priorities in the original data, and implement a single tie-breaking rule.
        The dotted colored curve are based on randomizing preferences and priorities: 
        Each student's preference list has unchanged length but its entries are drawn without replacement with the sampling probability of each program being proportional to the number of students who have applied to it in the original dataset, and each program uses a uniformly random and independent priority ordering over students.
        The bottom figure plots the market imbalance and the number of unmatched short-side agents (i.e., the minimum of the number of unassigned students and the number of unfilled seats).
        In both figures, the solid gray vertical line indicates the level at which the market becomes balanced, and the dashed gray vertical lines indicate the levels at which the number of unmatched short-side agent is half the market imbalance.
        }
        \label{fig:NYC-HS-counterfactuals}
\end{center}
\end{figure}



\paragraph{Design insights: How can the platform control competition over popular programs?}
Let us now restrict our attention to a submarket consisting of ``popular'' programs and the students who applied to those popular programs, in which students compete strongly for scarce seats.
Our analysis in Section~\ref{subsec:limiting-list-length} suggest that we can alleviate such strong competition among applicants by reducing the connectivity of the market, i.e., by limiting the number of popular programs that a student can report in their preference list.
We run simulations to estimate the effect of this intervention within the submarket, and observe that we can indeed increase the students' welfare substantially while keeping almost all popular seats filled.

Following the work of \cite{ashlagi2020matters}, we define \emph{popularity} $\alpha_c$ of program $c$ as the ratio between the number of students who had listed program $c$ on top of their preference lists and the capacity of program $c$.
In the actual data set, the popularity value ranges from $0.10$ to $12.0$ across 700 programs, e.g., the most popular program ($\alpha_c=12.0$) has one available seat while it is the most preferred one for 12 students out of 75,000 students.

We conduct a stylized study on a submarket that only includes ``popular'' programs satisfying $\alpha_c \geq 0.8$, which consists of nearly 300 programs with a total capacity of 33,000 seats. 
The preference lists of the students are refined accordingly, i.e., all the non-popular programs are removed from the preference lists so that they only contain these popular programs. 
The students with empty preference lists are removed from the submarket.
As a result, the submarket includes around 71,000 students with average preference list length 4.43, median length 4, and maximum length 12.

When running the DA algorithm on this submarket \emph{without any intervention}, we observe that only 3.8 seats ($\approx 0.01\%$ of total capacity) are left unfilled on average, 38,000 students are left unassigned accordingly (in the actual matching process over the entire market, most of them will be assigned to one of non-popular programs), and the assigned students' average rank of their matched (popular) program is 1.62. We find that the number of unfilled seats (i.e., the number of unmatched agents on the short side) is much smaller than the market imbalance, and applying Principle~\ref{prin:condition-for-weak-competition}, we can expect that the students are suffering from the extremely strong competition in this submarket.
Adopting a Pareto utility model with {$\alpha=1.2$}
 (see Section~\ref{sec:design-implications}), the average welfare of the assigned students is found to be  70.3.\footnote{In detail, the (expected) welfare of a student who get assigned to $r$-ranked program is computed using the $r$-th largest order statistics out of $m$ i.i.d. Pareto random samples where {$m$ is the number of popular programs that are accessible for the student. We set $m=30$ assuming that 10\% of popular programs are accessible for each student due to geographic considerations. The choice of $\alpha=1.2$ was loosely guided by the discussion in \cite[][Appendix A]{ibragimov2010optimal}; welfare gains were similar for $\alpha \in [1.1, 1.6]$, with larger improvement in welfare for $\alpha$ closer to 1.}
In the current setting, the students assigned to their most preferred program get average utility 95.0, and those who are assigned to their second most preferred program get average utility 15.8.} 

We now consider the platform intervention of imposing a limit $d_\text{max}$ on the maximum number of popular programs that each student can list on his/her preference list.
We investigate the consequence of this intervention by simulating the matching process after truncating all preference lists by the given length limit $d_\text{max}$, where we vary $d_\text{max}$ from 12 to 1 ($d_\text{max}=12$ corresponds to no intervention).
Our counterfactual analysis assumes that the students will report their preference lists truthfully under the intervention, which may be violated in the case of students who respond strategically to the restriction on the number of popular programs they are allowed to list. Strategic agent responses should only increase the welfare gains resulting from our intervention \citep{abdulkadirouglu2017welfare,che2019efficiency}, so our improvement estimates can be viewed as lower bounds.

Figure \ref{fig:NYC-HS-intervention} summarizes the simulation results, visualized analogously to Figure \ref{fig:limiting-list-length}.
When $d_\text{max}=1$, since every assigned student gets allocated to his/her most preferred program, their average rank becomes 1.0, and their average welfare attains its maximal value 95.0 ($\approx 35\%$ improvement compared to no-intervention).
This intervention incurs 59.0 (out of 33k) unfilled (popular) seats on average as a side effect.
When $d_\text{max}=2$, we obtain an average rank of 1.18 and the average welfare 80.3 ($\approx 14\%$ improvement) while keeping almost all seats filled (7.70 seats are left unfilled on average). This improvement is driven is large part by an increase in the fraction of seats assigned to students who listed the program as their top choice, from $70.4\%$ to $81.5\%$. 
We observe that there is no benefit from deploying $d_\text{max} \geq 3$, since it only degrades the match value while the reduction in the number of unfilled seats relative to $d_\text{max}=2$ is very marginal.

\begin{figure}[htb!]
    \centering
    \begin{minipage}{0.49\textwidth}
        \centering
        \includegraphics[width=0.95\textwidth]{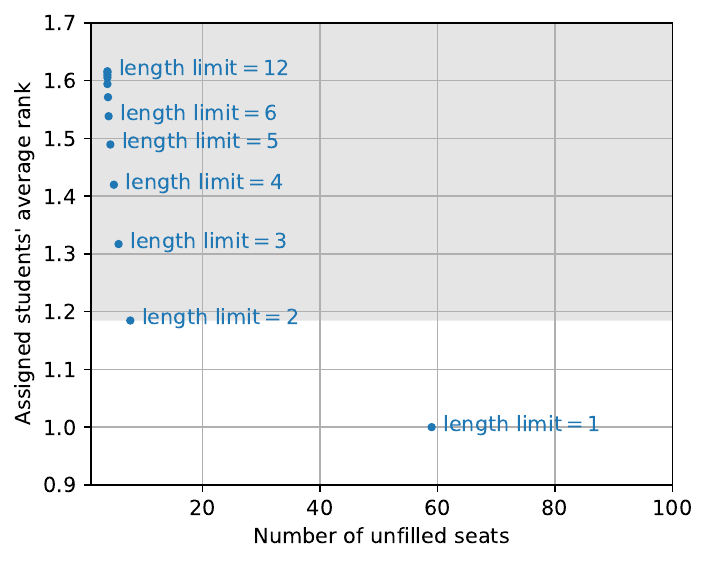}
    \end{minipage}\hfill
    \begin{minipage}{0.49\textwidth}
        \centering
        \includegraphics[width=0.95\textwidth]{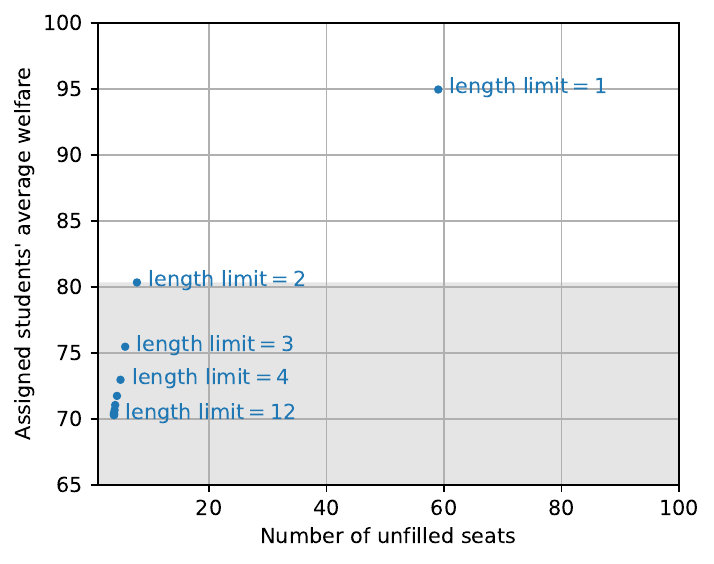}
    \end{minipage}
    \caption{
        Counterfactual analysis studying the effect of limiting applicant preference list lengths using high school admissions data, in the submarket of 300 ``popular'' programs with 33k seats and 71k students who applied to those programs.
        The figure shows 
        the expected number of unfilled seats and the students' average rank and match welfare (for assigned students only, under a Pareto utility model with $\alpha=1.2$) for preference length limits $d_\text{max} \in \{1,\ldots,12\}$
    }
    \label{fig:NYC-HS-intervention}
\end{figure}

This experiment demonstrates that we can make a Pareto improvement by allowing the students to include at most two popular programs in their preference lists (or we may let the students to include only one popular program if we are willing to let a few popular program seats remain unfilled).
This is consistent with our design insights obtained in Section~\ref{sec:design-implications}: the level of competition among long-side agents can be controlled by adjusting the connectivity of market, and the optimal level connectivity lies in the weak competition regime.

\section{Overview of the Proof of Theorem~\ref{thm:main-result}}\label{sec:proof-sketch}
This section provides an overview of the proof of Theorem \ref{thm:main-result}, which is our characterization of moderately connected random matching markets. 
Our proof uses the well-known analogy between DA and the coupon collector problem to bound women's average rank of their husbands, but also encounters and tackles the challenge of tracking the (strictly positive) number of men who have reached the bottom of their preference lists by constructing a novel bound using a tractable stochastic process. The latter challenge did not arise in the setting of \cite{AKL17} where \emph{all} short side agents are matched under stable matching, and similarly doesn't arise in our ``densely connected markets'' setting (Theorem~\ref{thm:main-result-dense}). Following \cite{AKL17} and the majority of other theoretical papers on matching markets, we prove our characterizations for large $n$ (and then use numerics to demonstrate that they extend to small $n$; see Section~\ref{sec:numeric}). Alongside an overview of the proof this section provides parenthetical pointers to the relevant formal lemmas; their statements and proofs can be found in Appendix~\ref{append:small-medium-d-proof}.

Our analysis tracks the progress of the following McVitie-Wilson \citep{mcvitie1971stable} (sequential proposals) version of the men-proposing Deferred Acceptance algorithm that outputs MOSM (the final outcome is known to be the MOSM, independent of the sequence in which proposals are made). Under this algorithm, only one man proposes at a time, and ``rejection chains'' are run to completion before the next man is allowed to make his first proposal. The algorithm takes the preference rankings of the agents as its input.
\begin{alg}[Man-proposing Deferred Acceptance]
	\label{alg:2MPDA}
	Initialize ``men who have entered'' $\hcM \leftarrow \phi$, unmatched women $\bar{\cW} \leftarrow \cW$, the number of proposals $t \leftarrow 0$, the number of unmatched men $\delta^m \leftarrow 0$. 
 \begin{compactenum}[1.,wide,labelwidth=!,labelindent=0pt]
		\item \label{Step:2New-Man}  If $\cM\backslash \hcM$ is empty then terminate. Else, let $i$ be the man with the smallest index in $\cM\backslash \hcM$. Add $i$ to $\hcM$.
		\item \label{Step:2Man-proposes} If man $i$ has not reached the end of his preference list, do $t \leftarrow t+1$ and man $i$ proposes to his most preferred woman $j$ whom he has not yet proposed. 
If he is at the end of his list, do $\delta^m \leftarrow \delta^m + 1$ go to Step \ref{Step:2New-Man}.
		\item Decision of $j$:
\begin{compactenum}[(a),wide,labelwidth=!,labelindent=0pt]		
			\item \label{Step:2Proposal_to_unmatched} If $j \in \bar{\cW}$, i.e., $j$ is currently unmatched, then she accepts $i$. Remove $j$ from $\bar{\cW}$. Go to Step \ref{Step:2New-Man}.
			\item \label{Step:2Proposal_to_matched} If $j$ is currently matched, she accepts the better of her current partner and $i$, and rejects the other. Set $i$ to be the rejected man 
and continue at Step \ref{Step:2Man-proposes}.
  \end{compactenum}
\end{compactenum}
\end{alg}

\paragraph{Principle of deferred decisions.} As we are interested in the behavior of Algorithm~\ref{alg:2MPDA} on a random matching market, we
think of the deterministic algorithm on a random input as a randomized algorithm, which is easier to analyze. The randomized, or coin flipping, version of the algorithm does not receive preferences as input, but draws them through the process of the algorithm. This is
often called the \emph{principle of deferred decisions}. The algorithm reads the next woman in the preference of a man in step \ref{Step:2Man-proposes} and whether a woman prefers a man over her current proposal in step \ref{Step:2Proposal_to_matched}. No man applies twice to the same
woman during the algorithm, and therefore the algorithm never reads previously revealed
preferences. In step \ref{Step:2Man-proposes} the randomized algorithm selects the woman $j$ uniformly at random
from those to whom man $i$ has not yet proposed. In step \ref{Step:2Proposal_to_matched}, the probability that $j$ prefers $i$ over her current match is $1/(\pop(j)+1)$ where $\pop(j)$ is the number of proposals previously received by woman $j$.

\paragraph{Stopping time.} Algorithm~\ref{alg:2MPDA}  defines that ``time'' $t$ ticks whenever a man makes a proposal.
First observe that the {current} number of unmatched men $\delta^m[t] = \delta^m$ at time $t$, i.e., men who have reached the bottom of their lists and are still unmatched, is \emph{non-decreasing} over time, whereas the {current} number of unmatched women $\delta^w[t] = |\bar{\cW}|$ at time $t$, i.e., women who have yet to receive their first proposal, is \emph{non-increasing} over time. The MOSM is found when the number of unmatched men exactly equals the number of unmatched women plus $k$.
We view this total number of proposals $\tau$ when DA terminates as a stopping time:
\begin{align}
	\tau = \min\{t\geq 1: \delta^m[t] = \delta^w[t] + k \}\, .
\label{eq:tau}
\end{align}
This total number of proposals $\tau$ serves as a key quantity enabling our formal characterization of the MOSM {(see Figure \ref{fig:sample-path} for an illustration)}. 
\begin{figure}[htbp!]
	\begin{center}
		\includegraphics[width=0.6\textwidth]{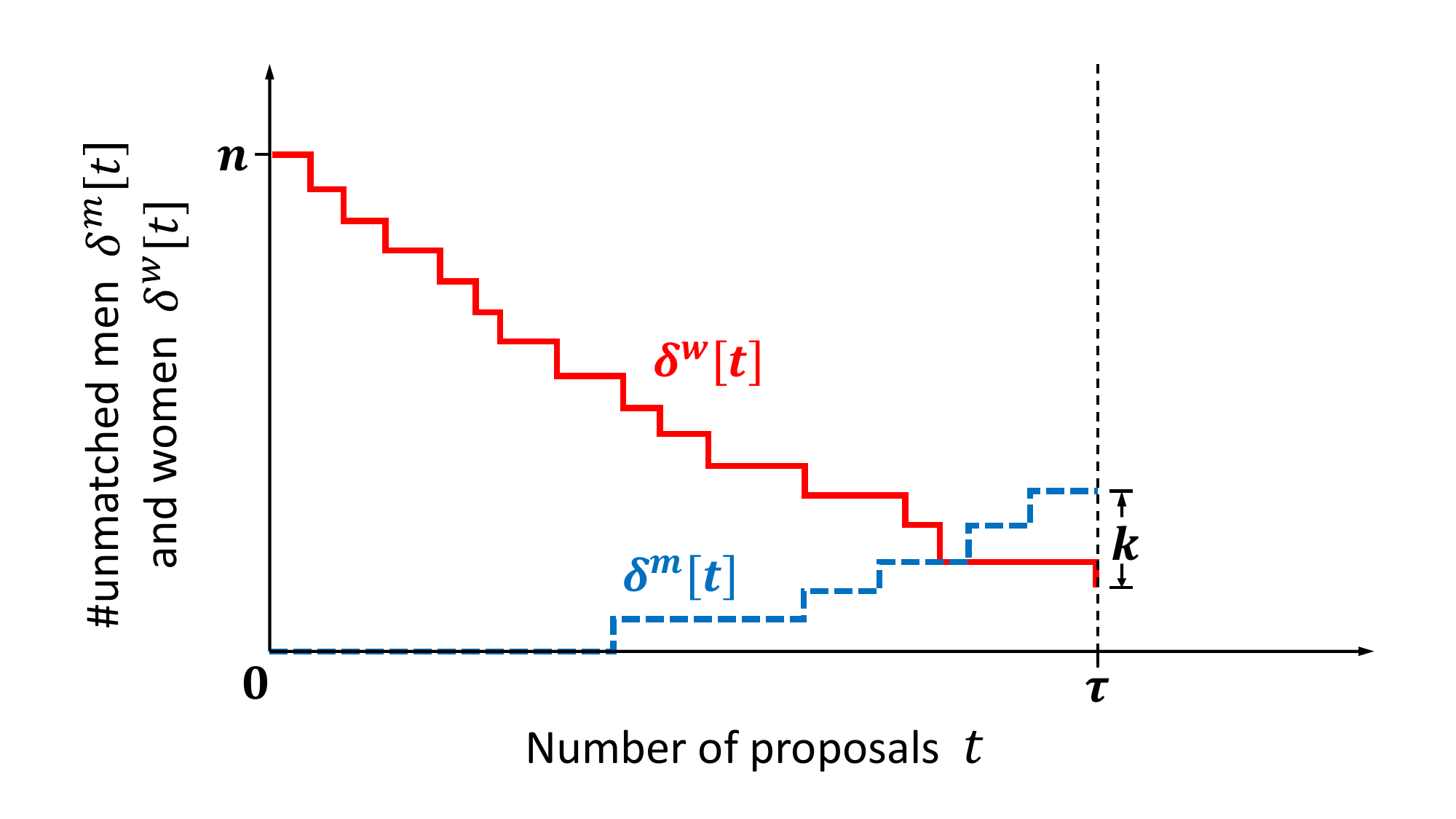}
		\caption{
			Illustration of a sample path of the current number of unmatched men $\delta^m[t]$ and unmatched women $\delta^w[t]$ under Man-proposing Deferred Acceptance (Algorithm~\ref{alg:2MPDA}).
			The algorithm terminates at $t=\tau$, the first time $\delta^m[t]=\delta^w[t]+k$. (In this illustration $k>0$).
		}
		\label{fig:sample-path}
	\end{center}
\end{figure}
On the men's side, the sum of men's rank of wives is approximately the total number of proposals $\tau$ (more precisely, this sum is $\tau + \delta^m[\tau]$ given that the rank for an unmatched agent is defined as one more than the length of the agent's preference list, but $\tau \gg \delta^m[\tau]$ is the dominant term).
On women's side, since each proposal goes approximately to a uniformly random woman, as a function of the total number of proposals we can tightly control the distribution of the number of proposals received by individual women (this distribution is close to Poisson and tightly concentrates around its average) and therefore their average rank of husbands (Propositions~\ref{prop:women-ranking-lower-bound} and \ref{prop:women-ranking-upper-bound}), as well as the number of unmatched women (Propositions~\ref{prop:unmatched-women-lower-bound} and\footnote{In Proposition \ref{prop:total-proposal-lower-bound}, we first upper bound the number of unmatched women, and then use the aforementioned observation to lower bound the number of proposals.} \ref{prop:total-proposal-lower-bound}). 

Therefore, the bulk of the proof of Theorem \ref{thm:main-result} is dedicated to bounding the total number of proposals $\tau$. Because of the aforementioned technical challenge that a positive number of agents remain unmatched on both sides, a direct application of the coupon collector analogy is not enough. Instead, we control the  two stochastic processes that track the current number of unmatched men $\delta^m[t]$ and unmatched women $\delta^w[t]$ at each time $t$ and make use of the identity \eqref{eq:tau} that $\delta^m[\tau] = \delta^w[\tau] + k$.
(Upon termination, the number of unmatched men must be $k$ plus the number of unmatched women.)
For technical purposes, we extend the definition of $\delta^{m}[t]$ and $\delta^{w}[t]$ to $t > \tau$ as follows: if there are no men waiting to propose (i.e., a stable matching has been found), we introduce a fake man who is connected to $d$  women (uniformly and independently drawn) with a uniformly random preference ranking over them, and keep running Algorithm \ref{alg:2MPDA}.

\paragraph{Upper bound on the total number of proposals.}
We show (in Proposition~\ref{prop:total-proposal-upper-bound}) that the total number of proposals cannot be too large, i.e., $\tau \leq (1+\epsilon)n\sqrt{d}$  with high probability for $\epsilon=d^{-1/4}=o(1)$.
We establish this bound by showing that after a large enough number of proposals have been made, i.e., at time $t = (1+\epsilon)n\sqrt{d}$, the current number of unmatched women $\delta^w[t]$ has (with high probability) dropped  below $n e^{-\sqrt{d}}$ 
whereas the current number of unmatched men $\delta^m[t]$ has (with high probability) increased above some level which is $\omega(n e^{-\sqrt{d}})
$ 
and hence, since $k = O(n e^{-\sqrt{d}})$, the stopping event ($\delta^m[\tau] = \delta^w[\tau]+k$) must have happened earlier, i.e., $\tau \leq (1+\epsilon)n\sqrt{d}$.
The upper bound on $\delta^w[(1+\epsilon)n\sqrt{d}]$ (see Lemma~\ref{lem:unmatched-woman-upper-bound}) is derived using a standard approach that utilizes the analogy to the coupon collector problem.
The lower bound on $\delta^m[(1+\epsilon)n\sqrt{d}]$ (see Lemma~\ref{lem:unmatched-man-lower-bound}) is obtained by counting the number of occurrences of  $d$-rejections-in-a-row during the men-proposing DA procedure (whenever rejections take place $d$ times in a row, at least one man becomes unmatched). Thus, our lower bound on $\delta^m[(1+\epsilon)n\sqrt{d}]$ ignores that some men are first accepted, and then later rejected causing them to reach the end of their preference lists via less than $d$ consecutive rejections. Our conservative approach provides tractability and saves us from needing to track how far down their preference lists the currently matched men are. Nevertheless, the slack in this step necessitates our stronger assumption $d = o(\log^2 n)$, despite our conjecture that the characterization extends for all $d < 0.99 \log^2 n$.

\paragraph{Lower bound on the total number of proposals.}
We prove (in Proposition~\ref{prop:total-proposal-lower-bound}) that the total number of proposals cannot be too small, i.e., $\tau \geq (1-\epsilon)n\sqrt{d}$ with high probability for some $\epsilon=o(1)$.
We start with upper bounding (in Lemma~\ref{lem:unmatched-women-expected-ub}) the expected number of unmatched men in the stable matching, $\mathbb{E}[ \delta^m ]$, by showing that the probability of the last proposing man being rejected cannot be too large given that each woman has received at most $(1+\epsilon) \sqrt{d}$ proposals on average (recall that $\tau \leq (1+\epsilon)n\sqrt{d}$ w.h.p.).
We then use Markov's inequality to derive an upper bound on $\delta^m$ which holds with high probability, 
and deduce (in Proposition \ref{prop:unmatched-women-upper-bound}) an upper bound on $\delta^w$ using the identity $\delta^m = \delta^w + k$.
Then we again use the coupon collector analogy to bound $\tau$ from below: the process cannot stop too early since the current number of unmatched women $\delta^w[t]$ does not decay fast enough to satisfy the upper bound on $\delta^w[\tau]$ ($=\delta^w$) if $\tau$ is too small.

\section{Conclusion} \label{sec:discussion}
Our investigation into stable matchings in two-sided markets with heterogeneous agent preferences has provided insights into the impact of market connectivity on equilibrium outcomes. By introducing a model of partially-connected random matching markets and developing new technical tools, we are able to explicitly characterize these outcomes. Based on the theoretical results, we have  derived guidance on the design of platform interventions which control market connectivity. Our theoretical analysis has moreover led us to a more general principle governing whether being on the short side confers a significant advantage in a given matching market, which can be applied based on market summary statistics alone. While we recognize that our model is highly stylized and may not capture all features of real-world markets, our numerical experiments, which included features such as many-to-one matching, correlation in preferences, and degree heterogeneity, have supported our findings and conjectures. We believe this study represents a crucial step towards a more complete understanding of competition in general matching markets, and we leave as interesting and challenging directions for future work to further investigate more general settings. In particular, it would be of interest to  develop and refine the conceptual principle we introduce along the lines of
methods developed in public and policy economics for facilitating counterfactual analysis using summary statistics alone \citep{chetty2009sufficient}.

\setstretch{1.0}

\bibliographystyle{abbrvnat}
\bibliography{matchingmarket,nostark}

%

\newpage


\setcounter{page}{1}
\appendix

\begin{center}
	{\LARGE {\textbf{Appendix to ``The Competition for Partners in Matching Markets''}}}
\end{center}
\smallskip

\textbf{Organization of the appendix.} The technical appendix is organized as follows.
\begin{itemize}
\item Appendix~\ref{app:buyer-seller-app} elaborates on our assertion that almost all buyer-seller markets exhibit strong competition, in contrast with our findings for random matching markets.
	\item Appendix \ref{append:preliminary} describes several concentration inequalities and auxiliary stochastic processes that will be heavily used in the following theoretical analysis.
	\item Appendix \ref{append:small-medium-d-proof} establishes Theorem \ref{thm:main-result}, the main result for moderately connected markets. The proof is lengthy and will be further divided into several steps, with an overview provided at the beginning of each step.
	\item Appendix \ref{append:dense-d} establishes Theorem \ref{thm:main-result-dense}, the main result for densely connected markets.
	
	\item Appendix \ref{append:additional-proof} provides the proofs that complement Section~\ref{sec:design-implications}.
	
	\item Appendix \ref{app:heuristic-numerical} refines the heuristic characterization of random matching markets made in Section~\ref{subsec:detailed-heuristic-picture}, describes Principle~\ref{prin:condition-for-weak-competition} for random markets, and introduces a numerical procedure to obtain refined estimates.
	
	\item Appendix \ref{app:numerical-additional} provides additional numerical simulation results that complement Section~\ref{sec:numeric}.
\end{itemize}

\section{Almost all buyer-seller markets exhibit strong competition}
\label{app:buyer-seller-app}

Consider a buyer-seller market where each of $n+k$ sellers is selling one unit of the same commodity, and each of $n$ buyers wants to buy one unit and has value $1$ for a unit. A bipartite graph $G$ with sellers on one side and buyers on the other captures which trades are feasible. (This is a special case of the Shapley-Shubik assignment model \citep{shapley1971assignment}.) We say that an unbalanced market with $k > 0$ (or $k < 0$) exhibits a stark effect of competition if, in any equilibrium, all trades occur at price $0$ (or $1$), i.e., the agents on the short side, namely buyers (sellers), capture all the surplus.
Then we know \citep{shapley1971assignment} that for $k \neq 0$ the market exhibits a stark effect of competition if the following requirement is satisfied: \begin{align*}
 {\Ev} \equiv \{ & \, \textup{For each agent $j$ on the long side, there exists a matching in $G$}\\
 &\textup{where all short side agents are matched but agent $j$ is unmatched}\, \} \, .
 \end{align*} 
Requirement $\Ev$ is only slightly stronger than connectivity of $G$:
Suppose, as in our model in Section~\ref{sec:model}, that each seller is connected to a uniformly random subset of $d$ buyers.
Under this stochastic model for $G$, for any sequence of $k$ such that $1 \leq |k|= O(1)$, event $\Ev$ occurs (i.e., there is a stark effect of competition) for all $d$ exceeding the connectivity threshold at $d = \log n$:
\begin{enumerate}
\item[(i)] For any $\eps > 0$ and $d \geq (1+\epsilon) \log n$, with high probability, $G$ is connected  and moreover, event $\Ev$ occurs, i.e., there is a stark effect of competition.
\item[(ii)] For any $\eps > 0$ and $d \leq (1-\epsilon) \log n$, with high probability, the connectivity graph $G$ is disconnected (in fact a positive number of buyers have degree zero).
\end{enumerate}


\section{Preliminaries}\label{append:preliminary}
\subsection{Basic Inequalities}
\begin{lem}\label{lem:basic-inequalities}
The following inequalities hold:
\begin{itemize}
\item For any $|x|\leq \frac{1}{2}$, we have $e^{-x-x^2}\leq 1 - x \leq e^{-x}$.
\item For any $k>0$ and $\epsilon\in\left(0,\frac{1}{k}\right)$, we have
$1 + k\epsilon \leq \frac{1}{1 - k\epsilon}$.
\end{itemize}
\end{lem}

\subsection{Negative Association of Random Variables}

The concept of negative association provides a stronger notion of negative correlation, which is useful to analyze the concentration of the sum of dependent random variables.

\begin{defn}[Negatively Associated Random Variables \citep{dubhashi1998balls}]
	A set of random variables $X_1, X_2, \ldots, X_n$ are \emph{negatively associated} (NA) if for any two disjoint index sets $I, J \subseteq \{1,\ldots,n\}$,
	\begin{equation*}
		\mathbb{E}\left[ f( X_i : i \in I ) \cdot g( X_j : j \in J ) \right] \leq \mathbb{E}\left[ f( X_i : i \in I ) \right] \cdot \mathbb{E}\left[ g( X_j : j \in J ) \right]
	\end{equation*}
	for any two functions $f:\mathbb{R}^{|I|} \mapsto \mathbb{R}$ and $g:\mathbb{R}^{|J|} \mapsto \mathbb{R}$ that are both non-decreasing or both non-increasing (in each argument).
\end{defn}

The following lemma formalizes that the sum of negatively associated (NA) random variables is as concentrated as the sum of independent random variables:

\begin{lem}[Chernoff-Hoeffding Bound for Negatively Associated Random Variables \citep{dubhashi1998balls}] \label{lem:Hoeffding-NA}
	Let $X_1, X_2, \ldots, X_n$ be NA random variables with $X_i \in [a_i,b_i]$ always.
	Then, $S \defeq \sum_{i=1}^n X_i$ satisfies the following tail bound:
	\begin{equation}
		\mathbb{P}\left( | S - \mathbb{E}[S] | \geq t \right) \leq 2 \exp\left( - \frac{2t^2}{ \sum_{i=1}^n (b_i - a_i)^2 } \right).
	\end{equation}
\end{lem}

We refer to \cite{dubhashi1998balls} for the proof.

The following lemma provides sufficient conditions for a set of random variables to be NA. For each sufficient condition, we provide a pointer to a paper where it has been established.

\begin{lem}[Sufficient Conditions for Negative Association] \label{lem:NA}
	The followings hold:
	\begin{enumerate}[(i)]
		\item\label{lem:NA-permutation} \emph{(Permutation distribution \cite[Theorem 2.11]{Joan83})}
		Let $x_1, x_2, \ldots, x_n$ be $n$ real numbers and let $X_1, X_2, \ldots, X_n$ be random variables such that $(X_1, X_2, \ldots, X_n)$ {is a uniformly random permutation of $(x_1, x_2, \ldots, x_n)$.} 
		Then $X_1, X_2, \ldots, X_n$ are NA.
		
		\item\label{lem:NA-independent} \emph{(Union of independent sets of NA random variables \cite[Property 7]{Joan83})}
		If $X_1, X_2, \ldots, X_n$ are NA, $Y_1, Y_2, \ldots, Y_m$ are NA, and $\{X_i\}_i$ are independent of $\{Y_j\}_j$, then $X_1, \ldots, X_n, Y_1, \ldots, Y_m$ are NA.
		
		\item\label{lem:NA-monotone} \emph{(Concordant monotone functions \cite[Property 6]{Joan83})}
		Increasing functions defined on disjoint subsets of a set of NA random variables are NA.
		More precisely, suppose $f_1, f_2,\ldots, f_k$ are all non-decreasing {in each coordinate}, or all non-increasing {in each coordinate}, with each $f_j : \mathbb{R}^{|I_j|} \mapsto \mathbb{R}$  defined on $(X_i)_{i \in I_j}$ for some disjoint index subsets $I_1, \ldots, I_k \subseteq \{1,\ldots,n\}$.
		If $X_1, X_2, \ldots, X_n$ are NA, then the set of random variables $Y_1 \triangleq f_1(X_i:i \in I_1), Y_2 \triangleq f_2(X_i : i \in I_2), \ldots, Y_k \triangleq f_k(X_i : i \in I_k)$ are NA.
	\end{enumerate}
\end{lem}

{
\subsection{Balls-into-bins}\label{subsec:balls-into-bins}
A balls-into-bins process with $T$ balls and $n$ bins is defined as follows: at each time $t=1,\ldots,T$, a ball is placed into one of $n$ bins uniformly at random, independently of the past.
Index the bins by $j \in \{1, \ldots, n\}$, and let $I_{j,t}$ be an indicator variable that equals one if the $t^\text{th}$ ball is placed in the $j^\text{th}$ bin and equals zero otherwise.
Further define $W_j \defeq \sum_{t=1}^T I_{j,t}$ representing the total number of balls placed into the $j^\text{th}$ bin.

A particular random variable of interest is the number of empty bins at the end of the process, for which we have the following concentration inequality.
\begin{lem}[Number of empty bins] \label{lem:balls-into-bins-uniform}
	Let $X$ be the number of empty bins at the end of a balls-into-bins process with $T$ balls and $n$ bins.
	For any $\epsilon > 0$, we have
	\begin{align*}
		\mathbb{P}\left(
		\frac{1}{n}X - \left(1 - \frac{1}{n}\right)^T
		\geq
		\epsilon
		\right)
		\leq
		\exp\left(
		- 2  n \epsilon^2
		\right)\, ,
	\\
		\mathbb{P}\left(
		\frac{1}{n}X - \left(1 - \frac{1}{n}\right)^T
		\leq
		\epsilon
		\right)
		\leq
		\exp\left(
		- 2  n \epsilon^2
		\right)\, .
	\end{align*}
\end{lem}

\begin{proof}
	Observe that $\{ I_{j,t} \}_{j \in \{1,\ldots,n\}, t \in \{1,\ldots,T\}}$ are negatively associated (NA) since $\{ I_{j,t} \}_{j \in \{1,\ldots,n\}}$ are NA for each $t$ (by Lemma \ref{lem:NA}--(\ref{lem:NA-permutation}), since $\{ I_{j,t} \}_{j \in \{1,\ldots,n\}}$ is  a uniformly random permutation of $n-1$ zeros and a single one) and they are independent across $t$ (Lemma \ref{lem:NA}--(\ref{lem:NA-independent})).
	Consequently, $W_1, \ldots, W_n$ are NA due to Lemma \ref{lem:NA}--(\ref{lem:NA-monotone}), since $f_j(I_{j,1}, \ldots, I_{j,T}) \defeq \sum_{t=1}^T I_{j,t}$ is non-decreasing in each coordinate.
	
	Define $Y_j \defeq \ind(W_j = 0)$ indicating whether the $j^\text{th}$ bin is empty at the end.
	Although $Y_j$'s are not independent, they are NA (again, by Lemma \ref{lem:NA}--(\ref{lem:NA-monotone})).
	Because $Y_j \sim \text{Bernoulli}\left( \left( 1 - \frac{1}{n} \right)^T \right)$ and $X = \sum_{j=1}^{n} Y_j$, by applying Hoeffding's bound (Lemma \ref{lem:Hoeffding-NA}), we obtain the desired result.
  \end{proof}

\begin{lem}\label{lem:balls-into-bins-reciprocal}
	Let $W_j$ denotes the number of balls in the $j^\text{th}$ bin at the end of a balls-into-bins process with $T$ balls and $n$ bins.
	For any $\Delta > 0$, we have
	\begin{equation*}
		\mathbb{P}\left( \frac{1}{n} \sum_{j=1}^n \frac{1}{W_j+1} \geq \frac{n}{T} + \Delta \right) \leq \exp\left( - 2 n \Delta^2 \right).
	\end{equation*}
\end{lem}

\begin{proof}
	Since $W_j \sim \text{Binomial}\left( T, \frac{1}{n} \right)$, we have
	\begin{align*}
		\mathbb{E}\left[ \frac{1}{W_j+1} \right]
			=\ &\sum_{k=0}^T \frac{1}{k+1} \times \binom{T}{k} \left( \frac{1}{n} \right)^k \left( 1 - \frac{1}{n} \right)^{T-k} \\
			=\ &\frac{n}{T+1} \sum_{k=0}^T \binom{T+1}{k+1} \left( \frac{1}{n} \right)^{k+1} \left( 1 - \frac{1}{n} \right)^{(T+1)-(k+1)}
			\\
			=\ &\frac{n}{T+1} \times \left( 1 - \left( 1 - \frac{1}{n} \right)^{T+1} \right)
			\leq \frac{n}{T}\, .
	\end{align*}
	In the proof of Lemma \ref{lem:balls-into-bins-uniform}, we have shown that $W_1, \ldots, W_n$ are NA.
	By Lemma \ref{lem:NA}--(\ref{lem:NA-monotone}), $\frac{1}{W_1+1}, \ldots, \frac{1}{W_n+1}$ are also NA.
	Therefore, by applying Hoeffding's bound (Lemma \ref{lem:Hoeffding-NA}), we obtain the desired result.
  \end{proof}
}

\subsection{Chernoff's Bound on Random Sum}
\begin{lem}\label{lem:random-sum-chernoff-bound}
Fix any $p \in (0,1)$ and any $p' \in (0,1)$.	Define the random sum
	\begin{align*}
		S \triangleq \sum_{i=1}^{N} X_i\, ,
	\end{align*}
	where $X_i$'s are i.i.d. random variables and have distribution\footnote{Here, by $\textup{Geometric}(p)$ we mean the distribution $\prob(X_i = k) = p (1-p)^{k-1}$ for $k\geq 1$, i.e., the support of the distribution is $\{1, 2, \dots \}$, and its expectation is $1/p > 1$.} $\textup{Geometric}(p)$,
	and $N\sim\textup{Geometric}(p')$ and is independent of $X_i$'s.
	Let $S_i$'s be i.i.d. random variables and have the same distribution as $S$, for $\lambda > \mathbb{E}[S]=1/(pp')$ we have
\begin{align*}
\mathbb{P}\left(\frac{1}{n}\sum_{i=1}^{n}S_i \geq \lambda\right)
\leq
\exp\left(-\frac{n}{2\lambda^2}\left(\lambda - \mathbb{E}[S]\right)^2\right)\, .
\end{align*}
\end{lem}

\begin{proof}
	Denote $q \triangleq 1-p$, $q'\triangleq 1-p'$. In the first step, we derive the moment generating function of $S$, which we denote by $M(t)$. Note that
	\begin{align*}
		M(t) =\ \mathbb{E}[e^{tS}]
		=\ \mathbb{E}\left[e^{t\sum_{i=1}^{N}X_i}\right]
		=\ \mathbb{E}\left[\left[\mathbb{E}e^{tX}\right]^N\right]
		=\ \mathbb{E}\left[ \gamma^N\right]\, .
	\end{align*}
	where
\begin{align}
\gamma \triangleq
\left \{
\begin{array}{ll}
\frac{p e^t}{1 - q e^t}  & \textup{if }  qe^t < 1 \, ,\\
\infty & \textup{otherwise}\, .
\end{array}
\right .
\end{align}
we have
\begin{align*}
\mathbb{E}[\gamma^N]
=
\sum_{k=1}^{\infty} \gamma^k (1-p')^{k-1} p'
=
p' \gamma \sum_{k=1}^{\infty} \gamma^{k-1} (q')^{k-1}
=
\left \{
\begin{array}{ll}
 \frac{p'\gamma}{1 - \gamma q'} & \textup{if } \gamma q' < 1 \, ,\\
\infty & \textup{otherwise}\, .
\end{array}
\right .
\end{align*}
By plugging in $\gamma$, we obtain
\begin{align*}
M(t)
=
\left \{
\begin{array}{ll}
\frac{p'\frac{p e^t}{1 - q e^t}}
{1 - q'\frac{p e^t}{1 - q e^t}}
= \frac{p' p e^t}
{1 - e^t(q+ q'p)} & \textup{if } t < \bar{t} \triangleq \log(1/(q+q'p)) \, ,\\
\infty & \textup{otherwise} \, .
\end{array}
\right .
\end{align*}
Here we used that $q+ q'p >q$ to simplify the condition for $M(t)$ to be finite to $e^t (q+ q'p) < 1 \Leftrightarrow t < \bar{t}$.

Now we derive the convex conjugate of $\log M(t)$, a.k.a. the large deviation rate function.
Note that $\Ex[S] = 1/(pp')$. Define $\Lambda^*: [1/(pp'), \infty) \rightarrow \mathbb{R}$ as
\begin{align*}
\Lambda^*(\lambda)
\triangleq
\sup_{t\geq 0}
\left(
\lambda t
-
\log M(t)
\right) = \sup_{t \in [0, \bar{t})}
\left(
\lambda t
-
\log M(t)
\right)
\end{align*}
Fix $\lambda \geq 1/(pp')$ and let $t^*$ be the maximizer of the supremum above. The derivative of $\lambda t - \log M(t)$ with respect to $t$ for $t \in [0, \bar{t})$ is
\begin{align*}
\lambda -1 -
\frac{e^{t}(q+ q'p)}
{1 - e^{t}(q+ q'p)}
= \lambda -
\frac{1}
{1 - e^{t}(q+ q'p)}\, ,
\end{align*}
and in particular it is decreasing in $t$, corresponding to the fact that $\lambda t - \log M(t)$ is concave in $t$ (we already knew concavity holds because the log moment generating function is always convex). Note further that the derivative at $t=0$ is non-negative since
$$
\lambda -
\frac{1}
{1 - (q+ q'p)} = \lambda - 1/(pp') \geq 0 \, ,
$$
and that the derivative eventually becomes negative since it tends to $- \infty$ as $t \rightarrow \bar{t}^-$. Hence the first order condition will give us the maximizer $t^* \in [0, \bar{t})$ of $\lambda t - \log M(t)$ as follows:
\begin{align*}
\lambda
=
\frac{1}
{1 - e^{t^*}(q+ q'p)}
 \qquad \Rightarrow \qquad
e^{t^*} = \frac{1 - \frac{1}{\lambda}}{q  + q' p }\, .
\end{align*}
Therefore, we have
\begin{align*}
\Lambda^*(\lambda) =&\ \lambda \log\left(1 - \frac{1}{\lambda}\right)
-
\lambda \log\left(q  + q' p  \right)
-
\log\left(
\frac{p' p \frac{1 - \frac{1}{\lambda}}{q  + q' p }}
{1/\lambda}
\right) \\
=&\ \lambda \log\left(1 - \frac{1}{\lambda}\right)
-
\lambda \log\left(q  + q' p  \right)
-
\log\left(
\lambda - 1
\right)
+
C,
\end{align*}
where $C$ is a constant.
A short calculation tells us that
\begin{align}
\frac{d\Lambda^*}{d\lambda}(\lambda)
=\
\log\left(1 - \frac{1}{\lambda}\right)
-
\log\left(q  + q' p  \right)\, ,
\qquad
\frac{d^2\Lambda^*}{d\lambda^2}(\lambda) = \frac{1}{\lambda(\lambda-1)}\, .
\label{eq:Lambda-dlambdas}
\end{align}

Let $S_1,\cdots,S_n$ be i.i.d. random variables with the same distribution as $S$.
Using Chernoff's bound, for $\lambda \geq \mathbb{E}[S]$ we have
\begin{align}
  \mathbb{P}\left(\frac{1}{n}\sum_{i=1}^{n}S_i \geq \lambda\right)
  \leq
  \exp\left(-n\Lambda^*(\lambda)\right).
\label{eq:Chernoff-geomx2}
\end{align}
Since $\Lambda^*(\cdot)$ is a  large deviation rate function, we have that $\Lambda^*(\mathbb{E}[S])=0$ and $\frac{d\Lambda^*}{d\lambda}(\mathbb{E}[S])=0$. We will now use  Taylor's theorem taking terms up to second order for $\Lambda^*(\lambda)$ around $\mathbb{E}[S]$ to obtain the desired bound.  Note that at any $\lambda' \in (\Ex[S], \lambda)\,$, using the explicit form of $\frac{d^2 \Lambda^*}{d\lambda^2}$ in \eqref{eq:Lambda-dlambdas}  we have
\begin{align*}
  \frac{d^2\Lambda^*}{d\lambda^2}(\lambda') \geq \frac{1}{(\lambda')^2} \geq \frac{1}{\lambda^2} \, ,
\end{align*}
where we used $\Ex[S]>1$. Now, using Taylor's theorem, we know that for some $\lambda' \in (0, \lambda)$ we have
$$\Lambda^*(\lambda)= \frac{1}{2}  \frac{d^2\Lambda^*}{d\lambda^2} (\lambda')  \left(\lambda - \mathbb{E}[S]\right)^2 \geq \frac{1}{2\lambda^2} \left(\lambda - \mathbb{E}[S]\right)^2\, .$$
Plugging into \eqref{eq:Chernoff-geomx2}, we obtain
\begin{align*}
  \mathbb{P}\left(\frac{1}{n}\sum_{i=1}^{n}S_i \geq \lambda\right)
  \leq
  \exp\left(-\frac{n}{2\lambda^2}\left(\lambda - \mathbb{E}[S]\right)^2\right)
\end{align*}
as required.
  \end{proof}

\subsection{Notations and Preliminary Observations} \label{append:notation}

We here introduce the variables that formally describe the state of a random matching market over the course of the men-proposing deferred-acceptance (MPDA) procedure (Algorithm~\ref{alg:2MPDA}).

The \emph{time} $t$ ticks whenever a man makes a proposal.
Let $I_t \in \mathcal{M}$ be the man who proposes at time $t$, and $J_t \in \mathcal{W}$ be the woman who receives that proposal.
We define $M_{i,t} \triangleq \sum_{s=1}^t \ind( I_s = i )$ that counts the number of proposals that a man $i$ has made up to time $t$, and define $W_{j,t} \triangleq \sum_{s=1}^t \ind( J_s = j )$ that counts the number of proposals that a woman $j$ has received up to time $t$.
We will often use $\vec{M}_t \triangleq (M_{i,t})_{i \in \mathcal{M}}$ and $\vec{W}_t \triangleq ( W_{j,t} )_{j \in \mathcal{W}}$ as vectorized notations.
By definition, we have
\begin{equation*}
	\sum_{i \in \mathcal{M}} M_{i,t} = \sum_{j \in \mathcal{W}} W_{j,t} = t
	\, ,
\end{equation*}
for any $0 \leq t \leq \tau$ where $\tau$ is the total number of proposals under MPDA.

Let $\mathcal{H}_t \subseteq \mathcal{W}$ be the set of women that the man $I_t$ had proposed to before time $t$: i.e., $\mathcal{H}_t \triangleq \{ J_s : I_s = i \text{ for some } s \leq t-1 \}$ and we have $|\mathcal{H}_t| < d$.
According to the principle of deferred decisions, the $t^\text{th}$ proposal goes to one of women that the man $I_t$ had not proposed to yet: i.e., $J_t$ is sampled from $\mathcal{W} \setminus \mathcal{H}_t$ uniformly at random.
And then, the proposal gets accepted by the woman $J_t$ with probability $1/(W_{J_t,t-1}+1)$.

We denote \emph{the current number of unmatched men} and \emph{women} at time $t$ by $\delta^m[t]$ and $\delta^w[t]$, respectively.
More precisely, $\delta^m[t]$ represents the number of men who have exhausted all his preference list but left unmatched\footnote{
	It is important that the definition of $\delta^m[t]$ does not count the men who have not entered the market until time $t$.
	In other words, it counts the number of men who are ``confirmed'' to be unmatched under MOSM, and correspond to the variable $\delta^m$ described in Algorithm \ref{alg:2MPDA}.
	This quantity is different from the number of unmatched men under the current matching $\mu_t$, which may decrease when a man proposes to a woman who has never received any proposal.
} at time $t$: i.e., $\delta^m[t] \triangleq \sum_{i \in \mathcal{M}} \ind( M_{i,t} = d, \mu_t(i) = i )$ where $\mu_t$ is the current matching at time $t$.
Also note that once a woman receives a proposal, she remains matched until the end of MPDA procedure: i.e., $\delta^w[t] \triangleq \sum_{j \in \mathcal{W}} \ind(\mu_t(j) = j) = \sum_{j \in \mathcal{W}} \ind( W_{j,t} = 0 )$.
We observe that $\delta^m[t]$ starts from zero (at $t=0$) and is \emph{non-decreasing} over time, and $\delta^w[t]$ starts from $n$ and is \emph{non-increasing} over time.

Recall that $\tau$  is the \emph{the total number of proposals} that is made until the end of MPDA, i.e., the time at which the men-optimal stable matching (MOSM) is found.
MPDA ends when there is no more man to make a proposal, i.e., when every unmatched man had already exhausted his preference list.
In \eqref{eq:tau}, we expressed $\tau$ as a stopping time, namely,
\begin{align*}
	\tau = \min\{t\geq 1: \delta^m[t] = \delta^w[t] + k \}\, .
\end{align*}
In particular, we have
\begin{equation*}
	\delta^m[\tau] = \delta^w[\tau]+k
	\, ,
\end{equation*}
since the number of matched men equals to the number of matched women under any feasible matching.
Furthermore, we have
\begin{equation*}
	\Rmen(\textup{MOSM}) = \frac{\tau+\delta^m[\tau]}{n+k}
	\, ,
\end{equation*}
by the definition of men's rank.
\\

\noindent
\textbf{An extended process.}
We introduce an \emph{extended process} as a natural continuation of the MPDA procedure that continues to evolve even after the MOSM is found (i.e., the extended process continues for $t > \tau$).
Recall that the MPDA procedure under the principle of deferred decisions works as follows: As described in Algorithm \ref{alg:2MPDA}, $n+k$ men in $\mathcal{M}$ sequentially enter the market one by one, and whenever a new man enters, he makes a proposal and the acceptance/rejection process continues until all men who have entered are either matched or have reached the bottom of their preference lists (i.e., until it finds a new MOSM among the men who have entered including the newly entered man).

To define the extended process, we start by defining an extended market, which has the same $n$ women but an infinite supply of men: $n+k$ ``real'' men $\cM$ who are present in the original market, and an infinity of ``fake'' men $\cMf$ in addition. The distribution of preferences in the extended market is again as described in Section~\ref{sec:model} (in particular, the preference distribution does not distinguish real and fake men).
We then define the \emph{extended process} as tracking the progress of Algorithm~\ref{alg:2MPDA} on the extended market: the $n+k$ real men enter first in Algorithm~\ref{alg:2MPDA}, as before, and we then continue Algorithm~\ref{alg:2MPDA} after time $\tau$ for all $t> \tau$ by continuing to introduce additional (fake) men sequentially after time $\tau$.
In particular, the extended process is identical to the original MPDA process until the MOSM is found (i.e., for $t \leq \tau$).

Observe that in this extended process, the MOSM among $\mathcal{M} \cup \mathcal{W}$ can be understood as a stable outcome found after $n+k$ men have entered the market.
Therefore, all the aforementioned notations ($I_t$, $J_t$, $M_{i,t}$, $W_{j,t}$, $\mathcal{H}_t$, $\mu_t$, $\delta^m[t]$, $\delta^w[t]$) are well-defined for any time $t \geq 0$ while preserving all their properties characterized above, 
and we similarly denote by $\hcM[t] \subset \cM \cup \cMf$ the set of men who have entered so far (consistent with the notation in Algorithm~\ref{alg:2MPDA}).
In the later proofs, we utilize these notations and their properties (e.g., $\delta^m[\tau] \leq \delta^m[t]$ implies that $\tau \leq t$ since $\delta^m[t]$ is non-decreasing over time for $t=0,1,\ldots$).
\\

\noindent
\textbf{Balls-into-bins process analogy.}
When we analyze the women side, we heavily utilize the balls-into-bins process as done in \cite{knuth1976mariages}.
We make an analogy between the number of proposals that each of $n$ women has received (denoted by $W_{j,t}$) and the number of balls that had been placed into each of $n$ bins.
For example, the number of unmatched women at time $t$ corresponds to the number of empty bins after $t$ balls had been placed.

Recall that, according to the principle of deferred decisions, the $t^\text{th}$ proposal goes to one of women uniformly at random among whom he had not yet proposed to (i.e., $\mathcal{W} \setminus \mathcal{H}_t$), and thus the recipients of proposals, $J_1, J_2, \ldots$, are not independent.
In the balls-into-bins process, in contrast, the $t^\text{th}$ ball is placed into one of $n$ bins uniformly at random, independently of the other balls' placement.
Despite this difference (sampling without replacement v.s. sampling with replacement), the balls-into-bins process provides a good enough approximation as the number of proposals made by an individual man (i.e., $|\mathcal{H}_t|$) is much smaller than the total number of men and women.
{
We will show that (e.g., in Lemma~\ref{lem:acceptance-probability-upper-bound} in the next section) that the corresponding error term can be effectively bounded.
}

\section{Proof for Small to Medium-Sized $d$: the case of $d=o(\log^2 n)$, $d=\omega(1)$} \label{append:small-medium-d-proof}
In this section, we consider the case such that $d=o(\log^2 n)$ and $d=\omega(1)$.
We will prove the following quantitative version of Theorem \ref{thm:main-result}.
\begin{thm}[Quantitative version of Theorem \ref{thm:main-result}]\label{thm:main-result-append}
		Consider a sequence of random matching markets indexed by $n$, with $n+k$ men and $n$ women ($k=k(n)$ can be positive or negative), and the men's degrees are $d=d(n)$.
	If $|k|=O(ne^{-\sqrt{d}})$, $d=\omega(1)$ and $d=o(\log^2 n)$, then with probability $1 - O(\exp(-d^{\frac{1}{4}}))$ we have
	\begin{enumerate}
		\item \emph{(Men's average rank of wives)}
		\begin{align*}
		\left| \Mrank - \sqrt{d} \right|
		\leq\ 6d^{\frac{1}{4}}
		\, .
		\end{align*}
		\item \emph{(Women's average rank of husbands)}
		\begin{align*}
		\left| \Wrank - \sqrt{d} \right|
		\leq\ 8d^{\frac{1}{4}}
		\, .
		\end{align*}
		\item \emph{(The number of unmatched men)}
		\begin{align*}
		\left|\log \delta^m - \log ne^{-\sqrt{d}}\right| \leq\ 3 d^{\frac{1}{4}}
		\, .
		\end{align*}
		\item \emph{(The number of unmatched women)}
		\begin{align*}
		\left|\log \delta^w - \log ne^{-\sqrt{d}}\right| \leq\ 2.5 d^{\frac{1}{4}}
		\, .
		\end{align*}
	\end{enumerate}
\end{thm}

The proofs are organized as follows:
\begin{itemize}
	\item (Section~\ref{append:small-medium-d-proof-step1}) We first show that with high probability, the stopping time of MPDA (Algorithm~\ref{alg:2MPDA}), namely, $\tau$, is bounded above as $\tau\leq n\left(\sqrt{d} + d^{\frac{1}{4}}\right)$, by utilizing the coupled extended process defined in Section \ref{append:notation}.
	This yields a high probability upper bound on $\Mrank$ and a lower bound on the number of unmatched men $\delta^m$ and unmatched women $\delta^w$.

	\item (Section~\ref{append:small-medium-d-proof-step2}) We prove the complementary bounds on $\Mrank$, $\delta^m$, and $\delta^w$: a lower bound on $\Mrank$ and an upper bound on the number of unmatched men $\delta^m$ and unmatched women $\delta^w$.
	To this end, we start by analyzing the rejection chains triggered by the last man to enter in MPDA, and deduce upper bounds on $\mathbb{E}[\delta^m]$ and $\mathbb{E}[\delta^w]$, using the fact that the order in which men enter does not matter.
	Using Markov's inequality, we then obtain high probability upper bounds on $\delta^m$ and $\delta^w$, which lead to lower bounds on $\tau$ and $\Mrank$.
	
	\item (Section~\ref{append:small-medium-d-proof-step3}) We construct the concentration bounds on $\Wrank$ based upon the concentration results on $\tau$.
	{In this step, we utilizes the balls-into-bins process to analyze the women's side while carefully controlling the difference between the MPDA procedure and the balls-into-bins process.}
		This completes the proof of Theorem \ref{thm:main-result-append}.
\end{itemize}

\subsection{Step 1: Upper Bound on the Total Number of Proposals $\tau$} \label{append:small-medium-d-proof-step1}
We prove the following two propositions.

\begin{prop}[Upper bound on men's average rank]\label{prop:total-proposal-upper-bound}
Consider the setting of Theorem \ref{thm:main-result}. With probability $1 - O(\exp(-\sqrt{n}))$, we have the following upper bounds on the total number of proposals and men's average rank:
	\begin{align*}
	\tau \leq\ n\left(\sqrt{d} + d^{\frac{1}{4}}\right)\, ,\qquad
		\Mrank
		\leq\
		\sqrt{d}
		+
		2d^{\frac{1}{4}}
		\, .
	\end{align*}
\end{prop}

\begin{prop}[Lower bound on the number of unmatched women]\label{prop:unmatched-women-lower-bound}
	Consider the setting of Theorem \ref{thm:main-result}.
	With probability $1 - O(\exp(-\sqrt{n}))$, we have the following lower bounds on the number of unmatched men $\delta^m$ and unmatched women $\delta^w$:
	\begin{align*}
	\delta^m \geq\ n \exp\left( -\sqrt{d} - 3 d^{\frac{1}{4}} \right)
	\, ,
	\qquad
	\delta^w \geq\ n\exp\left( -\sqrt{d} - 2 d^{\frac{1}{4}} \right)
	\, .
	\end{align*}
\end{prop}

Throughout the proofs we utilize the extended process defined in Section \ref{append:notation}, which enables us to analyze the state dynamics even after the termination of original DA procedure.
Most of the work is in proving Proposition~\ref{prop:total-proposal-upper-bound}, which is done in Sections~\ref{append:small-medium-d-proof-step1-1}--\ref{append:small-medium-d-proof-step1-4}. We then deduce Proposition~\ref{prop:unmatched-women-lower-bound} from Proposition~\ref{prop:total-proposal-upper-bound} in Section~\ref{append:small-medium-d-proof-step1-5}.
The overall proof structure is as follows:
\begin{itemize}
	\item (Sections~\ref{append:small-medium-d-proof-step1-1} and \ref{append:small-medium-d-proof-step1-2})
		We first analyze the women side using balls-into-bins process analogy:
		Given that a sufficient number of proposals have been made (in particular, for $t = (1+\epsilon) n\sqrt{d}$), we construct a high probability upper bound on the current number of unmatched women $\delta^w[t]$ and the probability $p_t$ of a proposal being accepted.
	
	\item (Sections~\ref{append:small-medium-d-proof-step1-3} and \ref{append:small-medium-d-proof-step1-4})
		We then analyze the men side and obtain a lower bound on the current number of unmatched men $\delta^m[t]$ at $t = (1+\epsilon) n \sqrt{d}$ by utilizing the upper bound on acceptance probability $p_t$.
		Since this lower bound exceeds the upper bound on $\delta^w[t]$ (plus $k$) which holds at the same $t$, we deduce that, whp, the algorithm has already terminated, $\tau \leq t= (1+\epsilon) n \sqrt{d}$, since we know that $\delta^m[\tau] = \delta^w[\tau]+k$.
{See Figure \ref{fig:sample-path} in Section~\ref{sec:proof-sketch} for illustration.}
		Consequently, an upper bound on $\Rmen$ follows from the identity $\Rmen = \frac{\tau + \delta^m}{n+k}$, thus completing the proof of Proposition \ref{prop:total-proposal-upper-bound}.

	\item (Section~\ref{append:small-medium-d-proof-step1-5}) Given the upper bound on $\tau$, we obtain a lower bound on $\delta^w$ using the balls-into-bins analogy again.
		This leads to a lower bound on $\delta^m$ due to the identity $\delta^m = \delta^w + k$, which completes the proof of Proposition \ref{prop:unmatched-women-lower-bound}.
\end{itemize}

\subsubsection{Upper bound on number of unmatched women after a sufficient number of proposals} \label{append:small-medium-d-proof-step1-1}
The following result formalizes the fact that there cannot be too many unmatched women after a sufficient number of proposals have been made.

\begin{lem}\label{lem:unmatched-woman-upper-bound}
 Consider the setting of Theorem \ref{thm:main-result} and the extended process defined in Section \ref{append:notation}.
 For any $\epsilon \in (0, \frac{1}{2})$ and $n\in\mathbb{Z}_+$, we have
	\begin{align}
		\mathbb{P}\left( \delta^w[(1+\epsilon)n\sqrt{d}] > ne^{-(1+\frac{\epsilon}{2})\sqrt{d}} \right)
		\leq\
		\exp\left(
		- \frac{1}{2} n d \epsilon^2 e^{-3\sqrt{d}}
		\right)\, .
		\label{eq:unmatched-woman-upper-bound}
	\end{align}
In words, after $t=(1+\epsilon)n\sqrt{d}$ proposals have been made, at most $ne^{-(1+\frac{\epsilon}{2})\sqrt{d}}$ women remain unmatched with high probability.
\end{lem}

\begin{proof}
	It is well known that for any $t>0$, $\delta^w[t]$ is stochastically dominated by the number of empty bins at the end of a balls-into-bins process (defined in Section \ref{subsec:balls-into-bins}) with $t$ balls and $n$ bins, which we denote by $X_{t,n}$. (See, e.g., \cite{knuth1976mariages}; the idea is that since men's preference lists sample women without replacement, the actual process has a weakly larger probability of proposing to an unmatched woman at each step relative to picking a uniformly random woman, and hence a stochastically smaller number of unmatched women at any given $t$.)
	Therefore we have
	\begin{align*}
		&\mathbb{P}\left( \delta^w[(1+\epsilon)n\sqrt{d}] > n e^{-(1+\frac{\epsilon}{2})\sqrt{d}} \right) \\ 
	\leq\ &
		\mathbb{P}\left( X_{(1+\epsilon)n\sqrt{d},n} > n e^{-(1+\frac{\epsilon}{2})\sqrt{d}} \right) \nonumber\\
	=\ &
		\mathbb{P}\left(
		\frac{1}{n}X_{(1+\epsilon)n\sqrt{d},n}
		-
		\left(1 - \frac{1}{n}\right)^{(1+\epsilon)n\sqrt{d}}
		>
		e^{-(1+\frac{\epsilon}{2})\sqrt{d}}
		-
		\left(1 - \frac{1}{n}\right)^{(1+\epsilon)n\sqrt{d}}
		\right)\, . \nonumber
	\end{align*}
	By Lemma \ref{lem:basic-inequalities} and Lemma \ref{lem:balls-into-bins-uniform}, we further have
	\begin{align*}
		 &\mathbb{P}\left( \delta^w[(1+\epsilon)n\sqrt{d}] > n e^{-(1+\frac{\epsilon}{2})\sqrt{d}} \right)\\ 
		\leq\ &
			\mathbb{P}\left(
		\frac{1}{n}X_{(1+\epsilon)n\sqrt{d},n}
		-
		\left(1 - \frac{1}{n}\right)^{(1+\epsilon)n\sqrt{d}}
		>
		e^{-(1+\frac{\epsilon}{2})\sqrt{d}}
		-
		e^{-(1+\epsilon)\sqrt{d}}
		\right) \\
		\leq\ &
		\exp\left(
		- 2 n
		\left(e^{-(1+\frac{\epsilon}{2})\sqrt{d}} - e^{-(1+\epsilon)\sqrt{d}}\right)^2
		\right)\, .
	\end{align*}
	For $0<a<b$, using the convexity of function $f(x)=e^{-x}$ we have $e^{-a}-e^{-b} \geq e^{-b}(b-a)$, and therefore for $\epsilon \in (0,\frac{1}{2})$ and any $n\in\mathbb{Z}_+$,
	\begin{align*}
		\mathbb{P}\left( \delta^w[(1+\epsilon)n\sqrt{d}] > n e^{-(1+\frac{\epsilon}{2})\sqrt{d}} \right) 
		\leq\
		\exp\left(
		- 2 n d \left(\epsilon-\frac{\epsilon}{2}\right)^2 e^{-2(1+\epsilon)\sqrt{d}}
		\right)
		\leq \
		\exp\left(
		- \frac{1}{2} n d \epsilon^2 e^{-3\sqrt{d}}
		\right)\, .
	\end{align*}
	This concludes the proof.
  \end{proof}

\subsubsection{Upper bound on ex-ante acceptance probability} \label{append:small-medium-d-proof-step1-2}
We define the \emph{ex-ante acceptance probability} as
\begin{equation}
	p_t \triangleq \frac{1}{| \mathcal{W} \setminus \mathcal{H}_t |} \sum_{j \in \mathcal{W} \setminus \mathcal{H}_t} \frac{1}{W_{j,t-1}+1} \, .
\label{eq:pt-def}
\end{equation}
This is the probability that the $t^\text{th}$ proposal is accepted after the proposer $I_t$ is declared but the recipient $J_t$ is not yet revealed (recall that $I_t$ is the identity of the man who makes the $t^\text{th}$ proposal, $J_t$ is the identity of the woman who receives it, and $\mathcal{H}_t$ is the set of women whom $I_t$ has previously proposed to).
In the following lemma, we construct a high probability upper bound on the summation in \eqref{eq:pt-def}, and the subsequent lemma will use it to obtain an upper bound on $p_{(1+\frac{\epsilon}{2})n \sqrt{d}}$ for small $\epsilon$. 

\begin{lem}\label{lem:acceptance-probability-upper-bound}
	For any $\Delta > 0$ and $t$ such that $1 \leq t \leq nd$, we have
	\begin{equation*}
		\mathbb{P}\left( \frac{1}{n} \sum_{j \in \cW} \frac{1}{W_{j,t}+1} \geq \frac{n}{t} + \frac{d^2}{n} + \Delta \right)
			\leq 2 \exp\left( - \frac{n \Delta^2}{8d} \right).
	\end{equation*}
	This is also valid for the extended process (i.e., when $t \geq \tau$).
\end{lem}

\begin{proof}
	Consider a balls-into-bins process with $t$ balls and $n$ bins, and let $\tilde{J}_s \in \{1,\ldots,n\}$ be the index of bin into which the $s^\text{th}$ ball is placed, and let $\tilde{W}_{j,t} \defeq \sum_{s=1}^t \ind( \tilde{J}_s = j)$ be the total number of balls placed in the $j^\text{th}$ bin.
	Recall that $\tilde{J}_s$ is being sampled from $\{1,\ldots,n\}$ ($=\cW$) uniformly at random.

	We make a coupling between the MPDA procedure and the balls-into-bins process as follows: when determining the $s^\text{th}$ recipient $J_s$, we take $J_s \gets \tilde{J}_s$ if $\tilde{J}_s \notin \mathcal{H}_s$, or otherwise, sample $J_s$ among $\cW \setminus \mathcal{H}_s$ uniformly at random.
	In other words, the man $I_s$ first picks a woman $\tilde{J}_s$ among the entire $\cW$ uniformly at random, and then proposes to her only if he had not proposed to her yet; if he already had proposed before, he proposes to another woman randomly sampled among $\cW \setminus \mathcal{H}_s$.
	It is straightforward that the evolution of the recipient process $J_s$ under this coupling is identical to that under the usual MPDA procedure.

	Define $D_t \defeq \sum_{s=1}^t \ind( J_s \ne \tilde{J}_s )$ representing the total discrepancy between the MPDA procedure and its coupled balls-into-bins process.
	Observe that $\ind(\tilde{J}_s \ne J_s) = \ind\left( \tilde{J}_s \in \mathcal{H}_s \right)$ and thus $\mathbb{P}\left( \left. \tilde{J}_s \ne J_s \right| \mathcal{F}_{s-1} \right) = \mathbb{P}\left( \left. \tilde{J}_s \in \mathcal{H}_s \right| \mathcal{F}_{s-1} \right) \leq \frac{d}{n}$ where $\mathcal{F}_{s-1}$ represents all information revealed up to time $s-1$.
	Let $Z_s \defeq D_s - \frac{d}{n} s$ and observe that $(M_s)_{s \geq 0}$ is a supermartingale with $Z_0 = 0$ and $|Z_{s+1} - Z_s| \leq 1$.
	By Azuma's inequality, we have for any $\Delta_0 > 0$,
	\begin{equation*}
		\mathbb{P}\left( D_t \geq \frac{dt}{n} + \Delta_0 \right) \leq \mathbb{P}( Z_t - Z_0 \geq \Delta_0 ) \leq \exp\left( - \frac{\Delta_0^2}{2t} \right).
	\end{equation*}
	
{
	On the other hand, since $0 \leq \frac{1}{w+1} \leq 1$ for any $w \geq 0$, we deduce that 
	\begin{align*}
		\sum_{j \in \cW} \frac{1}{W_{j,t}+1} - \sum_{j \in \cW} \frac{1}{\tilde{W}_{j,t}+1}
			\leq\ &\sum_{j \in \cW:W_{j,t} < \tilde{W}_{j,t}} \left( \frac{1}{W_{j,t}+1} - \frac{1}{\tilde{W}_{j,t}+1} \right)\\
			\leq\ &\left| \{ j \in \cW: W_{j,t} < \tilde{W}_{j,t} \} \right|
			\leq\ D_t\, ,
	\end{align*}
	where the last inequality follows from the fact that in order to observe $W_{j,t} < \tilde{W}_{j,t}$ for some $j$, at least one mismatch $\{ \tilde{J}_s \ne J_s \}$ should take place.
	Based on the high probability upper bound on $D_t$ obtained above, we have for any $\Delta_1 > 0$,
	\begin{equation} \label{eq:acceptance-probability-gap}
		\mathbb{P}\left( \frac{1}{n} \sum_{j \in \cW} \frac{1}{W_{j,t}+1} - \frac{1}{n}  \sum_{j \in \cW} \frac{1}{\tilde{W}_{j,t}+1} \geq \frac{dt}{n^2} + \Delta_1 \right)
			\leq \mathbb{P}\left( D_t \geq \frac{dt}{n} + n \Delta_1 \right)
			\leq \exp\left( - \frac{n^2 \Delta_1^2}{2t} \right).
	\end{equation}
	
	We now utilize the result derived for the balls-into-bins process.
	From Lemma \ref{lem:balls-into-bins-reciprocal}, we have for any $\Delta_2 > 0$,
	\begin{equation*}
		\mathbb{P}\left( \frac{1}{n} \sum_{j \in \cW} \frac{1}{\tilde{W}_{j,t}+1}
		\geq \frac{n}{t} + \Delta_2 \right)
		\leq \exp\left( - 2 n \Delta_2^2 \right).
	\end{equation*}
	Combined with \eqref{eq:acceptance-probability-gap},
	\begin{align*}
		&\mathbb{P}\left( \frac{1}{n} \sum_{j \in \cW} \frac{1}{W_{j,t}+1} \geq \frac{n}{t} + \Delta_2 + \frac{dt}{n^2} + \Delta_1 \right)
			\\
			\leq\ &\mathbb{P}\left( \frac{1}{n} \sum_{j \in \cW} \frac{1}{W_{j,t}+1} \geq \frac{n}{t} + \Delta_2 + \frac{dt}{n^2} + \Delta_1, \frac{1}{n} \sum_{j \in \cW} \frac{1}{\tilde{W}_{j,t} + 1} < \frac{n}{t} + \Delta_2 \right)\\
			& + \mathbb{P}\left( \frac{1}{n} \sum_{j \in \cW} \frac{1}{\tilde{W}_{j,t}+1} \geq \frac{n}{t} + \Delta_2 \right)
			\\
			\leq\ &\mathbb{P}\left( \frac{1}{n} \sum_{j \in \cW} \frac{1}{W_{j,t}+1} - \frac{1}{n} \sum_{j \in \cW} \frac{1}{\tilde{W}_{j,t}+1} \geq \frac{dt}{n^2} + \Delta_1 \right) + \exp\left( - 2n \Delta_2^2 \right)
			\\
			\leq\ &\exp\left( - \frac{n^2 \Delta_1^2}{2t} \right) + \exp\left( - 2n \Delta_2^2 \right)\, ,
	\end{align*}
	for any $\Delta_1 > 0$ and $\Delta_2 > 0$.

	We are ready to prove the claim.
	Given any $\Delta > 0$, take $\Delta_1 = \Delta_2 = \Delta/2$.
	Then,
	\begin{align*}
		&\mathbb{P}\left( \frac{1}{n} \sum_{j \in \cW} \frac{1}{W_{j,t}+1} \geq \frac{n}{t} + \frac{d^2}{n} + \Delta \right)\\
			\leq\ &\mathbb{P}\left( \frac{1}{n} \sum_{j \in \cW} \frac{1}{W_{j,t}+1} \geq \frac{n}{t} + \frac{dt}{n^2} + \Delta_1 + \Delta_2 \right)
			\\
			\leq\ &\exp\left( - \frac{n^2 \Delta_1^2}{2t} \right) + \exp\left( - 2n \Delta_2^2 \right)\\
			=\ &\exp\left( - \frac{n^2 \Delta^2}{8t} \right) + \exp\left( - \frac{1}{2} n \Delta^2 \right)
			\\
			\leq\ &\exp\left( - \frac{n \Delta^2 }{8d} \right) + \exp\left( - \frac{1}{2} n \Delta^2 \right)\\
			\leq\ &2 \exp\left( - \frac{n \Delta^2 }{8d} \right)\, ,
	\end{align*}
	where we utilized the fact that $\frac{dt}{n^2} \leq \frac{d^2}{n}$ and $\frac{n}{d} \leq \frac{n^2}{t}$ under the given condition $t \leq nd$.}
  \end{proof}

\begin{lem} \label{cor:acceptance-probability-upper-bound}
Fix any $\alpha \in (0,1)$, $\epsilon<0.2$ and sequences $(d(n))_{n \in \mathbb{N}}$, and $(\gamma(n))_{n \in \mathbb{N}}$ such that $d= d(n) = \omega(1)$ and $d = o(\log^2n)$, and $\gamma = \gamma(n) = \Theta\left( n^{-\alpha} \right)$. Define the maximal ex-ante acceptance probability (for any $t \leq nd$) as
	\begin{equation} \label{eq:acceptance-probability-worst-case}
		\overline{p}_t \defeq \max_{\mathcal{H} \subset \mathcal{W}: |\mathcal{H}| \leq d}\left\{ \frac{1}{| \mathcal{W} \setminus \mathcal{H} |} \sum_{j \in \mathcal{W} \setminus \mathcal{H}} \frac{1}{W_{j,t-1}+1} \right\}\, .
	\end{equation}
	Then there exists $n_0 < \infty$ such that for all $n > n_0$, we have
	$$
		\mathbb{P}\left( \overline{p}_{(1+\frac{\epsilon}{2}) n \sqrt{d}} \geq \frac{1+\gamma}{(1+\frac{\epsilon}{2}) \sqrt{d}} \right) \leq 2 \exp\left( - \frac{\gamma^2}{32} \frac{n}{d^2} \right).
	$$
	This is also valid for the extended process (i.e., when $(1+\frac{\epsilon}{2}) n\sqrt{d} \geq \tau$).
\end{lem}

\begin{proof}
	Let $t \defeq (1+\frac{\epsilon}{2}) n \sqrt{d}$ and $\mathcal{H}^*$ be the maximizer of \eqref{eq:acceptance-probability-worst-case}.
	Observe that $|\mathcal{W} \setminus \mathcal{H}^*| \geq n-d$ and $\sum_{j \in \mathcal{W} \setminus \mathcal{H}^*} \frac{1}{W_{j,t-1} +1} \leq \sum_{j \in \mathcal{W}} \frac{1}{W_{j,t-1}+1}$, and hence
	\begin{equation*}
		\overline{p}_t = \frac{1}{|\mathcal{W} \setminus \mathcal{H}^*|} \sum_{j \in \mathcal{W} \setminus \mathcal{H}^*} \frac{1}{W_{j,t-1}+1}
			\leq \frac{1}{n-d} \sum_{j \in \mathcal{W}}  \frac{1}{W_{j,t-1}+1}
			\leq \frac{1}{n-d} \left( 1 + \sum_{j \in \mathcal{W}}  \frac{1}{W_{j,t}+1} \right).
	\end{equation*}
The last inequality uses that at most one of the terms in the summation decreases from $t-1$ to $t$, and the decrease in that term is less than $1$.
	
Let $r \triangleq \frac{t}{n} = (1+\frac{\epsilon}{2}) \sqrt{d}$, and $\Delta \triangleq \frac{\gamma}{2\sqrt{d}}$.
	Under the specified asymptotic conditions, for $n$ large enough we have
	\begin{equation*}
		r\Delta = (1+\tfrac{\epsilon}{2})\sqrt{d} \, \cdot \, \frac{\gamma}{2\sqrt{d}} \leq 0.6 \gamma
		\, , \ \
		\frac{r}{n} = \frac{(1+\frac{\epsilon}{2})\sqrt{d}}{n} \leq 0.1 \gamma
		\, , \ \
		\frac{rd^2}{n} \leq \frac{(1+\frac{\epsilon}{2}) d^{5/2}}{n} \leq 0.1 \gamma \, , \ \ \frac{d}{n} \leq 0.1 \gamma \, .
	\end{equation*}
	Consequently, since $\gamma = o(1)$, for large enough $n$ we have
	\begin{align*}
		\frac{n}{n-d} \left( \frac{1}{r} + \frac{d^2}{n} + \Delta + \frac{1}{n} \right)
			 &= \frac{1}{1-d/n} \cdot \frac{1}{r} \cdot \left( 1 + \frac{rd^2}{n} + r \Delta + \frac{r}{n} \right)
			 \\&\leq \frac{1}{r} \cdot \frac{1}{1-0.1\gamma} \cdot \left( 1 + 0.1 \gamma + 0.6 \gamma + 0.1 \gamma \right)
			 \\&\leq \frac{1}{r} \cdot (1+\gamma)
			 = \frac{1+\gamma}{ (1+\frac{\epsilon}{2})\sqrt{d}}\, .
	\end{align*}
	As a result,
	\begin{align*}
		\mathbb{P}\left( \overline{p}_{(1+\frac{\epsilon}{2})\tau^*} \geq \frac{1+\gamma}{(1+\frac{\epsilon}{2})\sqrt{d}} \right)
			&\leq \mathbb{P}\left( \frac{1}{n-d} \left( 1 + \sum_{j \in \mathcal{W}} \frac{1}{W_{j,t}+1} \right) \geq \frac{1+\gamma}{(1+\frac{\epsilon}{2})\sqrt{d}} \right)
			\\&\leq \mathbb{P}\left( \frac{1}{n-d} \left( 1 + \sum_{j \in \mathcal{W}} \frac{1}{W_{j,t}+1} \right)  \geq \frac{n}{n-d} \times \left( \frac{1}{r} + \frac{d^2}{n} + \Delta + \frac{1}{n} \right) \right)
			\\&= \mathbb{P}\left( \frac{1}{n} \sum_{j \in \mathcal{W}} \frac{1}{W_{j,t}+1}  \geq \frac{n}{t} + \frac{d^2}{n} + \Delta  \right)
			\\&\leq 2 \exp\left( - \frac{n \Delta^2}{8d} \right)
			= 2 \exp\left( - \frac{\gamma^2}{32} \frac{n}{d^2} \right)\, ,
	\end{align*}
	where the last inequality follows from Lemma \ref{lem:acceptance-probability-upper-bound}.
  \end{proof}

\subsubsection{Lower bound on the number of unmatched men after a sufficient number of proposals} \label{append:small-medium-d-proof-step1-3}
The following result formalizes the fact that there cannot be too few unmatched men after an enough number of proposals have been made.

\begin{lem}\label{lem:unmatched-man-lower-bound}
	Consider the setting of Theorem \ref{thm:main-result} and the extended process defined in Section \ref{append:notation}.
	For any sequence $(\epsilon(n))_{n \in \mathbb{N}}$ such that $\epsilon = \epsilon(n) <0.2$ and $\epsilon(n) = \omega\left( \frac{1}{n^{0.49}} \right)$,
	there exists $n_0 < \infty$ such that for all $n > n_0$, we have
	\begin{align} \label{eq:unmatched-man-lower-bound}
	\mathbb{P}\left(
	\delta^m[(1+\epsilon)n\sqrt{d}]\leq \frac{\epsilon}{16}n e^{-(1 - \frac{\epsilon}{3})\sqrt{d}}
	\right)
	\leq
	\exp\left(
	- \sqrt{n}
	\right)
	.
	\end{align}
In words, after $(1+\epsilon)n\sqrt{d}$ proposals have been made, at least $\frac{\epsilon}{16}n e^{-(1 - \frac{\epsilon}{3})\sqrt{d}}$ men become unmatched with high probability.
\end{lem}

\begin{proof}
	Let $\tau^* \triangleq n \sqrt{d}$.
	To obtain a lower bound on the number of unmatched men at time $(1+\epsilon)\tau^*$, we count the number of \emph{$d$-rejection-in-a-row} events that occur during $[(1+\frac{\epsilon}{2})\tau^*,(1+\epsilon)\tau^*]$.
	This will provide a lower bound since whenever the rejection happens $d$ times in a row the number of unmatched men increases at least by one.

	For this purpose, we first utilize the upper bound on the ex-ante acceptance probability.
	By Lemma \ref{cor:acceptance-probability-upper-bound} we have: given that $\gamma =\gamma(n)= \Theta\left( \frac{1}{n^{\alpha}} \right)$ for some $\alpha \in (0,1)$, $\epsilon=\epsilon(n)<0.2$, and that $d=d(n)=\omega(1)$ and $d=o(\log^2 n)$, there exists $n_0>0$ such that for all $n>n_0$,
	\begin{equation}  \label{eq:man-unmatched-lb-aux-1}
		\mathbb{P}\left( \overline{p}_{(1+\frac{\epsilon}{2})\tau^*} \geq \frac{1+\gamma}{(1+\frac{\epsilon}{2}) \sqrt{d}} \right)
			\leq
			 2 \exp\left( - \frac{\gamma^2}{32} \frac{n}{d^2} \right)\, .
	\end{equation}
	
	Let $\hat{p} \triangleq \frac{1+\gamma}{(1+\frac{\epsilon}{2})\sqrt{d}}$ and consider the events where $\overline{p}_{(1+\frac{\epsilon}{2})\tau^*} \leq \hat{p}$ is satisifed.
	Since $\overline{p}_t$ is non-increasing over time on each sample path, we have $p_t \leq \hat{p}$ for all $t \geq (1+\frac{\epsilon}{2})\tau^*$ on this sample path: i.e., a proposal after time $(1+\frac{\epsilon}{2})\tau^*$ is accepted with probability at most $\hat{p}$.
	
	
	As an analogy, we imagine a coin tossing process with head probability $\hat{p}$ (which is an exaggeration of the actual acceptance probability, making it underestimate the occurrence of rejections and provides a valid lower bound on the actual number of $d$-rejection-in-a-row events), and count how many times $d$-tail-in-a-row takes place during $\frac{\epsilon}{2} \tau^*$ coin tosses.
	With $X_i \stackrel{ \text{i.i.d.} }{\sim} \text{Geometric}(\hat{p})$ representing the number of coin tosses required to observe one head (acceptance), the total number of coin tosses required to observe one $d$-tail-in-a-row is given by $\sum_{i=1}^N \min\{ X_i, d \}$	where $N$ is the smallest $i$ such that $X_i > d$.
	Note that $N \sim \text{Geometric}\left(\left(1 - \hat{p}\right)^{d}\right)$.
	However, $N$ is correlated with $X_i$'s.
	To upper bound the random sum, observe that conditioned on $N$, $\{ X_1,\cdots,X_{N-1} \}$ are independent \emph{truncated} Geomtric$(\hat{p})$ variables that only take value on $\{1,\cdots,d\}$, which are stochastically dominated by Geomtric$(\hat{p})$ random variables.
	Since $\min\{X_{N},d\}\leq d$, the random sum of interest is stochastically dominated by $d+S$, where
	$
		S = \sum_{i=1}^{N'}X_i\, ,
	$
	and $N' \sim \text{Geometric}\left(\left(1 - \hat{p}\right)^{d}\right)$ independent of $X_i$'s.
	(Note that by Wald's identity we have $\mathbb{E}[S] = \hat{p}^{-1} \left( 1 - \hat{p} \right)^{-d}$.)
	Consequently, the total number of coin tosses required to observe $\frac{\epsilon}{8}ne^{-d \hat{p}}$ $d$-tail-in-a-row's is stochastically dominated by
	\begin{equation*}
		\sum_{j=1}^{\frac{\epsilon}{8}ne^{-d \hat{p}}}(d + S_j)\, ,
	\end{equation*}
	where $S_1, S_2, \ldots$ are i.i.d. random variables with the same distribution as $S$ defined above.

	Let $R$ denote the total number of $d$-tail-in-a-row events that occur during $[(1+\frac{\epsilon}{2})\tau^*,(1+\epsilon)\tau^*]$.
	From the above argument, we deduce that
	\begin{align}\label{eq:man-unmatched-lb-aux-2}
		\mathbb{P}\left(
			R \leq \frac{\epsilon}{8}n e^{-d\hat{p}}
		\right)
		\leq\
		\mathbb{P}\left(
			\sum_{j=1}^{\frac{\epsilon}{8}n e^{-d\hat{p}}} (d + S_j) \geq \frac{\epsilon}{2} \tau^*
		\right)
		=\ \mathbb{P}\left(
		\sum_{j=1}^{\frac{\epsilon}{8}n e^{-d\hat{p}}}  S_j \geq \frac{\epsilon}{2} \tau^* - \frac{\epsilon}{8}nd e^{-d\hat{p}}
		\right) \, .
	\end{align}
	
	We now proceed to bound the RHS of \eqref{eq:man-unmatched-lb-aux-2}.
	Note that
	\begin{align*}
		\frac{ \frac{\epsilon}{2} \tau^*  - \frac{\epsilon}{8}nd e^{-d\hat{p}} }{ \frac{\epsilon}{8}n e^{-d\hat{p}} }
		=
		\frac{ 4 n \sqrt{d} }{ n e^{-\frac{1+\gamma}{1+\frac{\epsilon}{2}} \sqrt{d} } }-d
		=
		4 \sqrt{d} e^{\frac{1+\gamma}{1+\frac{\epsilon}{2}} \sqrt{d} } -d\, .
	\end{align*}
	Recall that $\gamma = \Theta\left( \frac{1}{n^\alpha} \right)$, $\epsilon < 0.2$, and $d = \omega(1)$, we have for large enough $n$,
	$4 \sqrt{d} e^{\frac{1+\gamma}{1+\frac{\epsilon}{2}} \sqrt{d} } -d > 3.9 \sqrt{d} e^{\frac{1+\gamma}{1+\frac{\epsilon}{2}} \sqrt{d} }$.
	Plugging $
	\lambda \defeq 3.9 \sqrt{d} e^{\frac{1+\gamma}{1+\frac{\epsilon}{2}} \sqrt{d} }$ into Lemma \ref{lem:random-sum-chernoff-bound}, we obtain
	\begin{align}
		\mathbb{P}\left(
		\sum_{j=1}^{\frac{\epsilon}{8}n e^{-d\hat{p}}}  S_j \geq \frac{\epsilon}{2} \tau^* - \frac{\epsilon}{8}nd e^{-d\hat{p}}
		\right)
		\leq\ &
		\exp\left(
			- \frac{ \frac{\epsilon}{8}n e^{-d\hat{p}} }{2\lambda^2}  \left( \lambda - \mathbb{E}[S] \right)^2
		\right) \nonumber\\
		\leq\ &
		\exp\left(
		- \frac{ \epsilon n  }{16}
		e^{-\frac{1+\gamma}{1+\frac{\epsilon}{2}} \sqrt{d}}
		  \left( 1 - \frac{\mathbb{E}[S]}{\lambda} \right)^2
		\right) \, . \label{eq:man-unmatched-lb-aux-3}
	\end{align}
	We also have $\hat{p} = \frac{1+\gamma}{(1+\frac{\epsilon}{2})\sqrt{d}} = o(1)$ and thus for large enough $n$,
	\begin{equation*}
		\left( 1 - \hat{p} \right)^{-d}
			\leq
			\left( e^{-\hat{p}-\hat{p}^2} \right)^{-d}
			= e^{\frac{1+\gamma}{1+\frac{\epsilon}{2}}\sqrt{d} + \left( \frac{1+\gamma}{1+\frac{\epsilon}{2}} \right)^2},
	\end{equation*}
	 where we use the fact that $1-x \geq e^{-x-x^2}$ for any $|x| \leq 0.5$.
	Further observe that for large enough $n$,
	\begin{equation*}
		\frac{1+\frac{\epsilon}{2}}{1+\gamma} e^{\left(\frac{1+\gamma}{1+\frac{\epsilon}{2}}\right)^2} \leq 1.2 e
		<3.3\, ,
	\end{equation*}
	and therefore,
	\begin{equation*}
		\mathbb{E}[S]
			=
			\hat{p}^{-1} \left( 1 - \hat{p} \right)^{-d}
			\leq
			\frac{1+\frac{\epsilon}{2}}{1+\gamma}\sqrt{d}
			e^{\frac{1+\gamma}{1+\frac{\epsilon}{2}}\sqrt{d} + \left( \frac{1+\gamma}{1+\frac{\epsilon}{2}} \right)^2}
			\leq
			3.3 \sqrt{d}  e^{\frac{1+\gamma}{1+\frac{\epsilon}{2}}\sqrt{d}}\, .
	\end{equation*}
	For RHS of \eqref{eq:man-unmatched-lb-aux-3}, we deduce that for large enough $n$,
	\begin{align*}
		\exp\left(
			- \frac{\epsilon n}{16}  e^{-\frac{1+\gamma}{1+\frac{\epsilon}{2}} \sqrt{d}} \left( 1 - \frac{\mathbb{E}[S]}{\lambda} \right)^2
		\right)
		&\leq
		\exp\left(
			- \frac{\epsilon n}{16}  e^{-\frac{1+\gamma}{1+\frac{\epsilon}{2}} \sqrt{d}}
			\left(1-\frac{3.3}{3.9}\right)^2
		\right)
		\\&\leq
		\exp\left(
			- \frac{\epsilon n}{800} e^{-\frac{1+\gamma}{1+\frac{\epsilon}{2}} \sqrt{d}}
		\right)\, .
	\end{align*}
	Combining all these results, for large enough $n$, we obtain
	\begin{equation*}
		\mathbb{P}\left(
			R \leq \frac{\epsilon}{8}n e^{-d\hat{p}}
		\right)
		\leq
		\exp\left(
			- \frac{\epsilon n}{800} e^{-\frac{1+\gamma}{1+\frac{\epsilon}{2}} \sqrt{d}}
		\right).
	\end{equation*}

	As a result, we obtain a high probability lower bound on the number of unmatched men for the sample paths satisfying $\overline{p}_{(1+\frac{\epsilon}{2})\tau^*} \leq \hat{p}$:
	\begin{align*}
		&\mathbb{P}\left( \left.
			\delta^m[(1+\epsilon)\tau^*]\leq \frac{\epsilon}{8}n e^{-\frac{1+\gamma}{1+\frac{\epsilon}{2}}\sqrt{d}}
			\right| \overline{p}_{(1+\frac{\epsilon}{2})\tau^*} \leq \frac{1+\gamma}{(1+\frac{\epsilon}{2})\sqrt{d}}
		 \right)\\
		\leq\ &\mathbb{P}\left(
			R\leq \frac{\epsilon}{8}n e^{-\frac{1+\gamma}{1+\frac{\epsilon}{2}}\sqrt{d}}
		\right)\\
		\leq\ &
		\exp\left(
		- \frac{\epsilon n}{800} e^{-\frac{1+\gamma}{1+\frac{\epsilon}{2}}\sqrt{d}}
		\right)\, .
	\end{align*}
	Combining with \eqref{eq:man-unmatched-lb-aux-1}, we obtain
	\begin{align*}
		&\mathbb{P}\left(
		\delta^m[(1+\epsilon)\tau^*]\leq \frac{\epsilon}{8}n e^{-\frac{1+\gamma}{1+\frac{\epsilon}{2}}\sqrt{d}}
		\right) \\
		\leq\ & \mathbb{P}\left( \left.
			\delta^m[(1+\epsilon)\tau^*]\leq \frac{\epsilon}{8}n e^{-\frac{1+\gamma}{1+\frac{\epsilon}{2}}\sqrt{d}}
			\right| \overline{p}_{(1+\frac{\epsilon}{2})\tau^*}\leq \frac{1+\gamma}{(1+\frac{\epsilon}{2})\sqrt{d}}
		\right)
		+
		\mathbb{P}\left( \overline{p}_{(1+\frac{\epsilon}{2})\tau^*} \geq \frac{1+\gamma}{(1+\frac{\epsilon}{2})\sqrt{d}}\right) \\
		\leq\ & \exp\left(
		- \frac{\epsilon n}{800} e^{-\frac{1+\gamma}{1+\frac{\epsilon}{2}}\sqrt{d}}
		\right)
		+
		2\exp\left(
		-\frac{\gamma^2}{32}\frac{n}{d^2}
		\right)\, .
	\end{align*}

	Now we take $\gamma = n^{-1/5}$.
	First observe that, for large enough $n$, since $d = o( \log^2 n )$, we have
	\begin{equation*}
		2 \exp\left( - \frac{\gamma^2}{32} \frac{n}{d^2} \right)
		= 2 \exp\left( - \frac{1}{32} \frac{n^{3/5}}{d^2} \right)
		\leq \frac{1}{2} \exp\left( - \sqrt{n} \right)
		\, ,
	\end{equation*}
	and furthermore, since $\epsilon <0.2$, 
	\begin{equation*}
		e^{-\frac{1+\gamma}{1+\frac{\epsilon}{2}}\sqrt{d}}
			\geq e^{ - (1+\gamma) (1 - \frac{\epsilon}{3} ) \sqrt{d} }
			= e^{ - (1 - \frac{\epsilon}{3} ) \gamma \sqrt{d} } \cdot e^{ - (1-\frac{\epsilon}{3}) \sqrt{d} }
			\geq \frac{1}{2} e^{ - (1-\frac{\epsilon}{3}) \sqrt{d} }
		\, .
	\end{equation*}
	Therefore, we obtain
	\begin{align*}
		\mathbb{P}\left(
			\delta^m[(1+\epsilon)\tau^*]\leq \frac{\epsilon}{16}n e^{-(1-\frac{\epsilon}{3})\sqrt{d}}
		\right)
		\leq\ &
		\mathbb{P}\left(
			\delta^m[(1+\epsilon)\tau^*]\leq \frac{\epsilon}{8}n e^{-\frac{1+\gamma}{1+\frac{\epsilon}{2}}\sqrt{d}}
		\right)
		\\
		\leq\ &
		\exp\left( - \frac{\epsilon n}{800} e^{-\frac{1+\gamma}{1+\frac{\epsilon}{2}}\sqrt{d}} \right)
		+
		2\exp\left( -\frac{\gamma^2}{32}\frac{n}{d^2} \right)
		\\
		\leq\ &
		\exp\left( - \frac{\epsilon n}{1600}  e^{-( 1 - \frac{\epsilon}{3} ) \sqrt{d}} \right)
		+
		\frac{1}{2} \exp\left( - \sqrt{n} \right)
		\\
		\leq\ &
		\exp\left( - \frac{\epsilon n}{1600} e^{- \sqrt{d}} \right)
		+
		\frac{1}{2} \exp\left( - \sqrt{n} \right)		
		\, .
	\end{align*}
	Given that $\epsilon = \omega( \frac{1}{n^{0.49}} )$, we further have for large enough $n$,
	\begin{equation*}
		\frac{\epsilon n}{1600} e^{- \sqrt{d}}  \geq \frac{1}{1600} n^{0.51} e^{-\sqrt{d}} \geq \sqrt{n} + \log 2 \, ,
	\end{equation*}
	thus,
	\begin{equation*}
		\exp\left( - \frac{\epsilon}{1600} n e^{- 2\sqrt{d}} \right) \leq \frac{1}{2} \exp\left( - \sqrt{n} \right) \, ,
	\end{equation*}
	which concludes the proof.
  \end{proof}

\subsubsection{Upper bound on the total number of proposals $\tau$ and men's average rank $\Rmen$ (Proposition \ref{prop:total-proposal-upper-bound})} \label{append:small-medium-d-proof-step1-4}
With the help of the coupling between the extended process and the men-proposing DA, we are now able to prove Proposition \ref{prop:total-proposal-upper-bound}.

\begin{proof}[Proof of Proposition \ref{prop:total-proposal-upper-bound}]
We make use of Lemma~\ref{lem:unmatched-woman-upper-bound}.
Denote $n\sqrt{d}$ by $\tau^*$.
Plug $\epsilon = d^{-\frac{1}{4}}$ in \eqref{eq:unmatched-woman-upper-bound}. For the RHS of \eqref{eq:unmatched-woman-upper-bound} we have
	\begin{align*}
		\exp\left(
		- \frac{1}{2} n d \epsilon^2 e^{-3\sqrt{d}}
		\right)
		=
		\exp\left(
		- \frac{1}{2} n \sqrt{d} e^{-3\sqrt{d}}
		\right)
		\leq
		\exp\left(
		- \frac{1}{2} n e^{-3\sqrt{d}}
		\right)
		\leq
		\exp\left(
		- \sqrt{n}
		\right)\, .
	\end{align*}
	Here the last inequality holds because $d=o(\log^2 n)$, and it follows that for any $\alpha>0$, $e^{-3\sqrt{d}} = \omega\left(\frac{1}{n^{\alpha}}\right)$. Therefore,
	\begin{align}\label{eq:unmatched-woman-upper-bound-refined}
	\mathbb{P}\left(\delta^w[(1+d^{-\frac{1}{4}})\tau^*] > ne^{-\sqrt{d}}\right)
	\leq\
	\exp\left(
	- \sqrt{n}
	\right)\, .
	\end{align}
	
	We further utilize Lemma \ref{lem:unmatched-man-lower-bound}.
	Plug $\epsilon = d^{-\frac{1}{4}}$ in \eqref{eq:unmatched-man-lower-bound}.
	For the LHS of \eqref{eq:unmatched-man-lower-bound}, because $\frac{1}{16}\frac{1}{x}e^{\frac{1}{3}x} \geq e^{\frac{1}{4}x}$ for large enough $x$, we have for large enough $n$,
	\begin{align*}
	\frac{\epsilon}{16}n e^{-(1 - \frac{\epsilon}{3})\sqrt{d}}
	=
	\frac{1}{16}n e^{-\sqrt{d}}
	d^{-\frac{1}{4}} e^{\frac{1}{3}d^{\frac{1}{4}}}
	\geq	n e^{-\sqrt{d}}
	e^{\frac{1}{4}d^{\frac{1}{4}}}\, ,
	\end{align*}
	and hence
	\begin{align}\label{eq:unmatched-man-lower-bound-refined}
		\mathbb{P}\left(
		\delta^m[(1+\epsilon)\tau^*]\leq n e^{-\sqrt{d}}
		e^{\frac{1}{4}d^{\frac{1}{4}}}
		\right)
		\leq\
		\mathbb{P}\left(
		\delta^m[(1+\epsilon)\tau^*]\leq \frac{\epsilon}{16}n e^{-(1 - \frac{\epsilon}{3})\sqrt{d}}
		\right)
		\leq
		\exp\left( - \sqrt{n} \right)
		\, .
	\end{align}
	
	Note that by assumption on the imbalance $k$, i.e., $|k| = O(n e^{-\sqrt{d}})$, there exists some constant $C$ such that $|k| \leq C n e^{-\sqrt{d}}$ for large enough $n$.
	Consequently, since $C+1 \leq e^{\frac{1}{4}d^{\frac{1}{4}}}$ for large enough $d$ (and hence for large enough $n$ as $d = \omega(1)$), we have for large enough $n$,
	\begin{align*}
		|k| \leq C n e^{-\sqrt{d}} \leq n e^{-\sqrt{d}} \left( e^{\frac{1}{4}d^{\frac{1}{4}}} - 1 \right)\, .
	\end{align*}
	Recall that $\tau$ is the smallest $t$ such that
	\begin{align*}
		\delta^m[t] - \delta^w[t] =  k\, ,
	\end{align*}
	where the process $\delta^m[t] - \delta^w[t]$ is non-decreasing over time.
	Therefore, we have
	\begin{align*}
		\mathbb{P}\left(\tau \geq (1+d^{-\frac{1}{4}})\tau^* \right)
			&\leq \mathbb{P}\left(
				\delta^m[(1+d^{-\frac{1}{4}})\tau^*] - \delta^w[(1+d^{-\frac{1}{4}})\tau^*] \leq k
			\right)
			\\&= \mathbb{P}\left(
				\delta^m[(1+d^{-\frac{1}{4}})\tau^*] - \delta^w[(1+d^{-\frac{1}{4}})\tau^*] \leq k, ~
				\delta^w[(1+d^{-\frac{1}{4}})\tau^*] \leq n e^{-\sqrt{d}}
			\right)
			\\& \qquad + \mathbb{P}\left(
				\delta^m[(1+d^{-\frac{1}{4}})\tau^*] - \delta^w[(1+d^{-\frac{1}{4}})\tau^*] \leq k, ~
				\delta^w[(1+d^{-\frac{1}{4}})\tau^*] > n e^{-\sqrt{d}}
			\right)
			\\&\leq \mathbb{P}\left( \delta^m[(1+d^{-\frac{1}{4}})\tau^*] \leq n e^{-\sqrt{d}} + k \right)
				+ \mathbb{P}\left( \delta^w[(1+d^{-\frac{1}{4}})\tau^*] > n e^{-\sqrt{d}} \right)
			\\&\leq \mathbb{P}\left( \delta^m[(1+d^{-\frac{1}{4}})\tau^*] \leq n e^{-\sqrt{d}} e^{\frac{1}{4}d^{\frac{1}{4}}} \right)
				+ \mathbb{P}\left( \delta^w[(1+d^{-\frac{1}{4}})\tau^*] > n e^{-\sqrt{d}} \right)
			\\&\leq 2\exp\left( - \sqrt{n} \right)\, ,
	\end{align*}
where we made use of \eqref{eq:unmatched-woman-upper-bound-refined} and \eqref{eq:unmatched-man-lower-bound-refined} in the last step.

	As a result, when the imbalance satisfies $|k|=O(ne^{-\sqrt{d}})$, with probability $1 - O(\exp(-\sqrt{n}))$, we have
	\begin{align*}
		\tau \leq\ n\left(\sqrt{d} + d^{\frac{1}{4}}\right)\, .
	\end{align*}
	By definition of $\Mrank$, we have
	\begin{align*}
		\Mrank = \frac{\tau + \delta^m}{n+k} \leq \frac{\tau + n}{n+k}\, .
	\end{align*}
	Hence for $
	\tau \leq\ n\left(\sqrt{d} + d^{\frac{1}{4}}\right)\, $, we have for large enough $n$,
	\begin{align*}
		\Mrank \leq \frac{n}{n+k}\left(\sqrt{d} + d^{\frac{1}{4}} + 1\right)
			\leq \left( 1 + 0.5 d^{-\frac{1}{4}} \right) \left(\sqrt{d} + d^{\frac{1}{4}} + 1\right)
			\leq \sqrt{d} + 2 d^\frac{1}{4},
	\end{align*}
	where we utilized the fact that $\frac{n}{n+k} \leq \frac{n}{n-|k|} \leq \frac{1}{1-Ce^{-\sqrt{d}}} \leq 1 + 2C e^{-\sqrt{d}} \leq 1 + 0.5d^{-\frac{1}{4}}$ for large enough $d$.

  \end{proof}

\subsubsection{Lower bounds on the number of unmatched women $\delta^w$ and unmatched men $\delta^m$ (Proposition \ref{prop:unmatched-women-lower-bound})}
\label{append:small-medium-d-proof-step1-5}
We now derive a lower on the number of unmatched women $\delta^w$.
Similar to the proof of Lemma \ref{lem:unmatched-woman-upper-bound}, we again make an analogy between balls-into-bins process and DA procedure, but we now consider a variation of balls-into-bins process that exaggerates the effect of ``sampling without replacement'' as opposed to the original balls-into-bins process that assumes sampling with replacement.
The lower bound on the number of empty bins in this process provides a lower bound on the number of unmatched women $\delta^w$, which immediately leads to a lower bound on the number of unmatched men $\delta^m$ by the identity $\delta^m = \delta^w + k$.

\begin{lem} \label{lem:unmatched-women-lower-bound-preliminary}
	For any $t \geq d$ and $\Delta > 0$, we have
	\begin{equation*}
		\mathbb{P}\left( \frac{\delta^w[t]}{n-d} - \left( 1 - \frac{1}{n-d} \right)^{t-d} \leq - \Delta \right) \leq \exp\left( - 2(n-d) \Delta^2 \right).
	\end{equation*}
	This is also valid for the extended process defined in Section \ref{append:notation}.
\end{lem}
\begin{proof}
	Note that the $t^\text{th}$ proposal goes to a woman chosen uniformly at random after excluding the set of women $\mathcal{H}_t$ that the man has previously proposed to.
	Therefore,
	\begin{equation*}
		\mathbb{P}\left( \left. \text{$t^\text{th}$ proposal goes to one of unmatched women} \right|  \delta^w[t-1], \mathcal{H}_t \right) = \frac{ \delta^w[t-1] }{ n - |\mathcal{H}_t| } \leq \frac{\delta^w[t-1]}{n-d},
	\end{equation*}
	since $|\mathcal{H}_t| \leq d$.
	Consider a process $\underline{\delta}^w[t]$ defined as
	\begin{equation*}
		\underline{\delta}^w[t] = \underline{\delta}^w[t-1] - X_t
		\quad \text{where} \quad
		X_t \sim \text{Bernoulli}\left( \min\left\{  \frac{\delta^w[t-1]}{n-d}, 1 \right\}  \right).
	\end{equation*}
	Since the process $\underline{\delta}^w[t]$ exaggerates the likelihood of an unmatched woman receiving a proposal and hence exaggerates the likelihood of decrementing by $1$ at each level,  $\delta^w[t]$ stochastically dominates $\underline{\delta}^w[t]$: i.e., $\mathbb{P}\left( \delta^w[t] \leq x \right) \leq \mathbb{P}\left( \underline{\delta}^w[t] \leq x \right)$ for all $x \in \mathbb{N}$.
	We also observe that $\underline{\delta}^w[t]$ counts the number of empty bins in a process (we refer to it below as the original process) similar to balls-into-bins process where $d$ bins are occupied during the first $d$ periods, and then the regular balls-into-bins process begins with $n-d$ empty bins.
	Consider Lemma~\ref{lem:balls-into-bins-uniform} applied to the ``modified'' balls-into-bins process of putting $t'$ balls into $n - d$ bins,
	where the bins correspond to those which are not occupied by the first $d$ balls in the original process, and $t'$ is the total number of balls which go into these bins in the original process up to $t$.   
Clearly, $t' \leq t - d$, since the first $d$ balls do not go into these bins. We hence deduce from Lemma~\ref{lem:balls-into-bins-uniform} that
	\begin{align*}
\mathbb{P}\left( \frac{\underline{\delta}^w[t]}{n-d} - \left( 1 - \frac{1}{n-d} \right)^{t-d} \leq - \Delta \right)
&\leq\
\mathbb{P}\left( \frac{\underline{\delta}^w[t]}{n-d} - \left( 1 - \frac{1}{n-d} \right)^{t'} \leq - \Delta \right) \
\\
&\leq\
\exp\left( - 2(n-d) \Delta^2 \right)\, .
	\end{align*}
  \end{proof}

\begin{lem} \label{lem:unmatched-women-lower-bound}
	Consider the setting of Theorem \ref{thm:main-result} and the extended process defined in Section \ref{append:notation}.
	Then there exists $n_0 < \infty$ such that for all $n > n_0$, we have the following lower bounds on the number of unmatched women:
	\begin{align}
		\mathbb{P}\left( \delta^w[(1+d^{-\frac{1}{4}}) n \sqrt{d} ] \leq n e^{-(1 + 2d^{-\frac{1}{4}})\sqrt{d}} \right) &\leq \exp\left( - \sqrt{n} \right),
		\label{eq:unmatched-women-lower-bound-1}
		\\
		\mathbb{P}\left( \delta^w[(1-5d^{-\frac{1}{4}})  n \sqrt{d} ] \leq n e^{-(1-2.5d^{-\frac{1}{4}})\sqrt{d}} \right) &\leq \exp\left( - \sqrt{n} \right).
		\label{eq:unmatched-women-lower-bound-2}
	\end{align}
\end{lem}

\begin{proof}
	Let $\tau^* \triangleq n \sqrt{d}$.
	
	\noindent
	{\bf Proof of \eqref{eq:unmatched-women-lower-bound-1}.} Fix $t = (1+d^{-\frac{1}{4}}) \tau^* = (1+d^{-\frac{1}{4}})n\sqrt{d}$.
	For large enough $n$, we have
	\begin{equation*}
		\frac{d}{n} \leq \frac{d}{e^{\sqrt{d}}} \leq 0.1 d^{-\frac{1}{4}}
		, \quad
		\frac{t-d}{n-d} \leq \frac{t}{n} \cdot \frac{1}{1-d/n} \leq \sqrt{d} \cdot \frac{ 1 + d^{-\frac{1}{4}} }{ 1 - 0.1 d^{-\frac{1}{4}} }
		\leq \sqrt{d} (1 + 1.2 d^{-\frac{1}{4}})
		.
	\end{equation*}
	Consequently, with $\Delta \defeq e^{-(1+2d^{-\frac{1}{4}})\sqrt{d}}$, for large enough $n$ we have
	\begin{align*}
		\frac{n-d}{n} \left[ \left( 1 - \frac{1}{n-d} \right)^{t-d} - \Delta \right]
			&\geq \left( 1 - \frac{d}{n} \right)\cdot \left[ \exp\left( - \frac{t-d}{n-d} - \frac{t-d}{(n-d)^2} \right) - e^{-(1+2d^{-\frac{1}{4}})\sqrt{d}} \right]
			\\&\geq \frac{1}{2} \cdot \left[ \exp\left( - \sqrt{d} ( 1 + 1.2d^{-\frac{1}{4}}) \cdot ( 1+ 1/(n-d) ) \right) - e^{-(1+2d^{-\frac{1}{4}})\sqrt{d}} \right]
			\\&\geq \frac{1}{2} \cdot \left[ \exp\left( - \sqrt{d} ( 1 + 1.4d^{-\frac{1}{4}})  \right) - \exp\left( - \sqrt{d}(1+2d^{-\frac{1}{4}}) \right) \right]
			\\&= e^{-\sqrt{d}} \times \frac{1}{2} \cdot \left( \exp\left( -1.4d^{\frac{1}{4}} \right) - \exp\left( -2 d^{\frac{1}{4}} \right) \right)
			\\&\geq e^{-\sqrt{d}} \times \frac{1}{2} \cdot e^{-2d^{\frac{1}{4}}} \left( 2.0 d^{\frac{1}{4}} - 1.4 d^{\frac{1}{4}} \right)
			\\&= e^{-\sqrt{d}} \times e^{-2d^{\frac{1}{4}}} \times 0.3 d^{\frac{1}{4}}
			\geq e^{-(1+2d^{-\frac{1}{4}})\sqrt{d}}\, .
	\end{align*}
	In the second last inequality, we utilize the fact that $e^{-a} - e^{-b} \geq e^{-b}(b-a)$ for any $0 < a < b$.
	Therefore, by Lemma \ref{lem:unmatched-women-lower-bound-preliminary},
	\begin{align*}
		\mathbb{P}\left( \delta^w[t] \leq n e^{-(1+2d^{-\frac{1}{4}})\sqrt{d}} \right)
			&= \mathbb{P}\left( \frac{\delta^w[t]}{n} \leq  e^{-(1+2d^{-\frac{1}{4}})\sqrt{d}} \right)
			\\&\leq \mathbb{P}\left( \frac{\delta^w[t]}{n} \leq  \frac{n-d}{n} \left( \left( 1 - \frac{1}{n-d} \right)^{t-d} - \Delta \right) \right)
			\\&\leq \mathbb{P}\left( \frac{\delta^w[t]}{n-d} - \left( 1 - \frac{1}{n-d} \right)^{t-d} \leq - \Delta \right)
			\\&\leq \exp\left( - 2(n-d) \Delta^2 \right) = \exp\left( - 2 (n-d) e^{-2(1+d^{-\frac{1}{4}})\sqrt{d}} \right).
	\end{align*}
	The claim follows from the fact that $2 (n-d)e^{-2(1+d^{-\frac{1}{4}})\sqrt{d}} \geq \sqrt{n}$ for large enough $n$.
	\\
	
	\noindent
	{\bf Proof of \eqref{eq:unmatched-women-lower-bound-2}.} Fix $t = (1-5d^{-\frac{1}{4}}) \tau^* = (1 - 5d^{-\frac{1}{4}})n\sqrt{d}$.
	For large enough $n$, we have
	\begin{equation*}
		\frac{d}{n} \leq \frac{d}{e^{\sqrt{d}}} \leq 0.1 d^{-\frac{1}{4}}
		, \quad
		\frac{t-d}{n-d} \leq \frac{t}{n} \cdot \frac{1}{1-d/n} \leq \sqrt{d} \cdot \frac{1 - 5d^{-\frac{1}{4}} }{ 1 - 0.1 d^{-\frac{1}{4}} }
		\leq \sqrt{d} (1 - 4.8 d^{-\frac{1}{4}})
		,
	\end{equation*}
	Consequently, with $\Delta \defeq e^{-(1-2.5d^{-\frac{1}{4}})\sqrt{d}}$,
	\begin{align*}
		\frac{n-d}{n} \left[ \left( 1 - \frac{1}{n-d} \right)^{t-d} - \Delta \right]
			&\geq \left( 1 - \frac{d}{n} \right)\cdot \left[ \exp\left( - \frac{t-d}{n-d} - \frac{t-d}{(n-d)^2} \right) - e^{-(1-2.5d^{-\frac{1}{4}})\sqrt{d}} \right]
			\\&\geq \frac{1}{2} \cdot \left[ \exp\left( - \sqrt{d} ( 1 - 4.8d^{-\frac{1}{4}}) \cdot ( 1+ 1/(n-d) ) \right) - e^{-(1-2.5d^{-\frac{1}{4}})\sqrt{d}} \right]
			\\&\geq \frac{1}{2} \cdot \left[ \exp\left( - \sqrt{d} ( 1 - 4.6d^{-\frac{1}{4}})  \right) - \exp\left( - \sqrt{d}(1-2.5d^{-\frac{1}{4}}) \right) \right]
			\\&= e^{-\sqrt{d}} \times \frac{1}{2} \cdot \left( \exp\left( 4.6d^{\frac{1}{4}} \right) - \exp\left( 2.5d^{\frac{1}{4}} \right) \right)
			\\&\stackrel{\textup(a)}{\geq} e^{-\sqrt{d}} \times \frac{1}{2} \cdot e^{2.5d^{\frac{1}{4}}} \left( 4.6 d^{\frac{1}{4}} - 2.5 d^{\frac{1}{4}} \right) 
			\\&\geq e^{-\sqrt{d}} \times e^{2.5d^{\frac{1}{4}}} \times d^{\frac{1}{4}}
			\geq e^{-(1-2.5d^{-\frac{1}{4}})\sqrt{d}}\, ,
	\end{align*}
	for large enough $n$. Here (a) follows from the fact that $f(x)=e^x$ is convex hence $f(x_2)-f(x_1)\geq f'(x_1)(x_2-x_1)$ for $x_2>x_1$.
	Therefore, by Lemma \ref{lem:unmatched-women-lower-bound-preliminary},
	\begin{align*}
		\mathbb{P}\left( \delta^w[t] \leq n e^{-(1-2.5d^{-\frac{1}{4}})\sqrt{d}} \right)
			&= \mathbb{P}\left( \frac{\delta^w[t]}{n} \leq  e^{-(1-2.5d^{-\frac{1}{4}})\sqrt{d}}\right)
			\\&\leq \mathbb{P}\left( \frac{\delta^w[t]}{n} \leq  \frac{n-d}{n} \left( \left( 1 - \frac{1}{n-d} \right)^{t-d} - \Delta \right) \right)
			\\&\leq \mathbb{P}\left( \frac{\delta^w[t]}{n-d} - \left( 1 - \frac{1}{n-d} \right)^{t-d} \leq - \Delta \right)
			\\&\leq \exp\left( - 2(n-d) \Delta^2 \right) = \exp\left( - 2 (n-d) e^{-2(1-2.5d^{-\frac{1}{4}})\sqrt{d}} \right).
	\end{align*}
	The claim follows from the fact that $2 (n-d) e^{-2(1-2.5d^{-\frac{1}{4}})\sqrt{d}} \geq \sqrt{n}$ for large enough $n$.
  \end{proof}

We are now able to prove Proposition \ref{prop:unmatched-women-lower-bound}.
\begin{proof}[Proof of Proposition \ref{prop:unmatched-women-lower-bound}]	By Proposition \ref{prop:total-proposal-upper-bound} and the monotonicity of $\delta^w[t]$, we have for large enough $n$,
	\begin{align}
		\mathbb{P}\left(
		\delta^{w}[\tau] \leq e^{-2d^{\frac{1}{4}}}ne^{-\sqrt{d}}
		\right)&\
		\leq
		\mathbb{P}\left(
		\delta^{w}[\tau] \leq e^{-2d^{\frac{1}{4}}} ne^{-\sqrt{d}}\, ,
		\tau \leq (1+d^{-\frac{1}{4}})\tau^*
		\right)
		+
		\mathbb{P}\left(
		\tau \geq (1+d^{-\frac{1}{4}})\tau^*
		\right) \nonumber \\
		&\ \leq
		\mathbb{P}\left(
		\delta^{w}[(1+d^{-\frac{1}{4}})\tau^*] \leq e^{-2d^{\frac{1}{4}}} ne^{-\sqrt{d}}
		\right)
		+
		\exp\left(
		-\sqrt{n}
		\right)\, .\label{eq:unmatched-woman-lower-bound-aux-1}
	\end{align}
	Moreover, by Lemma \ref{lem:unmatched-women-lower-bound}, we have for large enough $n$,
	\begin{align*}
		\mathbb{P}\left(
		\delta^{w}[(1+d^{-\frac{1}{4}})\tau^*] \leq e^{-2d^{\frac{1}{4}}} ne^{-\sqrt{d}}
		\right)
		\leq\
		\exp\left(-\sqrt{n}\right)\, .
	\end{align*}
	From \eqref{eq:unmatched-woman-lower-bound-aux-1}, we conclude that with probability $1 - 2 \exp(-\sqrt{n})$,
	\begin{align*}
	\delta^w \geq\ ne^{-\sqrt{d} - 2 d^{\frac{1}{4}}}
	\, .
	\end{align*}
	Since $|\delta^m - \delta^w| = |k|\leq O(ne^{-\sqrt{d}}) $, it follows that with probability $1 - 2\exp(-\sqrt{n})$,
	\begin{align*}
	\delta^m \geq\ ne^{-\sqrt{d} - 3 d^{\frac{1}{4}}}
	\, .
	\end{align*}
  \end{proof}

\subsection{Step 2: Lower Bound on the Total Number of Proposals $\tau$} \label{append:small-medium-d-proof-step2}
In this section, we prove the following two propositions.

\begin{prop}\label{prop:unmatched-women-upper-bound}
	Consider the setting in Theorem \ref{thm:main-result}. With probability $1 - O\left(\exp\left(-d^\frac{1}{4}\right)\right)$ ,we have the following upper bounds on the number of unmatched men $\delta^m$ and unmatched women $\delta^w$:
	\begin{align*}
	\delta^m \leq n \exp\left( -\sqrt{d} + 2.5 d^\frac{1}{4} \right)
	, \qquad
	\delta^w \leq n \exp\left( -\sqrt{d} + 2.5d^\frac{1}{4} \right)\, .
	\end{align*}
\end{prop}

\begin{prop}\label{prop:total-proposal-lower-bound}
	Consider the setting of Theorem \ref{thm:main-result}. With probability $1 -
O\left(\exp\left(-d^\frac{1}{4}\right)\right)$, we have the following lower bound on the total number of proposals and men's average rank under the men-optimal stable matching:
	\begin{align*}
	\tau \geq\ n\left(\sqrt{d} - 5d^{\frac{1}{4}}\right)\, ,\qquad
	\Mrank
	\geq\
	\sqrt{d}
	-
	6d^{\frac{1}{4}}
	\, .
	\end{align*}
\end{prop}

The proofs of Proposition \ref{prop:unmatched-women-upper-bound} and \ref{prop:total-proposal-lower-bound} have the following structure:
\begin{itemize}
\item (Sections~\ref{append:small-medium-d-proof-step2-1} and \ref{append:small-medium-d-proof-step2-2})
	Proof of Proposition \ref{prop:unmatched-women-upper-bound}:
	We first derive an upper bound on the expected number of unmatched men $\mathbb{E}[\delta^m]$ in Lemma \ref{lem:unmatched-women-expected-ub}, utilizing the fact that the probability of the last proposing man being rejected cannot be too large given that the total number of proposals $\tau$ is limited by its upper bound (Proposition \ref{prop:total-proposal-upper-bound}).
	We immediately deduce an upper bound $\mathbb{E}[\delta^w]$ by using the identity $\delta^m = \delta^w + k$.
	The high probability upper bounds on $\delta^m$ and $\delta^w$ follow by applying Markov's inequality.

\item (Section~\ref{append:small-medium-d-proof-step2-3})
	Proof of Proposition \ref{prop:total-proposal-lower-bound}: We obtain a lower bound on the total number of proposals $\tau$ by showing that the current number of unmatched women $\delta^w[t]$ does not decay fast enough (again argued with a balls-into-bins analogy) and hence it will violate the upper bound on $\delta^w[\tau]$ ($=\delta^w$) derived in Proposition \ref{prop:unmatched-women-upper-bound} if $\tau$ is too small.
The lower bound on $\tau$ immediately translates into the lower bound on
$\Mrank$ due to the identity $\Mrank = \frac{\tau + \delta^m}{n+k}$.

\end{itemize}

\subsubsection{Upper bound on the expected number of unmatched women $\mathbb{E}[\delta^w]$} \label{append:small-medium-d-proof-step2-1}
Using a careful analysis of the rejection chains triggered by the last proposing man's proposal, we are able to derive an upper bound on the expected number of unmatched women.
\begin{lem}\label{lem:unmatched-women-expected-ub}
	Consider the setting of Theorem~\ref{thm:main-result}. There exists $n_0< \infty$ such that for all $n > n_0$, we have the following upper bounds on the expected number of unmatched men and women under stable matching
	\begin{align}
	\Ex [\delta^m] \leq n \exp(-\sqrt{d} + 1.4d^{1/4}) \, , \quad
	\Ex [\delta^w] \leq n \exp(-\sqrt{d} + 1.5d^{1/4}) \, .
	\end{align}
\end{lem}

\begin{proof}
	We will track the progress of the man proposing DA algorithm making use of the principle of deferred decisions, and further make use of a particular sequence of proposals: we will specify beforehand an arbitrary man $i$ (before any information whatsoever is revealed), and then run DA to convergence on the other men, before man $i$ makes a single proposal.  We will show that the probability that the man $i$ remains unmatched is bounded as
	\begin{align}
	\prob(\mu(i) = i) \leq \exp(-\sqrt{d} + 1.4d^{1/4})
	\label{eq:prob-i-unmatched-ub}
	\end{align}
	for large enough $n$.
	This will imply that, by symmetry across men, the expected number of unmatched men under stable matching will be bounded above as
	$$\Ex[\delta^m] \leq  (n+k)\exp(-\sqrt{d} + 1.4 d^{1/4})\, .$$
	Finally the number of unmatched women at the end is exactly $\delta^w = \delta^m -k$, and so
	$$\Ex[\delta^w] = \Ex[\delta^m] -k \leq (n+k)\exp(-\sqrt{d} + 1.4d^{1/4}) - k \leq n \exp(-\sqrt{d} + 1.5d^{1/4})$$ for large enough $n$
	as required, using $k = O(ne^{-\sqrt{d}})$. The rest of proof is devoted to establishing \eqref{eq:prob-i-unmatched-ub}.
	
	Using Proposition~\ref{prop:total-proposal-upper-bound}, we have that with probability $1 - O(\exp(-\sqrt{n}))$, at the end of DA, $\tau$ is bounded above as
	\begin{align}
	\tau
	\leq n \big (
	\sqrt{d}
	+
	d^{\frac{1}{4}} \big )
	\, ,
	\label{eq:tauprime-ub}
	\end{align}
	and using Proposition~\ref{prop:unmatched-women-lower-bound}, we have that with probability $1 - O(\exp(-\sqrt{n}))$,
	\begin{align}
	\delta^w \geq ne^{-\sqrt{d} - 2 d^{\frac{1}{4}}}
	\, ,
	\label{eq:deltaprime-lb}
	\end{align}
	at the end of DA. Note that if \eqref{eq:tauprime-ub} holds at the end of DA, then the RHS of \eqref{eq:tauprime-ub} is an upper bound on $t$ throughout the run of DA. Similarly, since the number of unmatched woman $\delta^w[t]$ is monotone non-increasing in $t$, if \eqref{eq:deltaprime-lb} holds at the end of DA, then the RHS of \eqref{eq:deltaprime-lb} is a lower bound on $\delta^w[t]$ throughout the run of DA.
	If, at any stage during the run of DA either \eqref{eq:tauprime-ub} (with $t$ instead of $\tau$) or \eqref{eq:deltaprime-lb} (with $\delta^w[t]$ instead of $\delta^w$) 
	is violated, declare a ``failure'' event $\Ev \equiv \Ev_\tau$. By union bound, we know that $\prob(\Ev) = O(\exp(-\sqrt{n}))$. For $t \leq \tau$, let $\Ev_t$ denote the event that no failure has occurred during the first $t$ proposals of DA. We will prove \eqref{eq:prob-i-unmatched-ub} by showing an upper bound on the likelihood that man $i$ remains unmatched for sample paths where no failure occurs, and assuming the worst (i.e., that $i$ certainly remains unmatched) in the rare cases where there is a failure.
	
	Run DA to convergence on men besides $i$. Now consider proposals by $i$. At each such proposal, the recipient woman is drawn uniformly at random from among at least $n - d+ 1$ ``candidate''  women (the ones to whom $i$ has not yet proposed). Assuming $\Ev_t^c$, we know that
	\begin{align}
	t\leq n \big (\sqrt{d} + d^{\frac{1}{4}} \big )\, ,
	\label{eq:t-ub}
	\end{align}
	and hence the total number of proposals received by candidate women is at most $ n \big (\sqrt{d} + d^{\frac{1}{4}} \big )$, and hence the average number of proposals received by candidate women is at most $n(\sqrt{d}	+ d^{\frac{1}{4}})/(n-d+1) \leq \sqrt{d}(1	+ d^{-1/4} + \log^2 n/n) \leq \sqrt{d}(1	+ 1.1d^{-\frac{1}{4}}) \leq \sqrt{d} + 1.1d^{\frac{1}{4}}$ for large enough $n$, using $d = o(\log^2 n)$. If the proposal goes to woman $j$, the probability of it being accepted is
	$\frac{1}{w_{j,t}+1}$. Averaging over the candidate women and using Jensen's inequality for the function $f(x) = \frac{1}{x+1}$, the probability of the proposal being accepted is at least $\frac{1}{\sqrt{d} + 1.1d^{1/4}+1} \geq \frac{1}{\sqrt{d} + 1.2d^{1/4}}$. If the proposal is accepted, say by woman $j$, this triggers a rejection chain. We show that it is very unlikely that this rejection chain will cause an additional proposal to woman $j$ (which will imply that it is very unlikely that the rejection chain will cause $i$ himself to be rejected): For every additional proposal in the rejection chain, the likelihood that it goes to  an unmatched woman far exceeds the likelihood that it goes to woman $j$: if the current time is $t'$ and $\Ev_{t'}^c$ holds, then, since all $\delta^w[t']$ unmatched women are certainly candidate recipients of the next proposal, the likelihood of the proposal being to an unmatched woman is at least $\delta^w[t'] \geq ne^{-\sqrt{d} - 2 d^{\frac{1}{4}}} \geq \sqrt{n}$ times the likelihood of it being to woman $j$ for $n$ large enough, using $d = o(\log^2 n)$. Now if the proposal is to an unmatched woman, this causes the rejection chain to terminate, hence the expected number of proposals to an unmatched woman in the rejection chain is at most $1$. We immediately deduce that if a failure does not occur prior to termination of the chain, the expected number of proposals to woman $j$ in the rejection chain is at most $\frac{1}{\sqrt{n}}$. It follows that
	\begin{align*}
	&  \prob(i \textup{ is displaced from $j$ by the rejection chain triggered when $j$ accepts his proposal}) \\
	\leq \; & \prob(j \textup{ receives a proposal in the rejection chain triggered}) \\
	\leq \; & \Ex [\textup{Number of proposals received by $j$ in the rejection chain triggered}]\\
	\leq \; &  \frac{1}{\sqrt{n}} \, ,
	\end{align*}
	for $n$ large enough.
	Overall, the probability of the proposal by $i$ being ``successful'' in that it is both (a) accepted, and then (b) man $i$ is not pushed out by the rejection chain, is at least
	\begin{align*}
	\frac{1}{\sqrt{d} + 1.2d^{1/4}} \left ( 1- \frac{1}{\sqrt{n}} \right )
	\leq \frac{1}{\sqrt{d} + 1.3d^{1/4}}\, ,
	\end{align*}
	for large enough $n$. Hence the probability of an unsuccessful proposal (if there is no failure) is at most
	\begin{align*}
	1- \frac{1}{\sqrt{d} + 1.3d^{1/4}} \leq \exp\left \{ - \frac{1}{\sqrt{d} + 1.3d^{1/4}} \right \} \, ,
	\end{align*}
	and so the probability of all $d$ proposals being unsuccessful (if there is no failure) is at most
	\begin{align*}
	\exp\left \{ - \frac{d}{\sqrt{d} + 1.3d^{1/4}} \right \} \leq \exp \left \{ - \sqrt{d} + 1.3 d^{1/4}\right \} \, .
	\end{align*}
	Formally, what we have obtained is an upper bound on the quantity $\Ex\big [ \ind(\mu(i) = i)\ind(\Ev^c)\big ]$, namely,
	\begin{align*}
	\Ex\big [ \ind(\mu(i) = i)\ind(\Ev^c)\big ] \leq \exp \left \{ - \sqrt{d} + 1.3 d^{1/4}\right \}\, .
	\end{align*}
	
	Since the probability of failure is bounded as $\prob(\Ev) \leq O(\exp(-\sqrt{n}))$, the overall probability that of man $i$ remaining unmatched is bounded above as
	\begin{align*}
	\prob(\mu(i) = i) &\leq \Ex\big [ \ind(\mu(i) = i) \ind(\Ev^c)\big ] + \prob(\Ev)\\
	&\leq  \exp \left \{ - \sqrt{d} + 1.3 d^{1/4}\right \}  + O(\exp(-\sqrt{n})) \leq \exp \left \{ - \sqrt{d} + 1.4 d^{1/4}\right \}\,
	\end{align*}
	for large enough $n$, i.e., the bound \eqref{eq:prob-i-unmatched-ub} which we set out to show.
  \end{proof}

\subsubsection{Upper bound on the number of unmatched men $\delta^m$ and unmatched women $\delta^w$ (Proposition \ref{prop:unmatched-women-upper-bound})}
\label{append:small-medium-d-proof-step2-2}

\begin{proof}{Proof of Proposition \ref{prop:unmatched-women-upper-bound}.}
	Recall the results in Lemma \ref{lem:unmatched-women-expected-ub}:
	\begin{align}
	\Ex [\delta^m] \leq n \exp(-\sqrt{d} + 1.4d^{1/4}) \, , \quad
	\Ex [\delta^w] \leq n \exp(-\sqrt{d} + 1.5d^{1/4}) \, .
	\end{align}
	We use Markov's inequality for each $\delta^m$ and $\delta^w$:
	\begin{align*}
	\mathbb{P}\left(\delta^m > n \exp(-\sqrt{d} + 2.4d^{1/4})\right)
	\leq\
	\frac{\mathbb{E}[\delta^m]}{n \exp(-\sqrt{d} + 2.4d^{1/4})}	
	\leq\
	\exp(-d^{1/4})\, ,
	\\
	\mathbb{P}\left(\delta^w> n \exp(-\sqrt{d} + 2.5d^{1/4})\right)
	\leq\
	\frac{\mathbb{E}[\delta^w]}{n \exp(-\sqrt{d} + 2.5d^{1/4})}	
	\leq\
	\exp(-d^{1/4})\, .
	\end{align*}
  \end{proof}

\subsubsection{Lower bound on the number of total proposals $\tau$ (Proposition \ref{prop:total-proposal-lower-bound})} \label{append:small-medium-d-proof-step2-3}

\begin{proof}{Proof of Proposition \ref{prop:total-proposal-lower-bound}.}
	Consider the extended process defined in Appendix~\ref{append:notation}, and let $\delta^w[t]$ be the number of unmatched woman at time $t$ of the extended process.
	Let $\tau$ be the time when the men-optimal stable matching is found, i.e., $\delta^w = \delta^w[\tau]$.
	Let $\epsilon \triangleq d^{-1/4}$. We have
	\begin{align}
		\mathbb{P}\left(\tau<(1-5\epsilon)n\sqrt{d}\right)
		\leq &\
		\mathbb{P}\left(\tau<(1-5\epsilon)n\sqrt{d},\  \delta^w[\tau]< n e^{-(1-2.5\epsilon)\sqrt{d}}\right)
		+
		\mathbb{P}\left(\delta^w[\tau] \geq n e^{-(1-2.5\epsilon)\sqrt{d}}\right) \nonumber\\
		\leq &\
		\mathbb{P}\left(\delta^w[(1-5\epsilon)n\sqrt{d}]<ne^{-(1-2.5\epsilon)\sqrt{d}}\right)
		+
		\mathbb{P}\left(\delta^w[\tau] \geq ne^{-(1-2.5\epsilon)\sqrt{d}}\right)\, . \label{eq:total-proposal-lb-aux-1}
	\end{align}
	Here the last inequality holds because $\delta^w[t]$ is non-increasing over $t$ on each sample path.
	It follows from Proposition \ref{prop:unmatched-women-upper-bound} that the second term on the RHS of \eqref{eq:total-proposal-lb-aux-1} is $O(e^{-d^{1/4}})$.
	
	It remains to bound the first term on the RHS of \eqref{eq:total-proposal-lb-aux-1}.
	By Lemma \ref{lem:unmatched-women-lower-bound}, we have
	\begin{align*}
	 \mathbb{P}\left(\delta^w[(1-5\epsilon)n\sqrt{d}] < n e^{-(1-2.5\epsilon)\sqrt{d}} \right) \leq \exp( - \sqrt{n} ),
	\end{align*}
	for large enough $n$.
	By plugging this in the RHS of \eqref{eq:total-proposal-lb-aux-1}, we obtain
	\begin{align}
		\mathbb{P}\left(\tau<(1-5\epsilon)n\sqrt{d}\right)
		=\
		O\left( \exp(-d^\frac{1}{4}) \right)\, .
\label{eq:tau-lb}
	\end{align}
	
Note that by the definition of $\Mrank$, we have
	\begin{align*}
	\Mrank \geq \frac{\tau}{n+k}\, .
	\end{align*}
	Since $|k| = O(ne^{-\sqrt{d}})$, using an argument similar to the one at the end of the proof of Proposition \ref{prop:total-proposal-upper-bound}, we can 
deduce from \eqref{eq:tau-lb}
that
	\begin{align*}
	\mathbb{P}\left(\Mrank<(1-6\epsilon)\sqrt{d}\right)
	=\
	O\left( \exp\big(-d^\frac{1}{4} \big ) \right)\, .
	\end{align*}
	This concludes the proof.
  \end{proof}

\subsection{Step 3: Upper and Lower Bounds on Women's Average Rank $\Rwomen$} \label{append:small-medium-d-proof-step3}
In this section, we prove the following two propositions.

\begin{prop}[Lower bound on women's average rank]\label{prop:women-ranking-lower-bound}
	Consider the setting of Theorem \ref{thm:main-result}.
	With probability $1 - \frac{3}{n}$, we have the following lower bound on women's average rank:
	\begin{align*}
	\Wrank
	\geq\
	\sqrt{d} - 3 d^{\frac{1}{4}}
	\, ,
	\end{align*}
\end{prop}

\begin{prop}[Upper bound on women's average] \label{prop:women-ranking-upper-bound}
	Consider the setting of Theorem \ref{thm:main-result}. With probability $1 - O(\exp(-d^\frac{1}{4}))$, we have the following upper bound on the women's average rank:
	\begin{align*}
	\Wrank
	\leq\
	\sqrt{d}
	+
	8d^{\frac{1}{4}}
	\, .
	\end{align*}
\end{prop}

In order to characterize the women side, we introduce a different extended process which we call the \emph{continue-proposing process} that is slightly different from one introduced in Section \ref{append:notation}.
Until the MOSM is found (i.e., $t \leq \tau$), the continue-proposing process is identical to the original DA procedure.
After the MOSM is found (i.e., $t > \tau$), the proposing man $I_t$ is chosen arbitrarily among the men who have not yet exhausted their preference list (i.e., $\{i \in \mathcal{M}: M_{i,t-1} < d\}$), and we let him propose to his next candidate.
We do not care about the matching nor the acceptance/rejection after $\tau$, since we only keep track of the number of proposals that each man has made, $M_{i,t}$, and each woman has received, $W_{j,t}$.
The continue-proposing process terminates at time $t = (n+k)d$, when all men exhaust their preference lists.

To analyze the concentration of $\Rwomen$, we first construct upper and lower bounds on its conditional expectation.
More formally, we define
\begin{equation} \label{eq:women-ranking-mean}
	\bar{R}[t] \triangleq \frac{1}{n} \sum_{j \in \mathcal{W}} \frac{W_{j,(n+k)d} - W_{j,t}}{W_{j,t}+1},
\end{equation}
where $W_{j,(n+k)d}$ represents the degree of woman $j$ in a random matching market so that $W_{j,(n+k)d} - W_{j,t}$ represents the number of remaining proposals that woman $j$ will receive after time $t$.
In Lemma \ref{lem:women-ranking-preliminary2}, we prove that $\bar{R}[\tau]$ is concentrated around $\sqrt{d}$ given $\tau \approx n \sqrt{d}$.
In Lemma \ref{lem:women-ranking-preliminary3}, we show that $\bar{R}[\tau]$ (plus 1) is indeed the conditional expectation of $\Rwomen$ given $W_{j,\tau}$'s and $W_{j,(n+k)d}$'s, and further characterize the conditional distribution of $\Rwomen$ given $\bar{R}[\tau]$, which leads to the concentration bounds on $\Rwomen$.
Within the proofs, we also utilize the fact that $\bar{R}[t]$ is decreasing over time on each sample path.
	
\subsubsection{Concentration of expected women's average rank $\bar{R}_t$}
We first state a preliminary lemma that will be used to show the concentration of $\bar{R}_t$.

\begin{lem} \label{lem:women-ranking-preliminary}
	Fix any $t$ and $T$ such that $t < T$ and positive numbers $c_1,\ldots,c_n$ such that $c_j \in [0,1]$ for all $j$, and define
	\begin{equation*}
		Y_{t,T} \defeq \sum_{j \in \mathcal{W}} c_j (W_{j,T}-W_{j,t}).
	\end{equation*}
	With $S \defeq \sum_{j=1}^n c_j$, we have
	\begin{align}
		\mathbb{P}\left( \left. Y_{t,T} \geq (1+\epsilon) \frac{(T-t)S}{n-d} \right| \vec{W}_t \right) &\leq \exp\left( - \frac{1}{4} \epsilon^2 \times \frac{(T-t)S}{n} \right)
\label{eq:Yub}
		\\
		\mathbb{P}\left( \left. Y_{t,T} \leq (1-\epsilon) \frac{(T-t)(S-d)}{n-d} \right| \vec{W}_t \right) &\leq \exp\left( - \frac{1}{4} \epsilon^2 \times \frac{(T-t)(S-d)}{n-d} \right)
\label{eq:Ylb}
	\end{align}
	for any $\epsilon \in [0,1]$.
\end{lem}

\begin{proof}
	Throughout this proof, we assume that $W_{1,t}, \ldots, W_{n,t}$ are revealed, i.e. we consider the conditional probabilities/expectations given $W_{1,t}, \ldots, W_{n,t}$.
	In addition, we assume that $c_1 \leq c_2 \leq \ldots \leq c_n$ without loss of generality.
	
\medskip
	\noindent
	{\bf Proof of \eqref{eq:Yub}:} We first establish an upper bound using a coupling argument.
	Recall that $W_{j,s}$ counts the number of proposals that a woman $j$ had received up to time $s$, which is governed by the recipient process $J_s$.
	We construct a coupled process $\left( \overline{W}_{j,s} \right)_{j \in \mathcal{W}, s \geq t}$ that counts based on $\overline{J}_s$ as follows:
	\begin{enumerate}
		\item[(i)] Initialize $\overline{W}_{j,t} \gets W_{j,t}$ for all $j$.
		\item[(ii)] At each time $s = t+1, t+2, \ldots, T$, after the recipient $J_s$ is revealed (which is uniformly sampled among $\mathcal{W} \setminus \mathcal{H}_s$), determine $\overline{J}_s \in \{d+1,\ldots,n\}$:
		\begin{itemize}
			\item If $J_s \in \{d+1,\ldots,n\}$, set $\overline{J}_s \gets J_s$.
			\item If $J_s \in \{1,\ldots,d\}$, sample $\overline{J}_s$ according to the probability distribution $p_s(\cdot)$ defined as (the motivation for this definition is provided below)
			\begin{equation*}
				p_s(j) = \left\{ \begin{array}{ll}
						0 & \text{if } j \in \{1,\ldots,d\},
						\\
						\frac{1}{n-d} \big/ \frac{ |\{1,\ldots,d\} \setminus \mathcal{H}_s| }{ |\mathcal{W} \setminus \mathcal{H}_s| } & \text{if } j \in \{d+1, \ldots, n\} \cap \mathcal{H}_s,
						\\
						\left( \frac{1}{n-d} - \frac{1}{ |\mathcal{W} \setminus \mathcal{H}_s|} \right) \big/ \frac{ |\{1,\ldots,d\} \setminus \mathcal{H}_s| }{  |\mathcal{W} \setminus \mathcal{H}_s| } & \text{if } j \in \{d+1, \ldots, n\} \setminus \mathcal{H}_s.
					\end{array} \right.
			\end{equation*}
		\end{itemize}
		\item[(iii)] Increase the counter of $\overline{J}_s$ instead of $J_s$: i.e., $\overline{W}_{j,s} \gets \overline{W}_{j,s-1} + \ind\{ \overline{J}_s = j \}$ for all $j$.
	\end{enumerate}
	In words, whenever a proposal goes to one of $d$ women who have smallest $c_j$ values (i.e., when $J_s \in \{1,\ldots,d\}$), we randomly pick one among the other $n-d$ women (i.e., $\overline{J}_s \in \{d+1,\ldots,n\}$) and increase that woman's counter $\overline{W}_{\overline{J}_s}$.
	Otherwise (i.e., when $J_s \in \{d+1,\ldots,n\}$), we count the proposal as in the original process.
	In any case, we have $c_{\overline{J}_s} \geq c_{J_s}$.
	
	Note that we do not alter the proposal mechanism in this coupled process, but just count the proposals in a different way.
	Therefore, we have
	\begin{equation} \label{eq:women-ranking-intermediary-upper-bound}
		\sum_{j \in \mathcal{W}} c_j (W_{j,T} - W_{j,t}) \leq \sum_{j \in \mathcal{W}} c_j(\overline{W}_{j,T} - \overline{W}_{j,t}),
	\end{equation}
	Also note that the (re-)sampling distribution $p_s(\cdot)$ was constructed in a way that $\overline{J}_s$ is chosen uniformly at random among $\{d+1,\ldots,n\}$, unconditioned on $J_s$, independently of $\mathcal{H}_s$.
	More formally, we have for any $j \in \{d+1,\ldots,n\} \setminus \mathcal{H}_s$,
	\begin{align*}
		\mathbb{P}( \overline{J}_s = j | \mathcal{H}_s )
			&= \mathbb{P}( J_s = j | \mathcal{H}_s )
				+\mathbb{P}( \overline{J}_s = j | \mathcal{H}_s, J_s \in \{1,\ldots,d\} ) \cdot \mathbb{P}( J_s \in \{1,\ldots,d\} | \mathcal{H}_s )
			\\&= \frac{1}{|\mathcal{W} \setminus \mathcal{H}_s|} + \left( \frac{1}{n-d} - \frac{1}{ |\mathcal{W} \setminus \mathcal{H}_s|}  \right) = \frac{1}{n-d}.
	\end{align*}
	Similarly it can be verified that $\mathbb{P}( \overline{J}_s = j | \mathcal{H}_s ) = \frac{1}{n-d}$ also for any $j \in \{d+1,\ldots,n\} \cap \mathcal{H}_s$.
	The fact that $|\mathcal{H}_s| < d$ guarantees that $p_s(\cdot)$ is a well-defined probability mass function.
	Therefore,
	\begin{equation*}
		\sum_{j \in \mathcal{W}} c_j(\overline{W}_{j,T} - \overline{W}_{j,t}) \stackrel{\text{d}}{=} \sum_{j=d+1}^n c_j X_j,
	\end{equation*}
	where
	$X_j \sim \text{Binomial}\left( T - t, \frac{1}{n-d} \right)$ for $j \in \{d+1,\ldots,n\}$.
	Although $X_j$'s are not independent, they are negatively associated as in the balls-into-bins process (see Section \ref{subsec:balls-into-bins}).
	For any $\lambda \in \mathbb{R}$, $\exp( \lambda c_j X_j)$'s are also NA due to Lemma \ref{lem:NA}--(\ref{lem:NA-monotone}), and therefore,
	\begin{align*}
		\mathbb{E}\left[ \exp\left( \lambda \sum_{j=d+1}^n c_j X_j \right) \right]
			\leq\ &\prod_{j=d+1}^n \mathbb{E}\left[ e^{\lambda c_j X_j} \right] \\
			=\ &\prod_{j=d+1}^n \left( 1 - \frac{1}{n-d} + \frac{1}{n-d} e^{\lambda c_j} \right)^{T-t}\\
			\leq\ &\prod_{j=d+1}^n \exp\left( - \frac{1}{n-d} + \frac{1}{n-d} e^{\lambda c_j} \right)^{T-t}\\
			=\ &\exp\left\{ (T-t)\left( -1 + \frac{1}{n-d} \sum_{j=d+1}^n e^{\lambda c_j} \right) \right\}\, .
	\end{align*}
	Since $c_j \in [0,1]$ and $e^x \leq 1+x+x^2$ for any $x \in (-\infty, 1]$, we have for any $\lambda \in [0,1]$,
	\begin{equation*}
		-1 + \frac{1}{n-d} \sum_{j=d+1}^n e^{\lambda c_j} \leq -1 + \frac{1}{n-d} \sum_{j=d+1}^n (1 + \lambda c_j + \lambda^2 c_j^2 )
			\leq \frac{\lambda + \lambda^2}{n-d} \sum_{j=d+1}^n c_j.
	\end{equation*}
	By Markov's inequality, for any $\lambda \in [0,1]$,
	\begin{align*}
		\mathbb{P}\left( \sum_{j=d+1}^n c_j X_j \geq (1+\epsilon) \frac{T-t}{n-d} \sum_{j=d+1}^n c_j \right)
			&\leq \frac{ \mathbb{E}\left[ \exp\left( \lambda \sum_{j=d+1}^n c_j X_j \right) \right] }{ \exp\left( \lambda(1+\epsilon) \frac{T-t}{n-d} \sum_{j=d+1}^n c_j \right) }
			\\&\leq \exp\left\{ (T-t) \cdot \frac{\lambda + \lambda^2}{n-d} \sum_{j=d+1}^n c_j - \lambda(1+\epsilon) \frac{T-t}{n-d} \sum_{j=d+1}^n c_j \right\}
			\\&\leq \exp\left\{ (\lambda^2 - \lambda \epsilon) \cdot \frac{T-t}{n-d} \sum_{j=d+1}^n c_j \right\}.
	\end{align*}
	By taking $\lambda \defeq \frac{\epsilon}{2}$, we obtain
	\begin{equation*}
		\mathbb{P}\left( \sum_{j=d+1}^n c_j X_j \geq (1+\epsilon) \frac{T-t}{n-d} \sum_{j=d+1}^n c_j \right) \leq \exp\left( - \frac{1}{4} \epsilon^2 \times \frac{T-t}{n-d} \sum_{j=d+1}^n c_j \right).
	\end{equation*}
	Also note that
	\begin{equation*}
		\frac{S}{n} = \frac{1}{n} \sum_{j=1}^n c_j \leq \frac{1}{n-d} \sum_{j=d+1}^n c_j .
	\end{equation*}
	Therefore, together with \eqref{eq:women-ranking-intermediary-upper-bound},
	\begin{align*}
		\mathbb{P}\left( \left. Y_{t,T} \geq (1+\epsilon) \frac{(T-t)S}{n-d} \right| W_{1,t}, \ldots, W_{n,t} \right)
			&\leq \mathbb{P}\left( \left. Y_{t,T} \geq (1+\epsilon) \frac{T-t}{n-d} \sum_{j=d+1}^n c_j \right| W_{1,t}, \ldots, W_{n,t} \right)
			\\&\leq \exp\left( - \frac{1}{4} \epsilon^2 \times \frac{T-t}{n-d} \sum_{j=d+1}^n c_j \right)
			\\&\leq \exp\left( - \frac{1}{4} \epsilon^2 \times \frac{(T-t)S}{n} \right).
	\end{align*}
	
\medskip
	\noindent
	{\bf Proof of \eqref{eq:Ylb}:} Similarly to above, we can construct a coupled process $\left( \underline{W}_{j,s} \right)_{s \geq t}$ under which $\underline{J}_s$ is resampled among $\{1,\ldots,n-d\}$ whenever a proposal goes to one of $d$ women who have largest $c_j$ values (i.e., when $J_s \in \{n-d+1,\ldots,n\}$) while $\mathbb{P}\left( \left. \underline{J}_s = j \right| \mathcal{H}_s \right) = \frac{1}{n-d}$ for any $j \in \{1,\cdots,n-d\}$ and any $\mathcal{H}_s$.
	With this process, we have
	\begin{equation*}
		\sum_{j \in \mathcal{W}} c_j (W_{j,T} - W_{j,t}) \geq \sum_{j \in \mathcal{W}} c_j(\underline{W}_{j,T} - \underline{W}_{j,t}) \stackrel{\text{d}}{=} \sum_{j=1}^{n-d} c_j X_j,
	\end{equation*}
	where $X_j \sim \text{Binomial}\left( T-t, \frac{1}{n-d} \right)$ for $j \in \{1,\ldots,n-d\}$ and $X_j$'s are NA.
	
	For any $\lambda \in [-1,0]$,
	\begin{align*}
	\mathbb{E}\left[ \exp\left( \lambda \sum_{j=1}^{n-d} c_j X_j \right) \right]
	\leq\ &\prod_{j=1}^{n-d} \mathbb{E}\left[ e^{\lambda c_j X_j} \right] \\
	=\ &\prod_{j=1}^{n-d} \left( 1 - \frac{1}{n-d} + \frac{1}{n-d} e^{\lambda c_j} \right)^{T-t}\\
	\leq\ &\prod_{j=1}^{n-d} \exp\left( - \frac{1}{n-d} + \frac{1}{n-d} e^{\lambda c_j} \right)^{T-t}\\
	=\ &\exp\left\{ (T-t)\left( -1 + \frac{1}{n-d} \sum_{j=1}^{n-d} e^{\lambda c_j} \right) \right\}\, .
	\end{align*}
	Since $c_j \in [0,1]$ and $e^x \leq 1+x+x^2$ for any $x \in (-\infty, 1]$, we have for any $\lambda \in [-1,0]$,
	\begin{equation*}
	-1 + \frac{1}{n-d} \sum_{j=1}^{n-d} e^{\lambda c_j} \leq -1 + \frac{1}{n-d} \sum_{j=1}^{n-d} (1 + \lambda c_j + \lambda^2 c_j^2 )
	\leq \frac{\lambda + \lambda^2}{n-d} \sum_{j=1}^{n-d} c_j\, .
	\end{equation*}
	
	Using Markov's inequality, we have
	\begin{align*}
	\mathbb{P}\left( \sum_{j=1}^{n-d} c_j X_j \leq (1-\epsilon) \frac{T-t}{n-d} \sum_{j=1}^{n-d} c_j \right)
	=\ &\mathbb{P}\left( \exp\left(\lambda\sum_{j=1}^{n-d} c_j X_j\right) \geq \exp\left(\lambda(1-\epsilon) \frac{T-t}{n-d} \sum_{j=1}^{n-d} c_j \right) \right)\\
	\leq\ &\frac{ \mathbb{E}\left[ \exp\left( \lambda \sum_{j=1}^{n-d} c_j X_j \right) \right] }{ \exp\left( \lambda(1-\epsilon) \frac{T-t}{n-d} \sum_{j=1}^{n-d} c_j \right) }\\
	\leq\ &\exp\left\{ (T-t) \cdot \frac{\lambda + \lambda^2}{n-d} \sum_{j=1}^{n-d} c_j - \lambda(1-\epsilon) \frac{T-t}{n-d} \sum_{j=1}^{n-d} c_j \right\}
	\\
	\leq\ &\exp\left\{ (\lambda^2 + \lambda \epsilon) \cdot \frac{T-t}{n-d} \sum_{j=1}^{n-d} c_j \right\}\, .
	\end{align*}
	With $\lambda \defeq -\frac{\epsilon}{2}$, we obtain
	\begin{equation*}
		\mathbb{P}\left( \sum_{j=1}^{n-d} c_j X_j \leq (1-\epsilon) \frac{T-t}{n-d} \sum_{j=1}^{n-d} c_j \right) \leq \exp\left( - \frac{1}{4} \epsilon^2 \times \frac{T-t}{n-d} \sum_{j=1}^{n-d} c_j \right).
	\end{equation*}
	Consequently, since $S-d = \sum_{j=1}^n c_j - d \leq \sum_{j=1}^{n-d}c_j$,
	\begin{align*}
		\mathbb{P}\left( \left. Y_{t,T} \leq (1-\epsilon) \frac{(T-t)(S-d)}{n-d} \right| W_{1,t}, \ldots, W_{n,t} \right)
			\leq\ &\mathbb{P}\left( \left. Y_{t,T} \leq (1-\epsilon) \frac{T-t}{n-d} \sum_{j=1}^{n-d} c_j \right| W_{1,t}, \ldots, W_{n,t} \right)
			\\
			\leq\ &\exp\left( - \frac{1}{4} \epsilon^2 \times \frac{T-t}{n-d} \sum_{j=1}^{n-d} c_j \right)
			\\
			\leq\ &\exp\left( - \frac{1}{4} \epsilon^2 \times \frac{(T-t)(S-d)}{n-d} \right)\, .
	\end{align*}
  \end{proof}

\begin{lem} \label{lem:women-ranking-preliminary2}
	Consider the setting of Theorem \ref{thm:main-result} and $\bar{R}[t]$ defined in \eqref{eq:women-ranking-mean}.
	There exists $n_0 < \infty$ such that for all $n > n_0$, we have
	\begin{align}
		\label{eq:women-ranking-mean-upper-bound}
		\mathbb{P}\left( \bar{R}\left[ n (\sqrt{d} + d^{\frac{1}{4}} ) \right] \leq \sqrt{d} - 2.3d^{\frac{1}{4}} \right) &\leq \exp\left( - \frac{n}{8} \right)\, .
		\\
		\label{eq:women-ranking-mean-lower-bound}
		\mathbb{P}\left( \bar{R}\left[ n (\sqrt{d} - 5 d^{\frac{1}{4}} ) \right] \geq \sqrt{d} + 7.5d^{\frac{1}{4}} \right) &\leq \exp\left( - \frac{n}{d^4} \right)\, .
	\end{align}
	
\end{lem}

\begin{proof}
	{\bf Proof of \eqref{eq:women-ranking-mean-upper-bound}:} Fix $t = n\left( \sqrt{d} + d^\frac{1}{4} \right)$ and let $S \triangleq \sum_{j \in \mathcal{W}} \frac{1}{W_{j,t}+1}$.
	Due to the convexity of $f(x) \triangleq \frac{1}{x+1}$, we have for large enough $d$ (i.e. large enough $n$ since $d=\omega(1)$),
	\begin{equation}\label{eq:S-jensen-ineq}
		\frac{S}{n}
			= \frac{1}{n} \sum_{j \in \mathcal{W}} f(W_{j,t})
			\geq f\left( \frac{1}{n} \sum_{j \in \mathcal{W}} W_{j,t} \right)
			= f\left( \frac{t}{n} \right)
			= \frac{1}{t/n+1}
			\geq \frac{1}{\sqrt{d} + 1.05 d^\frac{1}{4}}\, .
	\end{equation}
	Given the asymptotic condition, for large enough $n$,
	\begin{align*}
		\frac{|k|}{n} =&\ O(e^{-\sqrt{d}}) \leq 0.1 d^{-\frac{1}{4}}
		, &
		\frac{t}{nd} \leq&\ \frac{\sqrt{d} + d^{\frac{1}{4}}}{d} = d^{-\frac{1}{2}} + d^{-\frac{3}{4}} \leq 0.1d^{-\frac{1}{4}}
		, \\
		\frac{ S - d }{ n/ \sqrt{d} } \geq&\ \frac{1}{1 + 1.05 d^{-\frac{1}{4}}} - \frac{d^\frac{3}{2}}{n} \geq 1 - 1.1 d^{-\frac{1}{4}}\, . &&
	\end{align*}
	Therefore,
	\begin{align}
		\frac{((n+k)d - t)(S-d)}{n-d}
			\geq\ &\frac{((n+k)d - t)(S-d)}{n} \nonumber\\
			\geq\ & d \left( 1 - \frac{|k|}{n} - \frac{t}{nd} \right) \times \frac{n}{\sqrt{d}} \cdot \frac{S - d}{n/\sqrt{d}}
			\nonumber
			\\
			\geq\ & n\sqrt{d} \times \left( 1 - 0.1 d^{-\frac{1}{4}} - 0.1 d^{-\frac{1}{4}} \right) \times \left( 1 - 1.1 d^{-\frac{1}{4}} \right)
			\nonumber
			\\
			\geq\ & n \sqrt{d} \left( 1 - 1.3 d^{-\frac{1}{4}} \right).
			\label{eq:women-ranking-ratio-lower-bound}
	\end{align}
	Utilizing Lemma \ref{lem:women-ranking-preliminary}, with\footnote{
		In Lemma \ref{lem:women-ranking-preliminary}, we assume that $c_j$'s are some deterministic constants whereas we set $c_ j \defeq \frac{1}{W_{j,t}+1}$ here.
		This is fine because the results of Lemma \ref{lem:women-ranking-preliminary} are stated in terms of conditional probability given $\vec{W}_t$.
	}  $c_j \defeq \frac{1}{W_{j,t}+1}$, $T \triangleq (n+k)d$ and $\epsilon \triangleq d^{-\frac{1}{4}}$, we obtain
	\begin{align*}
		\mathbb{P}\left( \left. \bar{R}[t] \leq \sqrt{d} - 2.3d^\frac{1}{4} \right| \vec{W}_t \right)
			=\ &\mathbb{P}\left( \left. n \bar{R}[t] \leq n\sqrt{d}\left( 1-  2.3d^{-\frac{1}{4}} \right) \right| \vec{W}_t \right)
			\\
			\leq\ &\mathbb{P}\left( \left. n \bar{R}[t] \leq (1-\epsilon) \times n\sqrt{d}\left( 1-  1.3d^{-\frac{1}{4}} \right) \right| \vec{W}_t \right)
			\\
			\leq\ &\mathbb{P}\left( \left. n \bar{R}[t] \leq (1-\epsilon) \times \frac{((n+k)d-t)(S-d)}{n-d} \right| \vec{W}_t \right)
			\\
			\leq\ &\exp\left( - \frac{1}{4} \epsilon^2 \times \frac{((n+k)d-t)(S-d)}{n-d}  \right)
			\\
			\leq\ &\exp\left( - \frac{1}{4} d^{-\frac{1}{2}} \times n\sqrt{d} \left( 1-1.3d^{-\frac{1}{4}} \right)  \right)
			\\
			\leq\ &\exp( - n / 8 )\, ,
	\end{align*}
	where the last inequality follows from the fact that $1-1.3d^{-\frac{1}{4}} \geq \frac{1}{2}$ for large enough $d$.
	Since the above result holds for any realization of $\vec{W}_t$, the claim follows.

	\noindent {\bf Proof of \eqref{eq:women-ranking-mean-lower-bound}:} Fix $t = n(\sqrt{d} - 5 d^{\frac{1}{4}})$.
	Define $\mathcal{A} \triangleq \left\{ \vec{W}_t \left| \frac{1}{n} \sum_{j \in \mathcal{W}} \frac{1}{W_{j,t}+1} \leq \frac{n}{t} + \frac{d^2}{n} + d^{-\frac{3}{4}} \right. \right\} \subset \mathbb{N}^{|\mathcal{W}|}$.
	Applying Lemma \ref{lem:acceptance-probability-upper-bound} with $\Delta \defeq  d^{-\frac{3}{4}}$, we obtain
	\begin{equation*}
		\mathbb{P}\left( \vec{W}_t \notin \mathcal{A} \right)
			= \mathbb{P}\left( \frac{S}{n} > \frac{n}{t} + \frac{d^2}{n} + d^{-\frac{3}{4}} \right)
			\leq 2 \exp\left( - \frac{n \Delta^2}{8d} \right) = 2 \exp\left( - \frac{1}{8} n d^{-\frac{5}{2}} \right)\, .
	\end{equation*}
	Regarding the last term, observe that for large enough $d$, we have
	\begin{equation*}
		2 \exp\left( - \frac{1}{8} n d^{-\frac{5}{2}} \right) \leq  \frac{1}{2} \exp\left( - n d^{-4} \right) \, .
	\end{equation*}
	For any $\vec{W}_t \in \mathcal{A}$, we have for large enough $n$,
	\begin{equation*}
		\frac{S}{n}
			\leq \frac{n}{t} + \frac{d^2}{n} + d^{-\frac{3}{4}}
			= \frac{1}{ \sqrt{d} - 5 d^\frac{1}{4} } + \frac{d^2}{n} + d^{-\frac{3}{4}}
			= \frac{1}{\sqrt{d}} \left( \frac{1}{1-5d^{-\frac{1}{4}}} + \frac{ d^{\frac{5}{2}} }{n} + d^{-\frac{1}{4}} \right)
			\leq \frac{1}{\sqrt{d}} \left( 1 + 6.1 d^{-\frac{1}{4}} \right),
	\end{equation*}
	Furthermore, given the asymptotic conditions,
	\begin{align*}
		\frac{((n+k)d-t)S}{n-d}
			\leq\ &\frac{(n+k)dS}{n-d}\\
			\leq\ &\frac{n+|k|}{n} \cdot \frac{n}{n-d} \cdot dS
			\\
			\leq\ &\frac{n+|k|}{n} \cdot \frac{n}{n-d} \cdot n \sqrt{d} \left( 1 + 6.1 d^{-\frac{1}{4}} \right)
			\\
			=\ &\left( 1 + \frac{|k|}{n} \right) \cdot \frac{1}{1-d/n} \cdot n \sqrt{d} \left( 1 + 6.1 d^{-\frac{1}{4}} \right)
			\\
			\leq\ &\left( 1 + 0.1 d^{-\frac{1}{4}} \right) \cdot \left( 1 + 0.1 d^{-\frac{1}{4}} \right) \cdot n \sqrt{d} \left( 1 + 6.1 d^{-\frac{1}{4}} \right)
			\\
			\leq\ & n \sqrt{d} \left( 1 + 6.4 d^{-\frac{1}{4}} \right)\, ,
	\end{align*}
	where we used the fact that $\frac{|k|}{n} = O(e^{-\sqrt{d}}) \leq 0.1 d^{-\frac{1}{4}}$ and $\frac{d}{n} \leq 0.1 d^{-\frac{1}{4}}$ for large enough $n$.
	We further utilize Lemma \ref{lem:women-ranking-preliminary}: By taking $c_j \triangleq \frac{1}{W_{j,t}+1}$, $T \triangleq (n+k)d$ and $\epsilon \triangleq d^{-\frac{1}{4}}$, we obtain
	\begin{align*}
		\mathbb{P}\left( \left. \bar{R}[t] \geq \sqrt{d} + 7.5 d^\frac{1}{4} \right| \vec{W}_t \right)
			&= \mathbb{P}\left( \left. n \bar{R}[t] \geq n\sqrt{d}\left( 1+ 7.5d^{-\frac{1}{4}} \right) \right| \vec{W}_t \right)
			\\&\leq \mathbb{P}\left( \left. n \bar{R}[t] \geq (1+\epsilon) \times n\sqrt{d}\left( 1+ 6.4d^{-\frac{1}{4}} \right) \right| \vec{W}_t \right)
			\\&\leq \mathbb{P}\left( \left. n \bar{R}[t] \geq (1+\epsilon) \times \frac{((n+k)d-t)S}{n-d} \right| \vec{W}_t \right)
			\\&\leq \exp\left( - \frac{1}{4} \epsilon^2 \times \frac{((n+k)d-t)S}{n}  \right).
	\end{align*}
	for any $\vec{W}_t \in \mathcal{A}$.
	Also note that, from \eqref{eq:women-ranking-ratio-lower-bound}, for large enough $n$, 
	\begin{align*}
		\frac{((n+k)d-t)S}{n}
		\geq\
		\left(\left(1 - \frac{|k|}{n}\right)d - \frac{t}{n}\right)S
		\geq\
		\left(\left(1 - 0.1 d^{-\frac{1}{4}}\right)d - \sqrt{d}+5 d^{\frac{1}{4}}\right)S\, ,
	\end{align*}
	where we used the fact that	$\frac{|k|}{n} = O(e^{-\sqrt{d}}) \leq 0.1 d^{-\frac{1}{4}}$. Because $\left(1 - 0.1 d^{-\frac{1}{4}}\right)d - \sqrt{d}+5 d^{\frac{1}{4}} \geq d-1.3d^{\frac{3}{4}}$ for large enough $n$, and that $S\geq \frac{n}{t/n+1}$ as derived in \eqref{eq:S-jensen-ineq}, we have
		\begin{align*}
		\frac{((n+k)d-t)S}{n}
		\geq\
		\left( d-1.3d^{\frac{3}{4}} \right)\frac{n}{t/n+1}
		\geq\
		\left( d-1.3d^{\frac{3}{4}} \right)\frac{n}{\sqrt{d}}
		\geq\ n \sqrt{d} \left( 1 - 1.3 d^{-\frac{1}{4}} \right)\, ,
	\end{align*}
	
	and therefore,
	\begin{equation*}
		\exp\left( - \frac{1}{4} \epsilon^2 \times \frac{((n+k)d-t)S}{n}  \right)
			\leq \exp\left( - \frac{1}{4} d^{-\frac{1}{2}} \times n \sqrt{d} \left( 1 - 1.3 d^{-\frac{1}{4}} \right) \right)
			\leq \exp\left( - n/8 \right).
	\end{equation*}
	Combining all results, we obtain the desired result: for large enough $n$,
	\begin{align*}
		\mathbb{P}\left( \bar{R}[t] \geq \sqrt{d} + 7.5 d^\frac{1}{4} \right)
			&\leq \mathbb{P}\left( \left. \bar{R}[t] \geq \sqrt{d} + 7.5 d^\frac{1}{4} \right| \vec{W}_t \in \mathcal{A} \right) \cdot \mathbb{P}\left( \vec{W}_t \in \mathcal{A} \right)
				+ \mathbb{P}\left( \vec{W}_t \notin \mathcal{A} \right)
			\\&\leq \exp\left( - n/8 \right) + \frac{1}{2} \exp\left( - n d^{-4} \right)
			\\&  \leq \exp\left( - n d^{-4} \right)  \, ,
	\end{align*}
	where the last inequality follows from that $\frac{n}{8} \geq \frac{n}{d^4} + \log 2$ for large enough $n$ and $d$.
  \end{proof}

\subsubsection{Concentration of women's average rank $\Rwomen$}
The following lemma states that conditioned on $(\vec{W}_\tau, \vec{W}_{(n+k)d})$, $\Wrank$ is concentrated around $\bar{R}[\tau]$.

\begin{lem} \label{lem:women-ranking-preliminary3}
	For any given $n$, $k$ and $d$ and $(\vec{W}_\tau, \vec{W}_{(n+k)d})$ which arises with positive probability, we have $\Ex[\Wrank|\vec{W}_\tau, \vec{W}_{(n+k)d}]  = 1 + \bar{R}[\tau]$. Furthermore,
	for any $\epsilon > 0$ we have
	\begin{align}
		\mathbb{P}\left( \left. \Wrank \geq 1 + (1+\epsilon) \bar{R}[\tau] \,\right| \vec{W}_\tau, \vec{W}_{(n+k)d}\right)
			&\leq
			\exp \left ( -\frac{2\eps^2  n^2 \bar{R}[\tau]^2 }{\sum_{j \in \mathcal{W}}W_{j,(n+k)d}^2 } \right )
		\, , \label{eq:women-rank-concentration-1}
		\\
		\mathbb{P}\left( \left. \Wrank \leq 1 + (1-\epsilon) \bar{R}[\tau] \, \right| \vec{W}_\tau, \vec{W}_{(n+k)d}\right)
			&\leq
			\exp \left ( -\frac{2\eps^2  n^2  \bar{R}[\tau]^2 }{\sum_{j \in \mathcal{W}}W_{j,(n+k)d}^2 } \right )
		\, . \label{eq:women-rank-concentration-2}
	\end{align}
\end{lem}

\begin{proof}
		Within this proof, we assume that $\tau$, $\vec{W}_\tau = \left( W_{j,\tau} \right)_{j \in \mathcal{W}}$, and $\vec{W}_{(n+k)d} = \left( W_{j,(n+k)d} \right)_{j \in \mathcal{W}}$ are revealed (and hence so is $\bar{R}[\tau]$). 
In what follows, $\mathbb{P}( \cdot )$ and $\mathbb{E}[ \cdot ]$ denote the associated conditional probability and the conditional expectation, respectively.

For brevity, let $w_j \triangleq W_{j,\tau}$, $w_j' \triangleq W_{j,(n+k)d} - W_{j,\tau}$, and $R_j \triangleq \text{Rank}_j( \text{MOSM})|\vec{W}_\tau, \vec{W}_{(n+k)d}$. 
Note that a woman $j$ receives $w_j$ proposals until time $\tau$ and receives $w_j'$ proposals after time $\tau$ (the total number of proposals $w_j + w_j' = W_{j,(n+k)d}$ equals to her degree).
Under MOSM, each woman $j$ is matched to her most preferred one among the first $w_j$ proposals, and the rank of her matched partner under MOSM, $R_j$, can be determined by the number of men among the remaining (at time $\tau$) $w_j'$ men on her list that she prefers to her matched partner.

More specifically, fix $j$ and let $Z^j_{t}$ be the indicator that the woman $j$ prefers her $t^\text{th}$ proposal to all of her first $w_j$ proposals for $t \in \{w_j+1, \ldots, w_j+w_j'\}$.
	Then, the rank $R_j$ can be represented as
	\begin{align*}
	R_j &= 1 + \sum_{t=w_j+1}^{w_j+w_j'} \ind\left( \text{woman $j$ prefers her $t^\text{th}$ proposal to all of her first $w_j$ proposals} \right) 
	\\&= 1 + \sum_{t=w_j+1}^{w_j+w_j'}Z_t^j\, .
	\end{align*}
Note that $(Z_t^j)_{t=w_j+1}^{w_j+w_j'}$ has the same distribution as $(\ind\{U_t^j > V_j\})_{t=w_j+1}^{w_j+w_j'}$, where $(U_t^j)_{t=w_j+1}^{w_j+w_j'}$ are i.i.d. Uniform$[0,1]$ random variables, $V_j$ is the largest order statistic of $w_j$ i.i.d. Uniform$[0,1]$ random variables, and $V_j$ is independent of $U_t^j$'s. 
Therefore,
\begin{equation*}
\mathbb{E}\left[ R_j \right] 
=\ 1 + w_j' \cdot \mathbb{E}[Z_{w_j+1}^j]
=\ 1 + w_j' \cdot \mathbb{P}(U_{w_j+1}^j > V_j)
=\ 1 + \frac{w_j'}{w_j+1}\, ,
\end{equation*}
and
\begin{align*}
	\Ex[\Wrank|\vec{W}_\tau, \vec{W}_{(n+k)d}]
	=\ \frac{1}{n}\sum_{j\in \cW}\mathbb{E}[R_j] 
	=\ 1 + \bar{R}[\tau]\, ,
\end{align*}
which proves the first claim in Lemma \ref{lem:women-ranking-preliminary3}.

Note that $(R_j)_{j\in \cW}$ are i.i.d. and that $R_j\in [0,w_j']$. Applying Hoeffding's inequality, we have
\begin{align*}
	\mathbb{P}\left( \left. \Wrank \geq 1 + (1+\epsilon) \bar{R}[\tau] \,\right| \vec{W}_\tau, \vec{W}_{(n+k)d}\right)
	=\ &
	\mathbb{P}\left( \frac{1}{n}R_j
	\geq
	1 + (1+\epsilon)\mathbb{E}[\bar{R}[\tau]]
	\right)\\
	\leq\ &\exp \left ( -\frac{2\eps^2 n^2 \bar{R}[\tau]^2  }{\sum_{j \in \mathcal{W}}W_{j,(n+k)d}^2 } \right )\, .
\end{align*}

Similarly, we can show that
\begin{align*}
\mathbb{P}\left( \left. \Wrank \leq 1 + (1-\epsilon) \bar{R}[\tau] \,\right| \vec{W}_\tau, \vec{W}_{(n+k)d}\right)
\leq\ \exp \left ( -\frac{2\eps^2 n^2 \bar{R}[\tau]^2  }{\sum_{j \in \mathcal{W}}W_{j,(n+k)d}^2 } \right )\, .
\end{align*}
This concludes the proof.

  \end{proof}

\begin{proof}[Proof of Proposition \ref{prop:women-ranking-lower-bound}]
	We obtain a high probability lower bound on $\Rwomen$ by combining the results of Proposition \ref{prop:total-proposal-upper-bound}, and Lemmas \ref{lem:women-ranking-preliminary2} and \ref{lem:women-ranking-preliminary3}.
	By Proposition \ref{prop:total-proposal-upper-bound} and  Lemma \ref{lem:women-ranking-preliminary2} and by the fact that $\bar{R}[t]$ is decreasing on each sample path,
	\begin{align}
		\mathbb{P}\left( \bar{R}[\tau] \leq \sqrt{d} - 2.3d^{\frac{1}{4}} \right)
			\leq\ &\mathbb{P}\left( \bar{R}[\tau] \leq \sqrt{d} - 2.3d^{\frac{1}{4}}, \tau < n (\sqrt{d} + d^{\frac{1}{4}}) \right) + \mathbb{P}\left( \tau \geq n (\sqrt{d}+d^{\frac{1}{4}} ) \right)\nonumber
			\\
			\leq\ &\mathbb{P}\left( \bar{R}[n(\sqrt{d}+d^{\frac{1}{4}})] \leq \sqrt{d} - 2.3d^{\frac{1}{4}}, \tau < n (\sqrt{d} + d^{\frac{1}{4}}) \right) + \mathbb{P}\left( \tau \geq n (\sqrt{d}+d^{\frac{1}{4}} ) \right)\nonumber
			\\
			\leq\ &\exp\left( -\frac{n}{8} \right) + O( \exp(-\sqrt{n}) )
			= O( \exp(-\sqrt{n}) )\, .\label{eq:Rtau-ineq-4}
	\end{align}
	
	We also need a high probability upper bound on $\sum_{j\in \cW} W^2_{j,(n+k)d}$.
	Since $W_{j,(n+k)d}\allowdisplaybreaks \sim\allowdisplaybreaks \text{Binomial} ( (n+k)d, \frac{1}{n} )$, we have for large enough $n$,
	\begin{align*}
	\mathbb{E}\left[ W^2_{j,(n+k)d} \right]
	=\ \mathbb{E}^2\left[ W_{j,(n+k)d}\right] + \textup{Var}\left[ W_{j,(n+k)d}\right]
	=\ (n+k)^2 d^2 \frac{1}{n^2} \left(1 + (1-\frac{1}{n})^2\right)
	\leq\ 2 d^2\, .
	\end{align*}
Denote $\mu\triangleq \mathbb{E}[W_{1,(n+k)d}]=\frac{(n+k)d}{n}$.
Looking up the table of the central moments of Binomial distribution, we have
\begin{align*}
	\mathbb{E}[(W_{1,(n+k)d}-\mu)^4]
	=\ (n+k)d \frac{1}{n}\left(1-\frac{1}{n}\right)\left(1 + (3(n+k)d-6)\frac{1}{n}\left(1-\frac{1}{n}\right)\right)\, .
\end{align*}
Using the fact that $k=o(n)$, $d=o(n)$ and $d=\omega(1)$, we have for large enough $n$,
\begin{align*}
\mathbb{E}[(W_{1,(n+k)d}-\mu)^4]
\leq\ 2d\left(1 + 3(n+k)d\frac{1}{n}\right)
\leq\ 2d\cdot 4d
=\ 8d^2\, .
\end{align*}
Therefore, for large enough $n$,
	\begin{align*}
		\textup{Var}[W^2_{1,(n+k)d}]\leq\ &
		\mathbb{E}[W^4_{1,(n+k)d}] \\
		=\ & \mathbb{E}[(\mu + (W_{1,(n+k)d}-\mu) )^4]\\
		\leq\ & 8\mu^4 + 8 \mathbb{E}[(W_{1,(n+k)d}-\mu)^4]\\
		=\ & 8 \frac{(n+k)^4 d^4}{n^4} + 64d^2 \\
		\leq\ & 10d^4\, .
	\end{align*}
	In the proof of Lemma \ref{lem:balls-into-bins-uniform}, we have shown that $W_{1,(n+k)d}, \ldots, W_{n,(n+k)d}$ are NA.
	By Lemma \ref{lem:NA}--(\ref{lem:NA-monotone}), $W_{1,(n+k)d}^2, \ldots, W_{n,(n+k)d}^2$ are NA, hence we have for large enough $n$,
	\begin{align*}
	\textup{Var}\left[
	\sum_{j\in\cW}W_{j,(n+k)d}^2
	\right]
	\leq\
	n 	\textup{Var}\left[
	W_{1,(n+k)d}^2
	\right]
	\leq
	10 nd^4\, .
	\end{align*}
	Applying Chebyshev's inequality, we have for large enough $n$,
	\begin{align}
		\mathbb{P}\left(\sum_{j\in\cW}W_{j,(n+k)d}^2
		\geq 4 n d^2
		\right)
		\leq\ &
		\mathbb{P}\left(\sum_{j\in\cW}(W_{j,(n+k)d}^2 - \mathbb{E}[W_{j,(n+k)d}^2] )
		\geq 2 n d^2
		\right)
		\nonumber\\
		=\ &
		\mathbb{P}\left(\sum_{j\in\cW}(W_{j,(n+k)d}^2 - \mathbb{E}[W_{j,(n+k)d}^2] )
		\geq \frac{2\sqrt{n}}{\sqrt{10}}\sqrt{10 nd^4}
		\right)\nonumber\\
		\leq\ &
		\mathbb{P}\left(\sum_{j\in\cW}(W_{j,(n+k)d}^2 - \mathbb{E}[W_{j,(n+k)d}^2] )
		\geq \frac{2\sqrt{n}}{\sqrt{10}}
		\sqrt{
			\textup{Var}\left[
		\sum_{j\in\cW}W_{j,(n+k)d}^2
		\right]
	}
		\right)\nonumber\\
		\leq\ & \frac{5}{2n} \leq\ \frac{3}{n}\, .\label{eq:list-length-chebyshev}
	\end{align}

	Given that $\bar{R}[\tau] > \sqrt{d} - 2.3d^\frac{1}{4}$, by plugging $\epsilon \triangleq 0.5d^{-\frac{1}{4}}$ in \eqref{eq:women-rank-concentration-2} of Lemma \ref{lem:women-ranking-preliminary3}, we obtain for large enough $n$,
	\begin{equation}\label{eq:Rtau-ineq-3}
		1 + (1-\epsilon)\bar{R}[\tau]
			\geq 1 + (1-0.5d^{-\frac{1}{4}}) \cdot \sqrt{d}(1 - 2.3 d^{-\frac{1}{4}})
			\geq \sqrt{d} ( 1 - 3 d^{-\frac{1}{4}})
			= \sqrt{d} - 3 d^{\frac{1}{4}}\, .
	\end{equation}
	Therefore,
	\begin{align*}
		\mathbb{P}\left( \Rwomen \leq \sqrt{d} - 3d^{\frac{1}{4}} \right)
			\leq\ &\mathbb{P}\left( \left. \Rwomen \leq \sqrt{d} - 3d^{\frac{1}{4}} \right| \bar{R}[\tau] > \sqrt{d} - 2.3 d^{\frac{1}{4}},\sum_{j\in\cW}W_{j,(n+k)d}^2
			< 4 n d^2\right) \\
			&+ \mathbb{P}\left( \bar{R}[\tau] \leq \sqrt{d} - 2.3d^\frac{1}{4} \right)
			+ \mathbb{P}\left(\sum_{j\in\cW}W_{j,(n+k)d}^2
			\geq 4 n d^2\right)
			\\
			\stackrel{\textup{(a)}}{\leq}\ &\mathbb{P}\left( \left. \Rwomen \leq 1+(1-\epsilon)\bar{R}[\tau] \right| \bar{R}[\tau] > \sqrt{d} - 2.3 d^{\frac{1}{4}},\sum_{j\in\cW}W_{j,(n+k)d}^2
			< 4 n d^2\right) \\
			&+ O(\exp(-\sqrt{n}))
			+ \frac{3}{n}
			\\
			\stackrel{\textup{(b)}}{\leq}\ & \mathbb{E}\left[\left. \exp\left(
			-\frac{\frac{1}{2}d^{-\frac{1}{2}}n^2 \bar{R}[\tau]^2}{4nd^2}
			\right)
			\right|
			\bar{R}[\tau] > \sqrt{d} - 2.3 d^{\frac{1}{4}}
			\right]	+ \frac{4}{n}
			 \\
			 \leq\ &\exp\left( - \frac{1}{8} d^{-\frac{5}{2}} n \cdot d (1 - 2.3 d^{-\frac{1}{4}})^2 \right) + \frac{4}{n}
			\\
			\leq\ &\exp\left( - \frac{n d^{-\frac{3}{2}}}{16} \right)  + \frac{4}{n}\\
			\leq\ & \frac{5}{n}\, .
	\end{align*}
	Here inequality (a) follows from \eqref{eq:Rtau-ineq-3}, \eqref{eq:Rtau-ineq-4}, and \eqref{eq:list-length-chebyshev}; inequality (b) follows from Lemma \ref{lem:women-ranking-preliminary3}.
  \end{proof}

\begin{proof}[Proof of Proposition \ref{prop:women-ranking-upper-bound}]
	We obtain a high probability lower bound on $\Rwomen$ by combining the results of Proposition \ref{prop:total-proposal-lower-bound}, and Lemma \ref{lem:women-ranking-preliminary2} and \ref{lem:women-ranking-preliminary3}.
	By Proposition \ref{prop:total-proposal-lower-bound} and  Lemma \ref{lem:women-ranking-preliminary2},
	\begin{align*}
		\mathbb{P}\left( \bar{R}[\tau] \geq \sqrt{d} + 7.5d^{\frac{1}{4}} \right)
			\leq\ &\mathbb{P}\left( \bar{R}[\tau] \geq \sqrt{d} + 7.5d^{\frac{1}{4}} , \tau > n (\sqrt{d} - 5d^{\frac{1}{4}}) \right)
				 +  \mathbb{P}\left( \tau \leq n (\sqrt{d}-5 d^{\frac{1}{4}} ) \right)
			\\
			\leq\ &\mathbb{P}\left( \bar{R}[n(\sqrt{d}- 5d^{\frac{1}{4}})] \geq \sqrt{d} + 7.5d^{\frac{1}{4}} , \tau > n (\sqrt{d} - 5d^{\frac{1}{4}}) \right)
				 +  \mathbb{P}\left( \tau \leq n (\sqrt{d}-5 d^{\frac{1}{4}} ) \right)
			\\
			\leq\ &\exp\left( - \frac{n}{d^4} \right)  + O(\exp(-d^{\frac{1}{4}}))
			\\
			\leq\ & O( \exp(-d^{\frac{1}{4}}) )\, .
	\end{align*}
	Given that $\bar{R}[\tau] < \sqrt{d} + 7.5d^\frac{1}{4}$, by plugging $\epsilon \triangleq 0.1d^{-\frac{1}{4}}$ in \eqref{eq:women-rank-concentration-1} of Lemma \ref{lem:women-ranking-preliminary3}, we obtain
	\begin{equation*}
		1 + (1+\epsilon)\bar{R}[\tau]
			\leq 1 + (1+0.1d^{-\frac{1}{4}}) \cdot \sqrt{d}(1 + 7.5d^{-\frac{1}{4}})
			\leq \sqrt{d} ( 1 + 8 d^{-\frac{1}{4}})
			= \sqrt{d} + 8 d^{\frac{1}{4}},
	\end{equation*}
	for large enough $n$.
	Recall that we have shown in the proof of Proposition \ref{prop:women-ranking-lower-bound} that
	\begin{align*}
	\mathbb{P}\left(\sum_{j\in\cW}W_{j,(n+k)d}^2
	\geq 4 n d^2
	\right)
\leq\  \frac{3}{n}\, .
	\end{align*}
	Therefore, similar to the proof of Proposition \ref{prop:women-ranking-lower-bound}, we have
	\begin{align*}
		&\mathbb{P}\left( \Rwomen \geq \sqrt{d} + 8d^{\frac{1}{4}} \right)\\
	\leq\ &\mathbb{P}\left( \left. \Rwomen \geq \sqrt{d} + 8d^{\frac{1}{4}} \right| \sqrt{d} - 2.3 d^{\frac{1}{4}} < \bar{R}[\tau] < \sqrt{d} + 7.5 d^{\frac{1}{4}},
	\sum_{j\in\cW}W_{j,(n+k)d}^2
	< 4 n d^2\right)
				\\
		&+ \mathbb{P}\left( \bar{R}[\tau] \leq \sqrt{d} - 2.3d^\frac{1}{4} \right) + \mathbb{P}\left( \bar{R}[\tau] \geq \sqrt{d} + 7.5 d^{\frac{1}{4}} \right)
		+  \mathbb{P}\left( \sum_{j\in\cW}W_{j,(n+k)d}^2
		\geq 4 n d^2 \right)
			\\
		\leq\ &\mathbb{P}\left( \left. \Rwomen \geq 1 + (1+\epsilon) \bar{R}[\tau] \right| \sqrt{d} - 2.3 d^{\frac{1}{4}} < \bar{R}[\tau] < \sqrt{d} + 7.5 d^{\frac{1}{4}},
		\sum_{j\in\cW}W_{j,(n+k)d}^2
		< 4 n d^2
		\right)\\
				&+ O(\exp(-d^{\frac{1}{4}}) )
			\\
		\leq\ &\mathbb{E}\left[ \left. \exp\left( - \frac{2\epsilon^2 n^2 \bar{R}[\tau]^2}{\sum_{j\in\cW}W_{j,(n+k)d}^2} \right) \right| \sqrt{d} - 2.3 d^{\frac{1}{4}} < \bar{R}[\tau] < \sqrt{d} + 7.5 d^{\frac{1}{4}},\sum_{j\in\cW}W_{j,(n+k)d}^2
		< 4 n d^2\right]\\
				&+ O(\exp(-d^{\frac{1}{4}}) )
			\\
		\leq\ &\exp\left( - \frac{1}{200} d^{-\frac{5}{2}} n \cdot d (1 - 2.3 d^{-\frac{1}{4}})^2 \right) + O(\exp(-d^{\frac{1}{4}}) )
			\\
			\leq\ &\exp\left( - \frac{n d^{-\frac{3}{2}}}{300} \right)  + O(\exp(-d^{\frac{1}{4}}) )
			=\ O(\exp(-d^{\frac{1}{4}}) )\, .
	\end{align*}
  \end{proof}

\subsection{Proof of Theorem \ref{thm:main-result-append}}
Theorem \ref{thm:main-result-append} immediately follows from Propositions \ref{prop:total-proposal-upper-bound}, \ref{prop:unmatched-women-lower-bound}, \ref{prop:unmatched-women-upper-bound}, \ref{prop:total-proposal-lower-bound},
\ref{prop:women-ranking-lower-bound}, and \ref{prop:women-ranking-upper-bound}.

\section{Proof for Large Sized $d$: the Case of $d = \omega(\log^2 n)$, $d=o(n)$}\label{append:dense-d}
In this section, we consider the case such that $d=\omega(\log^2 n)$ and $d=o(n)$.
We will prove a quantitative version of Theorem \ref{thm:main-result-dense}.
\begin{thm}[Quantitative version of Theorem \ref{thm:main-result-dense}]\label{thm:main-result-dense-append}
	Consider a sequence of random matching markets indexed by $n$, with $n+k$ men and $n$ women ($k=k(n)$ is negative), and the men's degrees are $d=d(n)$.
	If $|k|=o(n)$, $d=\omega(\log^2 n)$ and $d=o(n)$, we have the following results.
	\begin{enumerate}
		\item \emph{Men's average rank of wives.} With probability $1 - \exp(-\sqrt{\log n})$, we have
		\begin{align*}
		\Mrank
		\leq\
		\left(1 + 2\frac{|k|}{n} + 2\frac{1}{\sqrt{\log n}}\right)\log n
		\, .
		\end{align*}
		\item \emph{Women's average rank of husbands.} With probability $1 - O(\exp(-\sqrt{\log n}))$, we have
		\begin{align*}
		\Wrank
		\geq\
		\left( 1 -
		1.1\left(
		\frac{|k|}{n}
		+\frac{3}{\sqrt{\log n}}
		+\frac{d}{n/\log n}
		\right)
		\right)
		\frac{d}{\log n}
		\, .
		\end{align*}
	\end{enumerate}
\end{thm}

\begin{proof}[Proof of Theorem \ref{thm:main-result-dense-append}]

	\noindent\textbf{Proof of Theorem \ref{thm:main-result-dense-append} part 1.}
	Recall that $\tau$ is the \emph{the total number of proposals} that are made until the end of MPDA, i.e., the time at which the men-optimal stable matching (MOSM) is found.
	We introduce an extended process (which is \emph{different} from the one defined in Appendix \ref{append:notation}) as a natural continuation of the MPDA procedure that continues to evolve even after the MOSM is found (i.e., the extended process continues for $t > \tau$).
	To define the extended process, we start by defining an extended market, which has the same $n$ women and $n+k$ men, but each man has a complete preference list, i.e. each man ranks all $n$ women.
	We call the first $d$ women of a man's preference list his ``real'' preferences and the last $n-d$ women his ``fake'' preferences.
	The distribution of preferences in the extended market is again as described in Section~\ref{sec:model}.
	We then define the \emph{extended process} as tracking the progress of Algorithm~\ref{alg:2MPDA} on the extended market: the $n+k$ men enter first in Algorithm~\ref{alg:2MPDA} with only their real preferences, as before.
	After time $\tau$, we let the men see their fake preferences and continue Algorithm~\ref{alg:2MPDA} until the MOSM with full preferences is found.
	We denote by $\tau'$ the total number of proposals to find the MOSM with full preferences.
	It is easy to see that $\tau$ is stochastically dominated by $\tau'$.
	
	Note that $\tau'$ is the total number of proposals needed to find the MOSM in a completely-connected market, which has been studied in previous works including \cite{ashlagi2015unbalanced,pittel2019likely}.
	It is well-known that $\tau'$ is stochastically dominated by the number of draws in a coupon collector's problem, in which one coupon is chosen out of $n$ coupons uniformly at random at a time and it runs until $n$ distinct coupons are collected.
	Let $X$ be the number of draws in the coupon collector's problem.
	A widely used tail bound of $X$ is the following: for $\beta>1$, $\mathbb{P}(X\geq \beta n\log n)\leq n^{-\beta+1}$.
	By taking $\beta = 1 + \frac{1}{\sqrt{\log n}}$, we have
	\begin{align*}
		\mathbb{P}\left(
		X \geq n\log n + n\sqrt{\log n}
		\right)
		\leq
		n^{-\frac{1}{\sqrt{\log n}}}
		=
		e^{-\sqrt{\log n}}
		=
		o(1)\, .
	\end{align*}
	Hence with probability $1 - e^{-\sqrt{\log n}}$, we have $\tau \leq n(\log n + \sqrt{\log n})$.
	Because $X$ stochastically dominates $\tau$, we have, with probability $1 - e^{-\sqrt{\log n}}$,
	\begin{align*}
		\Mrank \leq \frac{n}{n+k}(\log n + \sqrt{\log n}) + 1\, .
	\end{align*}
	Because $k=o(n)$ and $k<0$, for large enough $n$ we have $\frac{n}{n+k} \leq 1+\frac{2|k|}{n}$, $\frac{|k|}{n}<\frac{1}{3}$, $1\leq \frac{1}{3}\sqrt{\log n}$, hence
	\begin{align*}
		&\frac{n}{n+k}(\log n + \sqrt{\log n}) + 1\\
		\leq\ & \left(1 + 2\frac{|k|}{n}\right)\log n
		+
		(1 + \frac{2}{3})\sqrt{\log n}
		+
		\frac{1}{3}\sqrt{\log n} \\
		=\ &
		 \left(1+2\frac{|k|}{n}+ 2\frac{1}{\sqrt{\log n}}\right)\log n\, .
	\end{align*}
	This concludes the proof.

\medskip

\noindent \textbf{Proof of Theorem \ref{thm:main-result-dense-append} part 2.}

	The proof is similar to that of Proposition \ref{prop:women-ranking-lower-bound}. Recall that the proof of Proposition \ref{prop:women-ranking-lower-bound} relies on Proposition \ref{prop:total-proposal-upper-bound}, Lemma \ref{lem:women-ranking-preliminary2}, and Lemma \ref{lem:women-ranking-preliminary3}.
	In the following, we first establish the counterparts of these results in dense markets.
	
	\emph{Counterpart of Proposition \ref{prop:total-proposal-upper-bound} in dense markets.} We have shown in the proof of Theorem \ref{thm:main-result-dense-append}(1) that with probability $1 - \exp(-\sqrt{\log n})$, we have
	\begin{align}\label{eq:tau-upper-bound-dense}
		\tau \leq n\left(\log n + \sqrt{\log n}\right)\, .
	\end{align}
	
	\emph{Counterpart of Lemma \ref{lem:women-ranking-preliminary2} in dense markets.}		
	Fix $t = n\left( \log n + \sqrt{\log n} \right)$.
	Given the asymptotic condition, we have for large enough $n$,
		\begin{align*}
		\frac{t}{nd} =&\ \frac{\log n+\sqrt{\log n}}{d} \leq  0.1(\log n)^{-1}\, .
		\end{align*}
		By examining the proof of Lemma \ref{lem:women-ranking-preliminary}, we can see that we have proved the following result (see the statement of Lemma \ref{lem:women-ranking-preliminary} for the definition of the notations), which is stronger than than \eqref{eq:Ylb}:
		\begin{align}\label{eq:Ylb-stronger}
			\mathbb{P}\left( \left. Y_{t,T} \leq (1-\epsilon) \frac{T-t}{n-d}\sum_{j=1}^{n-d}c_j \right| \vec{W}_t \right) &\leq \exp\left( - \frac{1}{4} \epsilon^2 \times \frac{T-t}{n-d}\sum_{j=1}^{n-d}c_j \right)
		\end{align}
		Let $c_j \defeq \frac{1}{W_{j,t}+1}$ where $W_{1,t}\geq W_{2,t} \geq\cdots\geq W_{n,T}$, and $T \triangleq (n+k)d$.
		Due to the convexity of $f(x) \triangleq \frac{1}{x+1}$, we have for large enough $n$,
		\begin{align}
		\frac{1}{n-d}\sum_{j=1}^{n-d}c_j
		=\ &\frac{1}{n-d} \sum_{j =1}^{n-d} f(W_{j,t})
		\geq\ f\left( \frac{1}{n-d} \sum_{j=1}^{n-d} W_{j,t} \right)
		\geq\ f\left( \frac{t}{n-d} \right) \nonumber\\
		=\ &\frac{1}{t/(n-d)+1}
		\geq \frac{1}{\log n\left(1 + 1.05 \frac{1}{\sqrt{\log n}} + 1.05\frac{d}{n}\right)}\, . \label{eq:S-jensen-ineq-dense}
		\end{align}
	Therefore, for large enough $n$,
		\begin{align}
		\frac{T - t}{n-d}\sum_{j=1}^{n-d}c_j
		\geq\ &n\frac{(n+k)d - t}{n}\frac{1}{\log n\left(1 + 1.05 \frac{1}{\sqrt{\log n}} + 1.05\frac{d}{n}\right)} \nonumber\\
		\geq\ & nd \left( 1 - \frac{|k|}{n} - \frac{t}{nd} \right)  \frac{1}{\log n\left(1 + 1.05 \frac{1}{\sqrt{\log n}} + 1.05\frac{d}{n}\right)}
		\nonumber
		\\
		\geq\ & nd \left( 1 - \frac{|k|}{n} - \frac{0.1}{\log n} \right)
		\frac{1}{\log n}
		\left(1 - 1.05\frac{1}{\sqrt{\log n}}  -1.05\frac{d}{n} \right)
		\nonumber
		\\
		\geq\ &
		\frac{n d}{\log n}\left( 1 - 1.1\frac{|k|}{n}
		-
		1.1\frac{1}{\sqrt{\log n}}
		-
		1.1\frac{d}{n}
		\right)\, .
		\label{eq:women-ranking-ratio-lower-bound-dense}
		\end{align}
		Utilizing Lemma \ref{lem:women-ranking-preliminary} (which does not use assumptions on $d$) with $\epsilon \triangleq \frac{1}{\sqrt{\log n}}$, we obtain
		\begin{align*}
		&\mathbb{P}\left( \left. \bar{R}[t] \leq
		\frac{d}{\log n}
		\left( 1 -
		\left(
		1.1\frac{|k|}{n}
		+\frac{2.1}{\sqrt{\log n}}
		+1.1\frac{d}{n}
		\right)
		\right)
		 \right| \vec{W}_t \right)\\
		=\ &\mathbb{P}\left( \left. n \bar{R}[t] \leq \frac{n d}{\log n}
		\left( 1 -
		\left(
		1.1\frac{|k|}{n}
		+\frac{2.1}{\sqrt{\log n}}
		+1.1\frac{d}{n}
		\right)
		\right)
		 \right| \vec{W}_t \right)
		\\
		\leq\ &\mathbb{P}\left( \left. n \bar{R}[t] \leq (1-\epsilon) \times
		\frac{n d}{\log n}\left( 1 - 1.1
		\left(
		\frac{|k|}{n}
		+\frac{1}{\sqrt{\log n}}
		+\frac{d}{n}
		\right)
		\right)
		 \right| \vec{W}_t \right)
		\\
		\stackrel{\textup{(a)}}{\leq}\ &\mathbb{P}\left( \left. n \bar{R}[t] \leq (1-\epsilon) \times \frac{(n+k)d-t}{n-d}\sum_{j=1}^{n-d}c_j  \right| \vec{W}_t \right)
		\\
		\stackrel{\textup{(b)}}{\leq}\ &\exp\left( - \frac{1}{4} \epsilon^2 \times \frac{(n+k)d-t}{n-d}\sum_{j=1}^{n-d}c_j  \right)
		\\
		\stackrel{\textup{(c)}}{\leq}\ &\exp\left( - \frac{1}{4} \frac{1}{\log n} \times
		\frac{n d}{\log n}\left( 1 - 1.1
		\left(
		\frac{|k|}{n}
		+\frac{1}{\sqrt{\log n}}
		+\frac{d}{n}
		\right)
		  \right)
		  \right)
		\\
		\leq\ &\exp( - n / 8 )\, .
		\end{align*}
		Here inequalities (a) and (c) follow from \eqref{eq:women-ranking-ratio-lower-bound-dense}, inequality (b) follows from Lemma \ref{lem:women-ranking-preliminary}, and the last inequality follows from the fact that $d=\omega(\log^2 n)$.
		Since the above result holds for any realization of $\vec{W}_t$, we have
\begin{align}
\label{eq:women-ranking-mean-upper-bound-dense}
\mathbb{P}\left( \bar{R}\left[ n (\log n + \sqrt{\log n} ) \right] \leq
\frac{d}{\log n}
\left( 1 -
1.1\left(
\frac{|k|}{n}
+\frac{2}{\sqrt{\log n}}
+\frac{d}{n}
\right)
\right)
 \right) &\leq \exp\left( - \frac{n}{8} \right)\, .
\end{align}

\emph{Counterpart of Lemma \ref{lem:women-ranking-preliminary3} in dense markets.} Note that the proof of Lemma \ref{lem:women-ranking-preliminary3} does not make any assumption on $d$, hence \eqref{eq:women-rank-concentration-2} still holds.

\textit{Proof of Theorem \ref{thm:main-result-dense-append} part 2}.
Using \eqref{eq:tau-upper-bound-dense} and \eqref{eq:women-ranking-mean-upper-bound-dense}, and the fact that $\bar{R}[t]$ is decreasing on each sample path,
\begin{align}
&\mathbb{P}\left( \bar{R}[\tau] \leq
\frac{d}{\log n}
\left( 1 -
1.1\left(
\frac{|k|}{n}
+\frac{2}{\sqrt{\log n}}
+\frac{d}{n}
\right)
\right)
 \right) \nonumber\\
\leq\ &\mathbb{P}\left( \bar{R}[\tau] \leq
\frac{d}{\log n}
\left( 1 -
1.1\left(
\frac{|k|}{n}
+\frac{2}{\sqrt{\log n}}
+\frac{d}{n}
\right)
\right),
\tau <  n (\log n + \sqrt{\log n} ) \right) \nonumber\\
&
+ \mathbb{P}\left( \tau \geq  n (\log n + \sqrt{\log n} ) \right)
\nonumber\\
\leq\ &
\mathbb{P}\left( \bar{R}\left[ n (\log n + \sqrt{\log n} ) \right] \leq
\frac{d}{\log n}
\left( 1 -
1.1\left(
\frac{|k|}{n}
+\frac{2}{\sqrt{\log n}}
+\frac{d}{n}
\right)
\right)
\right)\nonumber\\
& +
 \mathbb{P}\left( \tau \geq  n (\log n + \sqrt{\log n} ) \right)
\nonumber\\
\stackrel{\textup{(a)}}{\leq}\ &\exp\left( -\frac{n}{8} \right) + O( \exp(-\sqrt{\log n}) )
= O( \exp(-\sqrt{\log n}) )\, .\label{eq:Rtau-ineq-2}
\end{align}
Here inequality (a) follows from \eqref{eq:women-ranking-mean-upper-bound-dense}.
Recall inequality \eqref{eq:list-length-chebyshev}: for large enough $n$
\begin{align*}
\mathbb{P}\left(\sum_{j\in\cW}W_{j,(n+k)d}^2
\geq 4 n d^2
\right)
\leq\  \frac{3}{n}\, .
\end{align*}
In the derivation of the above inequality, we only used the fact that $d=\omega(1),d=o(n)$ and $k=o(n)$, which also holds in dense markets.

Given that $\bar{R}[\tau] > \frac{d}{\log n}
\left( 1 -
1.1\left(
\frac{|k|}{n}
+\frac{2}{\sqrt{\log n}}
+\frac{d}{n}
\right)
\right)$,
by plugging $\epsilon \triangleq 0.5 \frac{1}{\sqrt{\log n}}$ in \eqref{eq:women-rank-concentration-2} of Lemma \ref{lem:women-ranking-preliminary3}, we obtain for large enough $n$,
\begin{align}
1 + (1-\epsilon)\bar{R}[\tau]
\geq\ &1 + \left(1-0.5\frac{1}{\sqrt{\log n}}\right)
\frac{d}{\log n}
\left( 1 -
1.1\left(
\frac{|k|}{n}
+\frac{2}{\sqrt{\log n}}
+\frac{d}{n}
\right)
\right)\nonumber\\
\geq\ &\frac{d}{\log n}
\left( 1 -
1.1\left(
\frac{|k|}{n}
+\frac{3}{\sqrt{\log n}}
+\frac{d}{n}
\right)
\right)\, .\label{eq:Rtau-ineq-1}
\end{align}
Therefore,
\begin{align*}
&\mathbb{P}\left( \Rwomen \leq \frac{d}{\log n}
    \left( 1 -
    1.1\left(
    \frac{|k|}{n}
    +\frac{3}{\sqrt{\log n}}
    +\frac{d}{n}
    \right)
    \right) \right)\\
\leq\ &\mathbb{P}\left( \left. \Rwomen \leq
    \frac{d}{\log n}
    \left( 1 -
    1.1\left(
    \frac{|k|}{n}
    +\frac{3}{\sqrt{\log n}}
    +\frac{d}{n}
    \right)
    \right)
     \right| \right.\\
     &\left.
     \bar{R}[\tau] >
     \frac{d}{\log n}
     \left( 1 -
     1.1\left(
     \frac{|k|}{n}
     +\frac{2}{\sqrt{\log n}}
     +\frac{d}{n}
     \right)
     \right),
     \sum_{j\in\cW}W_{j,(n+k)d}^2
    < 4 n d^2\right) \\
    &+ \mathbb{P}\left( \bar{R}[\tau] \leq \frac{d}{\log n}
    \left( 1 -
    1.1\left(
    \frac{|k|}{n}
    +\frac{2}{\sqrt{\log n}}
    +\frac{d}{n}
    \right)
    \right) \right)
    + \mathbb{P}\left(\sum_{j\in\cW}W_{j,(n+k)d}^2
    \geq 4 n d^2\right)
\\
\stackrel{\textup{(a)}}{\leq}\ &\mathbb{P}\left(  \Rwomen \leq 1+(1-\epsilon)\bar{R}[\tau] \bigg | 
 \bar{R}[\tau] >
 \frac{d}{\log n}
 \left( 1 -
 1.1\left(
 \frac{|k|}{n}
 +\frac{2}{\sqrt{\log n}}
 +\frac{d}{n}
 \right)
 \right),\right.\\
 &\left.
 \sum_{j\in\cW}W_{j,(n+k)d}^2
< 4 n d^2\right) + O(\exp(-\sqrt{\log n}))
+ \frac{3}{n}
\\
\stackrel{\textup{(b)}}{\leq}\ & \mathbb{E}\left[\left. \exp\left(
-\frac{\frac{1}{2\log n}n^2 \bar{R}[\tau]^2}{4nd^2}
\right)
\right|
\bar{R}[\tau] >
\frac{d}{\log n}
\left( 1 -
1.1\left(
\frac{|k|}{n}
+\frac{2}{\sqrt{\log n}}
+\frac{d}{n}
\right)
\right)
\right]\\
&	+ O(\exp(-\sqrt{\log n}))
\\
\leq\ &\exp\left( - \frac{1}{8} \frac{n}{d^2\log n} \cdot
\frac{d^2}{2\log^2 n} \right) + O(\exp(-\sqrt{\log n}))
\\
\leq\ &\exp\left( - \frac{n (\log n)^{-3}}{16} \right)  + O(\exp(-\sqrt{\log n}))\\
\leq\ & O(\exp(-\sqrt{\log n}))\, .
\end{align*}
Here inequality (a) follows from \eqref{eq:Rtau-ineq-1}, \eqref{eq:Rtau-ineq-2}, and \eqref{eq:list-length-chebyshev}; inequality (b) follows from Lemma \ref{lem:women-ranking-preliminary3}.
This concludes the proof.
\end{proof}

\section{Additional Proofs}\label{append:additional-proof}
\subsection{Proof of Lemma \ref{lem:random-order-statistics}}

\begin{proof}
We first prove the first part. Let $M_k$ be the $k$-th largest order statistics out of $d(n)$ samples from Pareto distribution with scale parameter $1$ and shape parameter $\alpha>1$. Here $1\leq k \leq d(n)$ is a constant. Using equation (3.7) in \cite{malik1966exact} and our asymptotic notation, we have as $d(n),k\to\infty$:
\begin{align*}
\mathbb{E}[M_k]=\ & \frac{\Gamma(d+1)\Gamma(k-1/\alpha)}{\Gamma(k)\Gamma(d+1-1/\alpha)}\\
\doteq\ &
\frac{
\sqrt{d}\left(\frac{d}{e}\right)^d
\sqrt{k-1/\alpha-1}
\left(\frac{k-1/\alpha-1}{e}\right)^{k-1/\alpha-1}
}
{
\sqrt{k-1}\left(\frac{k-1}{e}\right)^{k-1}
\sqrt{d-1/\alpha}\left(\frac{d-1/\alpha}{e}\right)^{d-1/\alpha}
}\\
\doteq\ &
\frac{
(d)^d
(k-1/\alpha-1)^{k-1/\alpha-1}
}
{
(d-1/\alpha)^{d-\alpha^{-1}}
(k-1)^{k-1}
}\\
\doteq\ &
\left(\frac{d-1/\alpha}{k-1/\alpha-1}\right)^{1/\alpha}
\left(1 + \frac{1/\alpha}{d-1/\alpha}\right)^{d}
\left(1-\frac{1/\alpha}{k-1}\right)^{k-1}\\
\doteq\ & \left(\frac{d}{k}\right)^{1/\alpha}\, .
\end{align*}
Here the second row follows from Stirling's formula.

With the above, we have:
\begin{align*}
\mathbb{E}[M(d(n),\alpha,r(n))]
=\ \mathbb{E}_{\zeta}[\mathbb{E}[M_{\zeta}]|{\zeta}=k]
\doteq\ \mathbb{E}_{\zeta}\left[\left(\frac{d}{{\zeta}}\right)^{1/\alpha}\right]
=\ & d^{1/\alpha} \mathbb{E}[{\zeta}^{-1/\alpha}]\, ,
\end{align*}
where $\mathbb{E}_{\zeta}[\cdot]$ is expectation taken w.r.t. random variable
${\zeta}$; and ${\zeta}$ is the distribution of men's ranking of wives.
Note that the desired quantity is the inverse moment of a positive discrete distribution. Using Theorem 3 in \cite{znidaric2005asymptotic}, we have $  \mathbb{E}[{\zeta}^{-1/\alpha}]\doteq\ r^{-1/\alpha}$. Therefore we have the desired result, i.e.,
\begin{align*}
\mathbb{E}[M(d(n),\alpha,r(n))] \doteq \left(\frac{d(n)}{r(n)}\right)^{1/\alpha}\, .
\end{align*}
    
Now we proceed to prove the second part. Let $L_k$ be the largest order statistics out of $k$ samples from Pareto distribution with scale parameter 1 and shape parameter $\alpha>1$. Here $k$ is a constant. Similar as in the proof of the first part, we have as $k\to\infty$,
\begin{align*}
\mathbb{E}[L_k]=\ & \frac{\Gamma(k+1)\Gamma(1-1/\alpha)}{\Gamma(1)\Gamma(k+1-1/\alpha)}\\
\doteq\ &
\frac{
\sqrt{k}\left(\frac{k}{e}\right)^k
\Gamma(1-1/\alpha)
}
{
\sqrt{k-1/\alpha}\left(\frac{k-1/\alpha}{e}\right)^{k-1/\alpha}
}\\
\doteq\ &
\frac{
k^k
}
{
(k-1/\alpha)^{k-1/\alpha}
}\\
\doteq\ &
\left(1-\frac{1/\alpha}{k}\right)^{1/\alpha}
\left(1 + \frac{1/\alpha}{k-1/\alpha}\right)^{k}
k^{1/\alpha}\\
\doteq\ & k^{1/\alpha}\, .
\end{align*}

Therefore, we have
\begin{align*}
    \mathbb{E}[W(d(n),\alpha,r(n))]\doteq \mathbb{E}[\eta^{1/\alpha}]\, ,
\end{align*}
where $\eta$ is the number of proposals a woman receives. 
Since $\eta$ satisfies the tail condition in \cite{ahle2022sharp}, it follows from \cite{ahle2022sharp} that
\begin{align*}
    \mathbb{E}[W(d(n),\alpha,r(n))]\doteq
    \left(\mathbb{E}[\eta]\right)^{1/\alpha}
    =
    (r(n))^{1/\alpha}\, .
\end{align*}
This concludes the proof.   
\end{proof}

\section{Principle~\ref{prin:condition-for-weak-competition} for random markets and numerical procedure for refined estimates}
\label{app:heuristic-numerical}

In this section, we flesh out the detailed heuristic picture of equilibrium introduced in Section~\ref{subsec:detailed-heuristic-picture}, and deduce (i) a procedure to compute refined estimates of average ranks and the number of unmatched agents, whose predictions are shown in Section~\ref{sec:numeric}, (ii) Principle~\ref{prin:condition-for-weak-competition} for random markets.

Given the number of men $|\mathcal{M}|$, the number of women $|\mathcal{W}|$, and the men's average degree $d_m$ (and the women's average degree $d_w = d_m \times |\mathcal{M}|/|\mathcal{W}|$), we can estimate men's average rank $R_\text{MEN}$, women's average rank $R_\text{WOMEN}$, the number of unmatched men $\delta^m$, and the number of unmatched women $\delta^w$, by solving the following four equations:
\begin{align}
	R_\text{MEN} &\approx \frac{ |\mathcal{W}| }{ |\mathcal{M}| } \log\left( \frac{ |\mathcal{W}| }{ \delta^w } \right),
	\label{eq:heuristic-mrank-id}
	\\
	R_\text{WOMEN} &\approx \frac{ |\mathcal{M}| }{ |\mathcal{W}| } \log\left( \frac{ |\mathcal{M}| }{ \delta^m } \right),
	\label{eq:heuristic-wrank-id}
	\\
	|\mathcal{M}| - \delta^m &= |\mathcal{W}| - \delta^w,
	\label{eq:heuristic-unmatched-id}
	\\
	R_\text{MEN} \times R_\text{WOMEN} &\approx \min\left( d_m, d_w \right).
	\label{eq:heuristic-product-id}
\end{align}
Equations \eqref{eq:heuristic-mrank-id} and \eqref{eq:heuristic-wrank-id}  follow from an informal analogy with the coupon collector problem,\footnote{In order to collect $|\mathcal{W}| - \delta^w$ distinct coupons out of a universe of $|\mathcal{W}|$ coupons, $|\mathcal{W}| \times \log( |\mathcal{W}|/\delta^w)$ uniformly random coupons must be drawn, which corresponds to the total number of proposals made by men, $R_\text{MEN} \times |\mathcal{M}|$.} equation \eqref{eq:heuristic-unmatched-id} is the identity requiring that the number of matched men should equal the number matched women, and equation \eqref{eq:heuristic-product-id} captures the fact that a woman's expected rank is inversely proportional to the number of proposals that she receives.

The system of equations \eqref{eq:heuristic-mrank-id}-\eqref{eq:heuristic-product-id} admits a unique solution which can be easily computed (to a high degree of accuracy), e.g., via bisection method. 
Once the four market statistics $R_\text{MEN}$, $R_\text{WOMEN}$, $\delta^m$, $\delta^w$, are estimated, we can further approximate each side's rank distribution using the distributions conjectured in Section~\ref{subsec:detailed-heuristic-picture}.
The men's rank distribution can be well approximated by a truncated Geometric distribution, $\min\{ \text{Geo}( R_\text{WOMEN}/d_w), d_m \},$
and similarly, the women's rank distribution can be approximated by a truncated Geometric distribution, $\min\{ \text{Geo}( R_\text{MEN}/d_m), d_w \}$. 

\paragraph{Principle~\ref{prin:condition-for-weak-competition} for random markets.}
The above detailed heuristic picture leads to the following quantitative version of Principle~\ref{prin:condition-for-weak-competition} for random markets: 
\begin{conj}[Principle~\ref{prin:condition-for-weak-competition} for random markets]
Fix any $\epsilon> 0$ and any $\nu \in (\epsilon, \infty)$. Consider a sequence of markets indexed by $n$ with $k<0$ and $|k| \leq n^{1-\epsilon}$. 
Then there exists a threshold $d^*(n) = \log^2 (n/k) (1+o(1))$ and $g(n) = \Theta(\log n)$ such that the following holds:
\begin{itemize}
    \item If $d \leq d^*$, then, with high probability, we have an insignificant impact of imbalance $\Rmen \geq \Rwomen (1- (1+\epsilon)/g(n))$ and many unmatched agents on the short side $\delta^m > (\nu - \epsilon) |k|$.
    \item If $d > d^*$, then, with high probability, we have an significant impact of imbalance $\Rmen \leq \Rwomen (1- (1-\epsilon)/g(n))$ and few unmatched agents on the short side $\delta^m < (\nu + \epsilon) |k|$.
\end{itemize}
\label{conj:prin1-for-random-markets}
\end{conj}

We now informally deduce this conjecture from the detailed heuristic picture.
Since $k \leq n^{1-\epsilon}$, we make the approximation $d_w \approx d_m = d$, since the degree of women concentrates around the average, which is $d (1+ \Theta(|k|/n))$, and the $\Theta(|k|/n)$ error term will not play a role in our calculation. Similarly, we use the approximation $|\mathcal{M}| \approx |\mathcal{W}| = n$ in \eqref{eq:heuristic-mrank-id} and \eqref{eq:heuristic-wrank-id}, and the justification is similar. With these approximations, we obtain
\begin{align*}
    \delta^w &= n \exp(-\Rmen)\\
    \delta^m &= n \exp(-d/\Rmen)\, ,
\end{align*}
where the second identity is obtained by plugging $\Rwomen \times \Rmen \approx d$ into \eqref{eq:heuristic-wrank-id}. Plugging in $\delta^w = \delta^m + |k| = n \exp(-d/\Rmen) + |k|$ into the first of our equations, we obtain
\begin{align}
    |k| = n \exp(-\Rmen) - n \exp(-d/\Rmen) \, .
    \label{eq:solve-for-Rmen}
\end{align}
Note that the right-hand side is monotone decreasing in $\Rmen$, and hence \eqref{eq:solve-for-Rmen} uniquely identifies $\Rmen$ for any given $d$. Moreover, we see that the right-hand side is increasing in $d$, and hence the solution $\Rmen(d,k)$ is a decreasing function of $d$ (we can similarly see that it is a decreasing function of $k$). Recall $k \leq n^{1-\epsilon}$. Fixing any $\zeta_1 \in (-\infty, \infty)$, for $d = \log^2(n/k) (1 + \zeta_1/\log(n/k))$ one can verify that the unique solution satisfies $\Rmen = \sqrt{d} (1- \zeta_2 (1+o(1))/\log (n/k))$ for some $\zeta_2 \in (0, \infty)$ which is increasing in $\zeta_1$, leading to $\Rwomen = \sqrt{d} (1+ \zeta_2 (1+o(1))/\log (n/k))$ (since $\Rwomen \times \Rmen \approx d$) and hence $\delta^m = n \exp(-\Rwomen) = k \exp(-\zeta_2-\zeta_1/2)(1+o(1))$. Let $\zeta_1^*$ be the unique value of $\zeta_1$ such that $\exp(-\zeta_2 - \zeta_1/2) = \nu$, and set $d^* \triangleq \log^2(n/k) (1 + \zeta_1^*/\log(n/k))$ (so that $\delta^m = \nu |k| (1+o(1))$ for $d= d^*$). Define $g(n) \triangleq \log (n/k)/(2 \zeta_2)$, and note that $\log(n/k) = \Theta(\log n)$ since $k \in [1, n^{1-\epsilon}]$. Assuming concentration of average ranks and $\delta^m$ about their estimated values, and using the established monotonicities of our estimates of $\delta^m$ and $\Rmen$ in $d$, we then obtain  Conjecture~\ref{conj:prin1-for-random-markets}.

\section{Additional Numerical Simulations}
\label{app:numerical-additional}

{In this section, we provide additional numerical simulation results that are omitted from Section~\ref{sec:numeric} due to space constraints.
}

\paragraph{Heuristic refined approximation for larger imbalances.}
{We here verify the quality of our refined estimates for the average ranks and the number of unmatched agents in unbalanced markets obtained through the procedure described in Appendix~\ref{app:heuristic-numerical}.
We adopt the same experimental setup considered for verification of Theorem~\ref{thm:main-result} and \ref{thm:main-result-dense}, except the imbalance is no longer negligible compared the market size: the market now has 1,000 men, 1,050 women ($n=1050, k=-50$), and a varying men's degree $d$.
The results are reported in Figure~\ref{fig:fixed-size-heuristic}, visualized analogously to Figure~\ref{fig:fixed-size}.
}

\begin{figure}[htb!]
    \centering
    \begin{minipage}{0.48\textwidth}
        \centering
        \includegraphics[width=1.0\textwidth]{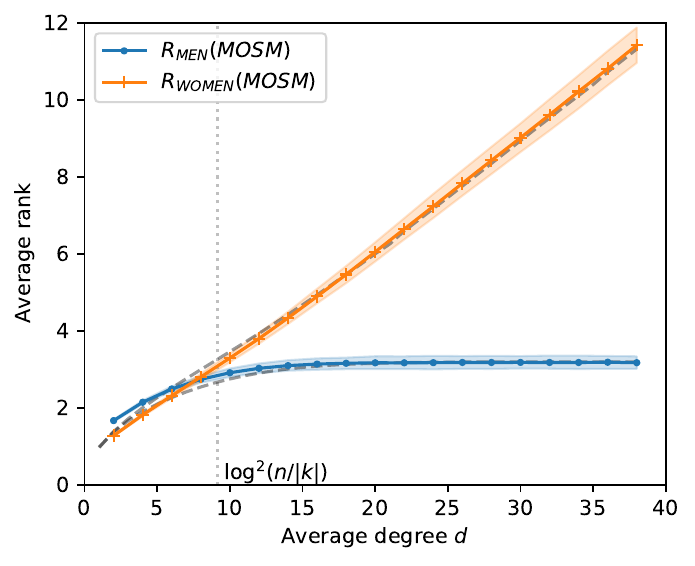}
    \end{minipage}\hfill
    \begin{minipage}{0.48\textwidth}
        \centering
        \includegraphics[width=1.0\textwidth]{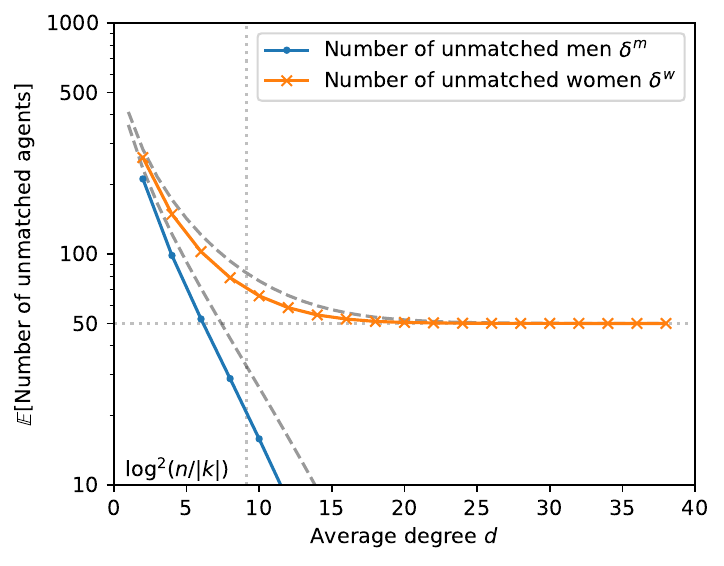}
    \end{minipage}
	\caption{
	    {Verification of heuristic refined estimates for unbalanced random matching markets.
	    }
        	Men's average rank of wives $\Rmen$ and women's average rank of husbands $\Rwomen$ (left) and the number of unmatched men $\delta^m$ and the number of unmatched women $\delta^w$ (right) under the MOSM in random matching markets with 1,000 men and 1,050 women ($n=1050$, $k=-50$), and varying length of men's preference list $d$.
        	The gray dashed lines represent the refined estimates that are obtained by solving the system of equations \eqref{eq:heuristic-mrank-id} -- \eqref{eq:heuristic-product-id} in Appendix~\ref{app:heuristic-numerical}. 
        }
	\label{fig:fixed-size-heuristic}
\end{figure}
In addition to the accuracy of the refined estimates, we also confirm that the effect of connectivity $d$ on the market equilibrium is consistent with our theoretical understanding: (i) $R_\text{MEN} \approx R_\text{WOMEN}$ when $d \leq \log^2 (n/|k|)$, (ii) $R_\text{MEN}$ and $R_\text{WOMEN}$ deviate from each other when $d \geq \log^2 (n/|k|)$, and (iii) the threshold degree $\log^2 (n/|k|)$ is close to the connectivity at which the number of unmatched agents is $0.5$ times the market imbalance.

\paragraph{Robustness to  heterogeneity in men's degree.}
{We now consider an Erdos-Renyi consideration graph as an alternative to the assumption in our main model that men all have degree $d$.
Instead of having the same and deterministic degree $d$ across all men, we randomly connect each pair $(i,j) \in \mathcal{M} \times \mathcal{W}$ with probability $d/n$, independently across all pairs.
The degree of a man is now distributed as $\text{Binomial}(n,d/n)$ (hence, there is heterogeneity in men's degree, but the average degree is again $d$), and the degree of a woman is distributed as $\text{Binomial}(n+k,d/n)$.

We repeat the same experiments above.
As reported in Figure~\ref{fig:erdos-fixed-size}, we find our quantitative predictions to be quite accurate, suggesting that our theoretical findings are also valid for an Erdos-Renyi consideration graph, and are robust to some heterogeneity in men's degree.
}

\begin{figure}[htb!]
    \centering
    \begin{minipage}{0.48\textwidth}
        \centering
        \includegraphics[width=1.0\textwidth]{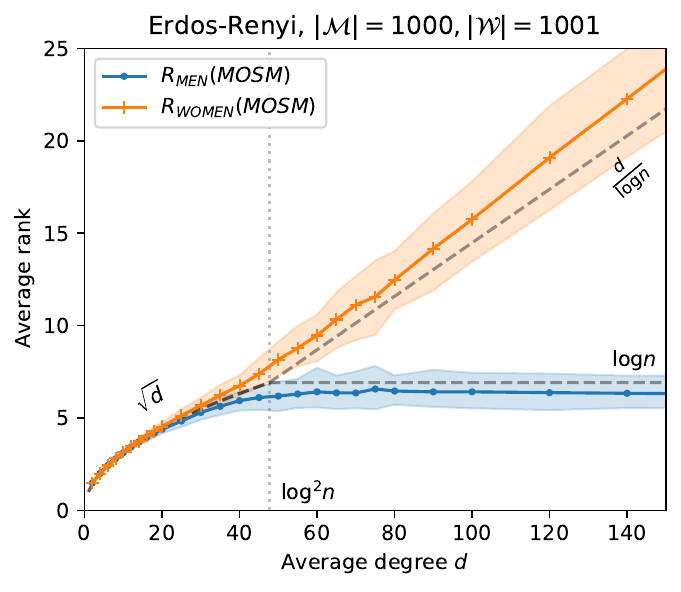}
        \includegraphics[width=1.0\textwidth]{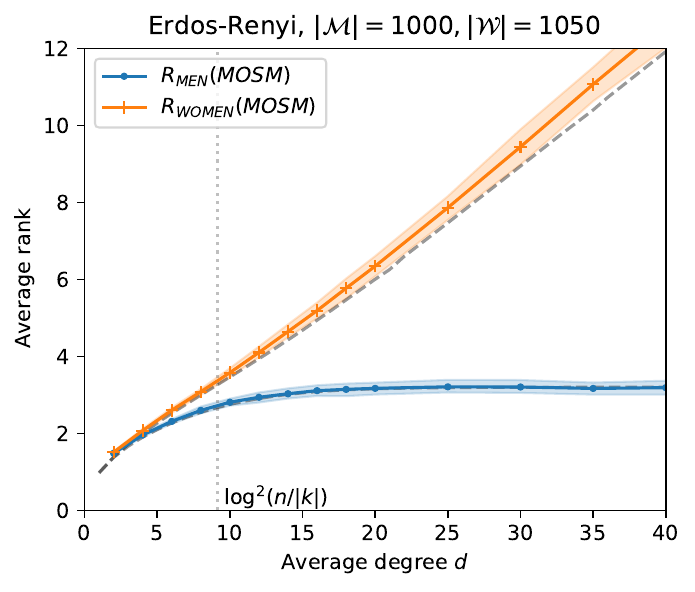}
    \end{minipage}\hfill
    \begin{minipage}{0.48\textwidth}
        \centering
        \includegraphics[width=1.0\textwidth]{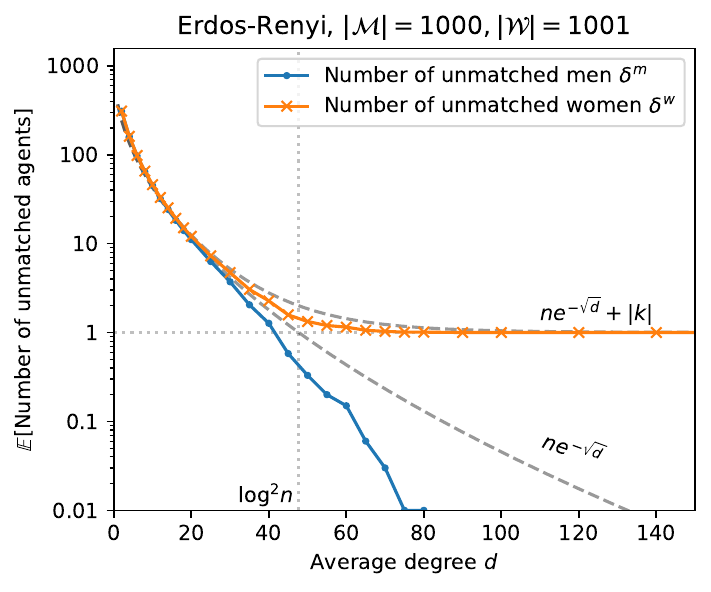}
        \includegraphics[width=1.0\textwidth]{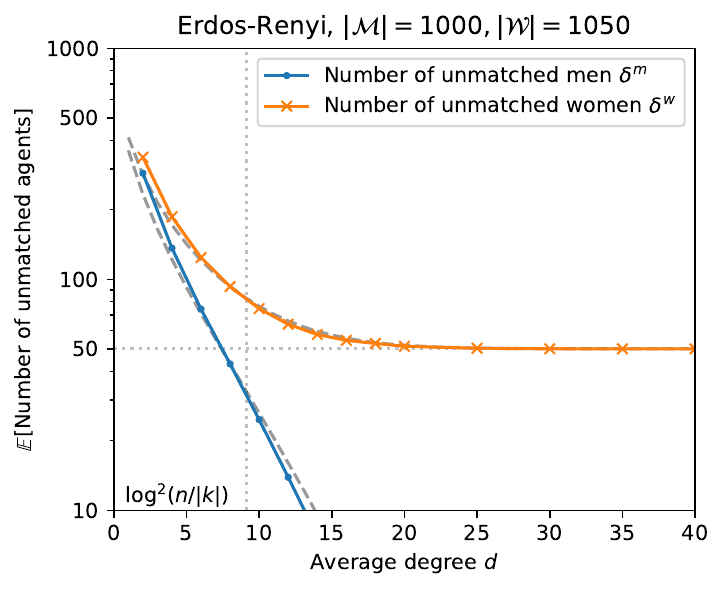}
    \end{minipage}
	\caption{
        	{Robustness check for random matching markets generated with Erdos-Renyi connectivity model ($|\mathcal{M}|=1000, |\mathcal{W}| \in \{1001, 1050\}$).
        	The solid lines represent simulation results obtained from 100 runs of simulation, and the dashed gray lines represents our theoretical predictions.
        	}
        }
	\label{fig:erdos-fixed-size}
\end{figure}

\paragraph{Behavior under correlated preferences.}
{In Section~\ref{sec:numeric}, we provided simulation results for random matching markets with correlated preferences, with the choice of $\beta=5$ in the adopted generative model (see Figure~\ref{fig:robust-check-correlated-unbalanced}).
We here report the results for different choices of market primitives and $\beta \in \{1, 5, 10 \}$ (i.e., moderate to strong correlations in preferences, since the quality term in the utility is $\beta$ times larger than the idiosyncratic term), summarized in Figure~\ref{fig:robust-check-correlated-1} and Figure~\ref{fig:robust-check-correlated-2}.
}

\begin{figure}[htb!]
    \centering
    \begin{minipage}{0.48\textwidth}
        \centering
        \includegraphics[width=1.0\textwidth]{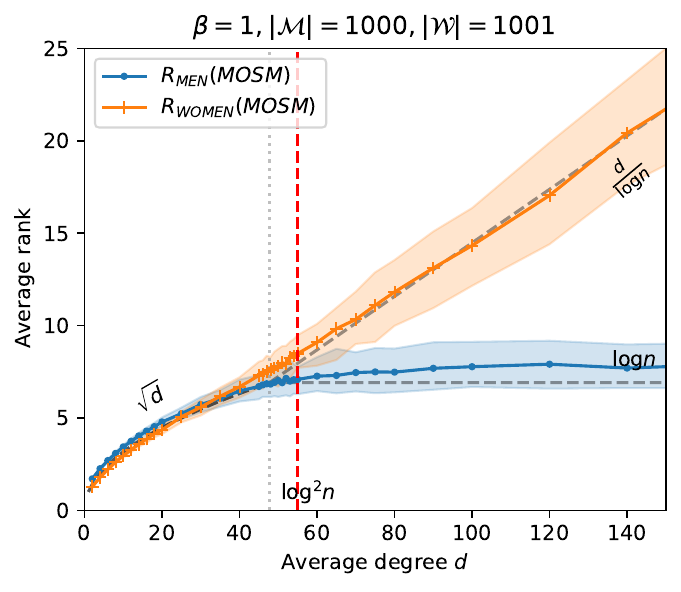}
        \includegraphics[width=1.0\textwidth]{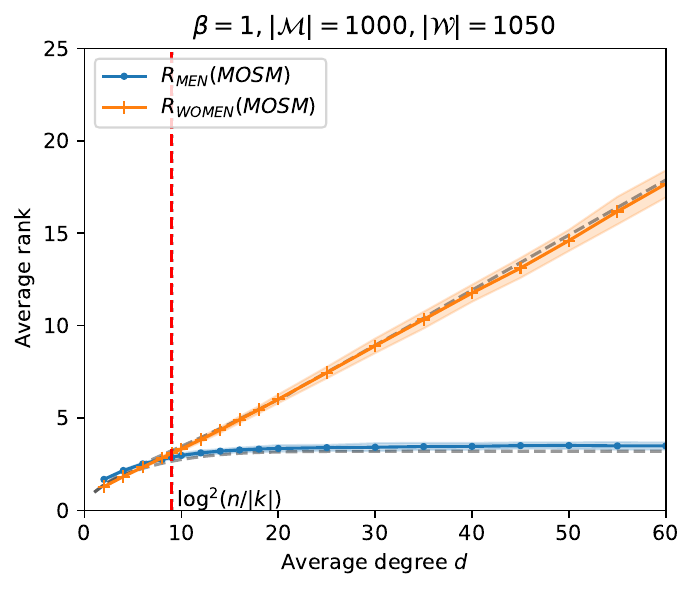}
    \end{minipage}\hfill
    \begin{minipage}{0.48\textwidth}
    \centering
    \includegraphics[width=1.0\textwidth]{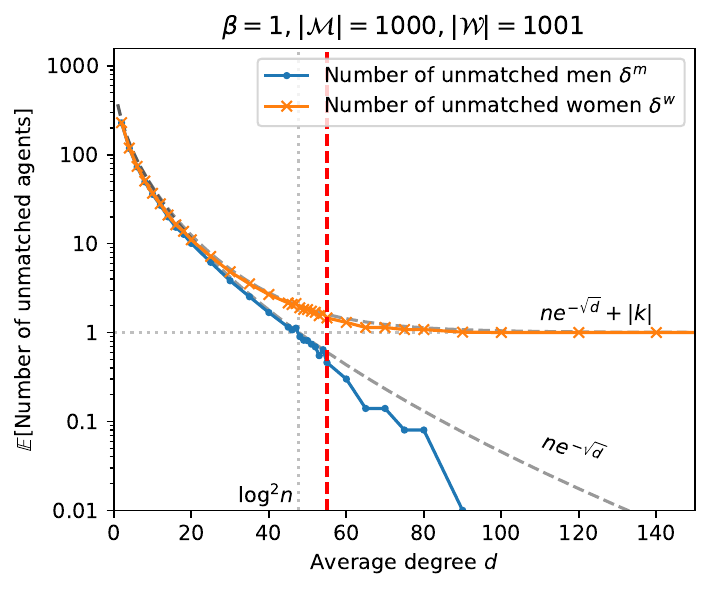}
    \includegraphics[width=1.0\textwidth]{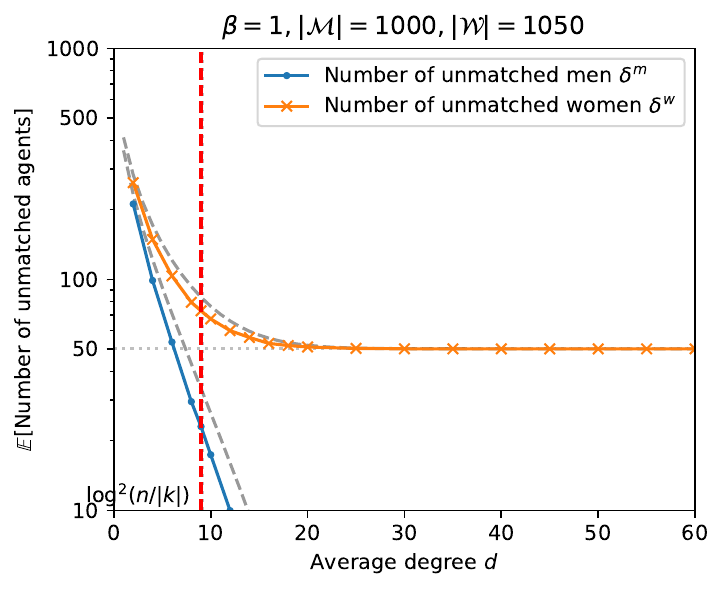}
    \end{minipage}
    \caption{{
    Stable matching in  random matching markets with correlated preferences ($\beta =1, (|\mathcal{M}|,|\mathcal{W}|)\in \{(1000,1001), (1000,1050)\}$). (The threshold predicted by Principle~\ref{prin:condition-for-weak-competition} is almost identical to $\log^2 (n/|k|)$ in the bottom two charts, hence the grey dotted line is not visible.)
    }}
    \label{fig:robust-check-correlated-1}
\end{figure}

\begin{figure}[htbp!]
    \centering
    \begin{minipage}{0.48\textwidth}
        \centering
        \includegraphics[width=1.0\textwidth]{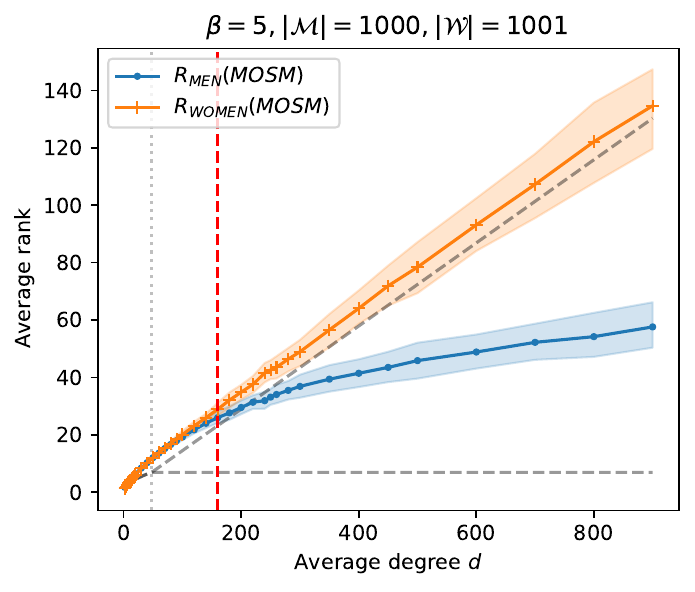}
        \includegraphics[width=1.0\textwidth]{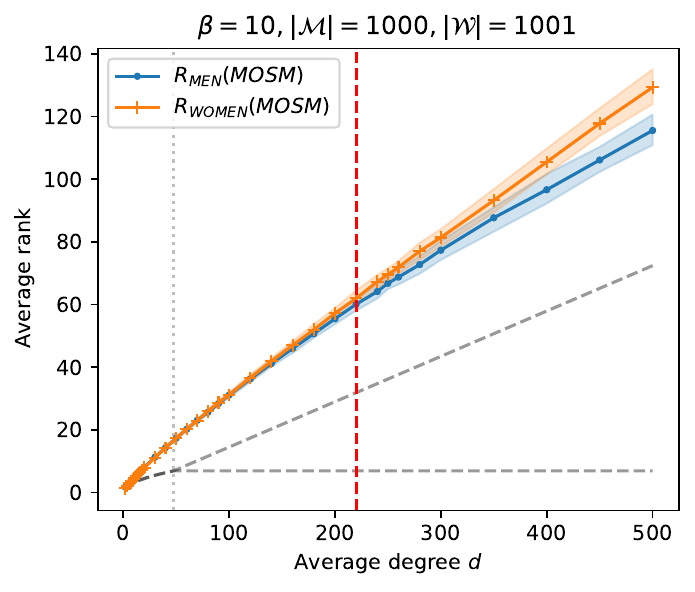}
        \includegraphics[width=1.0\textwidth]{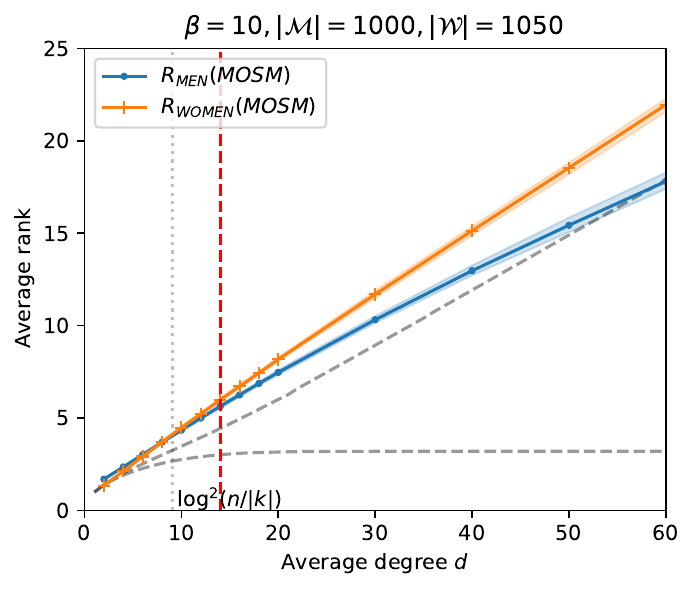}
    \end{minipage}\hfill
    \begin{minipage}{0.48\textwidth}
    \centering
    \includegraphics[width=1.0\textwidth]{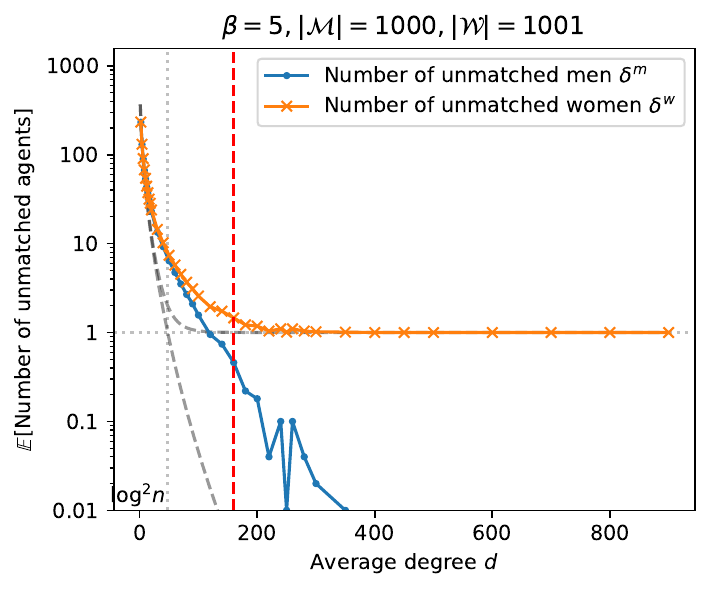}
    \includegraphics[width=1.0\textwidth]{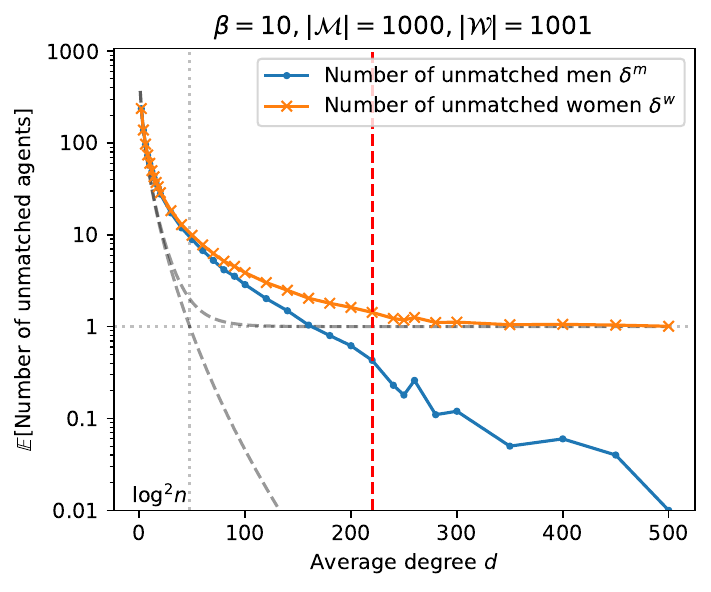}
    \includegraphics[width=1.0\textwidth]{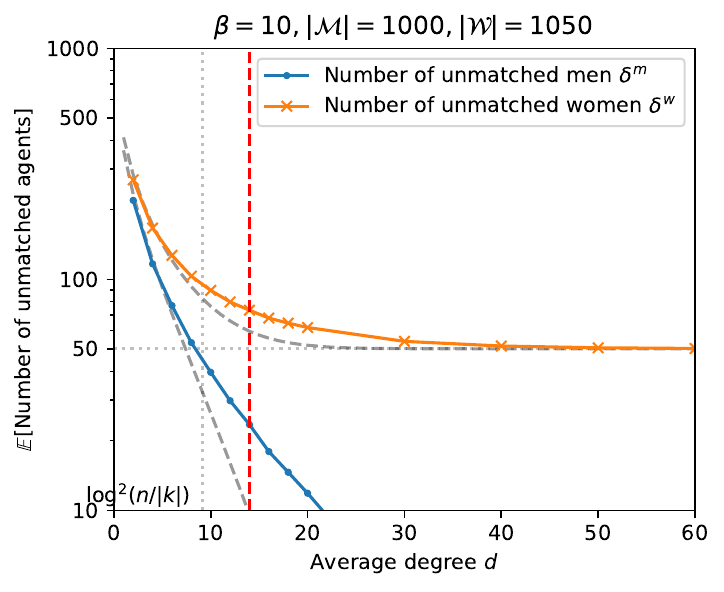}
    \end{minipage}
    \caption{{
    Stable matching in  random matching markets with correlated preferences ($\beta \in \{5,10\}, (|\mathcal{M}|,|\mathcal{W}|)\in \{(1000,1001), (1000,1050)\}$). 
    }}
    \label{fig:robust-check-correlated-2}
\end{figure}
In each figure, the red vertical line represents the threshold degree $d_\delta^*(n,k)$, defined in \eqref{eq:threshold-unmatched-men}, beyond which the number of unmatched short-side agents is smaller than the half of the market imbalance.
Consistently across all settings, we observe that $R_\text{MEN}$ and $R_\text{WOMEN}$ deviate from each other starting from the threshold $d_\delta^*(n,k)$, suggesting that Principle~\ref{prin:condition-for-weak-competition} is robust even to strong correlations in preferences. Our quantitative predictions (shown as gray dashed lines) are found to be surprisingly accurate  under moderate correlations $\beta = 1$. As expected, our quantitative predictions  are inaccurate under strong correlations $\beta \in \{5, 10\}$.

\end{document}